\documentclass[11pt, letterpaper]{article}
\usepackage{setspace}
\usepackage{amsmath,amsthm,dsfont}
\usepackage{caption}
\usepackage{verbatim}
\usepackage{subcaption}
\usepackage{graphicx,psfrag,epsf}
\usepackage{enumerate}
\usepackage{natbib}
\usepackage{bbm}
\usepackage{amssymb}
\usepackage{newtxtext,newtxmath}
\usepackage[affil-it]{authblk}
\usepackage{verbatim}
\usepackage{extarrows}
\usepackage{framed}
\usepackage[colorlinks=true,citecolor=blue]{hyperref}
\usepackage{algorithm}
\usepackage{algorithmic}
\usepackage{url} 
\usepackage{float}
\usepackage{geometry}
\usepackage{lscape} 
\usepackage{multirow}
\usepackage{booktabs}
\usepackage{pifont}
 \usepackage{graphicx}
 \usepackage{pdfpages}
\usepackage{empheq}
\usepackage{thmtools}       
\usepackage[title]{appendix}
\usepackage[normalem]{ulem}
\usepackage{etoc}
\newcommand{\stkout}[1]{\ifmmode\text{\sout{\ensuremath{#1}}}\else\sout{#1}\fi}
\usepackage[compact]{titlesec}

\def\x{\boldsymbol{x}}

\def \tildepi{\tilde{\pi}}

\newcommand{\R}{\ensuremath{\mathbb{R}}}

\newcommand{\E}{\ensuremath{\mathbb{E}}}

\newcommand{\Pisg}{\Pi_{\mathrm{Single}}}

\newcommand{\Piol}{\Pi_{\text{Overlap}}}

\usepackage{tikz}
\newcommand{\cir}[1]{\tikz[baseline]{%
    \node[anchor=base, draw, circle, inner sep=0, minimum width=0.4em]{#1};}}
\usepackage{rotating}

\theoremstyle{definition}
\newtheorem*{theorem*}{Theorem}
\newtheorem{assumption}{Assumption}

\newtheorem{remark}{Remark}
\newtheorem{theorem}{Theorem}
\newtheorem{lemma}{Lemma}

\newtheorem{definition}{Definition}

\newtheorem{proposition}{Proposition}
\newtheorem{Algorithm}{Algorithm}
\allowdisplaybreaks

\begin{document}
\title{\bf Safe Policy Learning under Regression Discontinuity Designs with Multiple Cutoffs\thanks{Imai acknowledges partial support from the Sloan Foundation (Economics Program; 2020-13946).  We thank Tatiana Velasco and her co-authors for sharing the dataset analyzed in this paper.  We have benefited from the Alexander and Diviya Magaro Peer Pre-Review Program. The proposed methodology is implemented via an R package \texttt{rdlearn} available at \url{https://cran.r-project.org/package=rdlearn}}}
\author{Yi Zhang\thanks{Department of Statistics, Harvard University, Cambridge, MA 02138. Email: \href{mailto:yizhang0017@gmail.com}{yizhang0017@gmail.com}
} \quad\quad Eli Ben-Michael\thanks{Assistant Professor, Department of Statistics \& Data Science and Heinz College of Information Systems \& Public Policy, Carnegie Mellon University. 4800 Forbes Avenue, Hamburg Hall, Pittsburgh PA 15213.  Email: \href{mailto:ebenmichael@cmu.edu}{ebenmichael@cmu.edu} URL:
      \href{https://ebenmichael.github.io}{ebenmichael.github.io}} \quad\quad Kosuke Imai\thanks{Edith and Benjamin Geisinger Professor, Department of Government and Department of
      Statistics, Harvard University.  1737 Cambridge Street,
      Institute for Quantitative Social Science, Cambridge MA 02138.
      Email: \href{mailto:imai@harvard.edu}{imai@harvard.edu} URL:
      \href{https://imai.fas.harvard.edu}{https://imai.fas.harvard.edu}}}
\date{November 2026}

\maketitle

\etocdepthtag.toc{mtchapter}
\etocsettagdepth{mtchapter}{subsection}
\etocsettagdepth{mtappendix}{none}

\begin{abstract}
We study policy learning under the regression discontinuity (RD) design, where deterministic treatment assignment requires extrapolation to identify improved treatment cutoffs. We develop a robust optimization approach that leverages the multiple cutoffs common in RD settings to minimize worst-case welfare loss relative to the status quo policy. Under the maintained assumptions, the resulting cutoffs guarantee overall welfare at least as high as that under the existing policy. We also establish the oracle-gap bounds for the learned policy.  Finally, we illustrate the proposed methodology through simulation studies and then apply it to data from an educational assistance program in Colombia.


\end{abstract}

\newpage
\section{Introduction}

The regression discontinuity (RD) design is widely used in causal
inference and program evaluation with observational data \citep[e.g.,][]{thistlethwaite1960regression,lee:more:butl:04,eggers_hainmueller_2009,card2009does,lee2010regression}.
A primary reason for this popularity is that the RD design can provide credible causal inference without assuming that the treatment is unconfounded.
Under the RD design, one can exploit a known \textit{deterministic} treatment assignment mechanism and identify the average treatment effect at the treatment assignment cutoff.
On the whole, the causal inference literature has focused on the development of rigorous estimation and inference methodology for this RD treatment effect \citep{calonico2014robust, calonico2018effect, cheng1997automatic, armstrong2018optimal, hahn2001identification, imbens2008regression, kolesar2018inference,imbens2019optimized}.

While treatment effect estimation at an existing cutoff is useful, policy makers may wish to choose alternative treatment assignment cutoffs \citep{heckman2005structural}. However, since this requires extrapolating the treatment effects away from the cutoff, these treatment effects may only be partially identified. When the partially identified interval for a conditional treatment effect contains zero, its pointwise sign is unknown. Moreover, a cutoff policy assigns treatment jointly over a range of the running variable, so even a point-identified sign rule need not be a feasible threshold policy. The relevant object is therefore the welfare of each possible cutoff policy, including the cumulative gains and losses over the region whose assignment would change.

We develop a safe policy-learning framework for this problem in RD designs with multiple cutoffs, which commonly arise across subpopulations defined by geography, time, or individual characteristics \citep[e.g.,][]{chay2005central,klavsnja2017incumbency, cattaneo2020extrapolating}. Variation in the existing cutoffs reveals cross-group outcome differences over parts of the running-variable support. Under a Lipschitz restriction on those differences, we bound the welfare of candidate vectors of group-specific cutoffs and select the policy with the highest worst-case welfare gain relative to the status quo policy \citep{ben2021safe}. Because the policy class contains the status quo and its value is identified, this population rule will not reduce welfare under the maintained restrictions. We complement this population result with a cross-fitted doubly robust criterion and establish statistical guarantees about the performance of the empirical, data-driven policy relative to both the status quo and the oracle policy in the class. More generally, the identification argument applies to any binary policy that agrees with the status quo outside the range of existing cutoffs, and the statistical guarantees extend to prespecified finite-VC subclasses of this supported class that contain the status quo. We focus on group-specific cutoff policies because they typically correspond directly to the administrative decision.

We apply the proposed methodology to a recent study of the Access with Quality to Higher Education (ACCES) loan program in Colombia \citep{melguizo2016credit}.
In this application, geographical departments in Colombia used different cutoffs, which are based on test scores, to determine the eligibility for financial aid to low-income students.
We use the proposed methodology to learn department-specific cutoffs that increase the estimated worst-case post-secondary enrollment rate. The estimated gains, though small, persist under capacity-neutral (i.e., reallocate eligibility across departments rather than changing the total number of the eligible) and order-preserving restrictions (i.e., adjust department-specific cutoffs without changing their order). A supplementary analysis applies the same criterion to a depth-2 policy-tree class, illustrating its implementation beyond threshold rules.

\paragraph{Related Work.}

Our work connects the literature on extrapolation in RD designs with the literature on treatment choice under ambiguity and policy learning with observational data. Broadly, the literature on extrapolation in RD designs focuses on identifying and estimating treatment effects away from existing cutoffs under additional restrictions \citep[e.g.,][]{mealli2012evaluating,dong2015identifying,angrist2015wanna,cattaneo2020extrapolating,dowd2021donuts,bertanha2020external,bertanha2020regression}. In particular, \cite{cattaneo2020extrapolating} uses a constant-bias assumption to extrapolate treatment effects in a multi-cutoff RD design. \cite{bertanha2020regression} uses many thresholds to identify and estimate average effects under researcher-specified counterfactual distributions, including in sharp and fuzzy designs. We address the complementary but different decision problem of selecting a policy when the welfare under each candidate policy is only partially identified. We construct worst-case welfare bounds and optimize them over a prespecified policy class, specializing to group-specific cutoff policies in the application. Our Lipschitz difference restriction relaxes the constant-bias restriction of \cite{cattaneo2020extrapolating} by allowing the cross-group difference to vary with the running variable.

Treatment choice under partial identification \citep{manski2000identification, manski2005social} has developed into a broad literature on statistical decisions under ambiguity. Closely related recent work studies minimax-regret decisions between two policies \citep{yata2021optimal,montiel2023decision}, aggregation of evidence across populations \citep{ishihara2021evidence}, and policy learning that is robust to distributional shifts \citep{adjaho2022externally,kido2022distributionally}. Like these studies, we directly take into account both ambiguity and statistical uncertainty in decision-making. However, our setting differs in its decision target and data structure. Specifically, we select a deterministic policy from a prespecified class, use the multi-cutoff RD design to bound welfare where deterministic assignment eliminates overlap, and establish statistical safety and regret guarantees for the learned policy. By comparison, \cite{yata2021optimal} derives an exact finite-sample minimax-regret rule for deciding between two prespecified policies, where the decision rule may randomize between them. Thus, the two approaches address different decision problems.

Our estimation problem also builds on existing policy learning methods for observational data, which commonly use inverse-probability or doubly robust scores under overlap \citep[e.g.,][]{swaminathan2015counterfactual,qian2011performance,zhao2012estimating,Kitagawa2018,athey2021policy}.
Additionally, recent work has focused on policy learning in partially identified settings \citep[e.g.][]{zhou2018offline, Kallus2021, pu2021estimating, ben2021safe}.
We build on this literature by combining a doubly robust score for the identified portion of the welfare with partial identification bounds for the counterfactual portion that lacks overlap under the cross-group Lipschitz restriction.

\paragraph{Organization of the paper.}
The paper is organized as follows.
Section~\ref{sec:setup} defines the policy learning problem in RD designs with multiple cutoffs. 
Section~\ref{sec:method:populationPL} presents the proposed methodological framework and our Lipschitz restriction for partial identification.
Section~\ref{sec:multi:empirical} discusses the empirical policy learning problem, where we propose a doubly-robust estimator.
Section~\ref{sec:theory} establishes asymptotic bounds on the welfare loss of our learned policy relative to the baseline policy and the oracle optimal-in-class policy. 
Section~\ref{sec:simu:main} evaluates the finite-sample performance of the method in simulation studies.
In Section~\ref{sec:app}, we apply the proposed methodology to learning new cutoffs for the ACCES loan program in Colombia.
Lastly, Section~\ref{sec:discussion} concludes.

\section{Problem setup and notation}\label{sec:setup}

We begin by describing the policy learning problem under the general multi-cutoff RD design \citep{cattaneo2016interpreting} and introducing the notation used throughout the paper.

\subsection{The multi-cutoff RD design}

In our setting, for each unit $i=1,\ldots,n$, we observe a univariate running variable $X_i \in \mathcal{X} \subset \R$, a binary treatment assignment $W_i \in \{0,1\}$, and an outcome variable $Y_i \in \R$.
We assume that the running variable $X$ has a continuous positive density $f_X$ on the support $\mathcal{X}$. 
Each unit $i$ belongs to one of $Q$ different subgroups, denoted by $G_i \in \mathcal{Q}:=\{1,\ldots,Q\}$, where each group $g$ has a corresponding cutoff value $\tilde{c}_g$. 
These existing cutoff values may differ between groups. 
For simplicity, we assume there are no ties, but note that our proposed methodology can still be applied in situations where tied cutoffs may exist.
We also assume that these cutoffs have been sorted in increasing order (i.e., $\tilde c_1 < \tilde c_2 <\ldots < \tilde{c}_Q$). 
Note that both $X$ and $G$ are pre-treatment variables, which are neither affected by the cutoffs nor the treatment variable.

Throughout the paper, we consider the  \textit{sharp} RD setting, where given a set of known cutoff values $\{\tilde{c}_{g}\}_{g=1}^Q$, the treatment assignment $W_i$ is completely determined by \emph{both} the running variable $X_i$ and the group indicator $G_i$.
Specifically, unit $i$ in group $g$ is assigned to the treatment condition when the value of the running variable $X_i$ exceeds a fixed, known treatment cutoff $\tilde{c}_g$, i.e., $$W_i = \sum_{g=1}^{Q} \mathds{1}(G_i =g)\mathds{1}(X_i \geq \tilde c_g)$$ where $\mathds{1}\left(\cdot\right)$ denotes the indicator function. 
Following the potential outcome framework \citep{neyman1923, rubin1974}, we define $Y_i(1)$ and $Y_i(0)$ as the potential outcomes under the treatment and control conditions, respectively.
The observed outcome is then given by $Y_{i}=W_{i} Y_{i}(1)+\left(1-W_{i}\right) Y_{i}(0)$. 
We further assume that the tuples $\{X_i, W_i,G_i, Y_i(1),Y_i(0)\}$ are drawn i.i.d from an overall target population $\mathcal{P}$.
For notational convenience, we will sometimes drop the unit subscript $i$ from random variables. 

Much of the RD design literature has focused on the evaluation of the conditional average treatment effect (CATE) at the cutoff.
Under the multi-cutoff RD design, we define the group-specific CATE for each group $g$ at the cutoff value $\tilde{c}_g$, 
\begin{equation*}\label{equ:LATE}
\tau_g(\tilde{c}_g) \ := \ \mathbb{E}[Y_{i}(1)-Y_{i}(0) \mid X_{i}=\tilde{c}_g, G_{i}=g].
\end{equation*}
The key identification assumption for $\tau_g(\tilde{c}_g)$ under the sharp RD design is the local continuity of the conditional expectation function of the potential outcome \citep{hahn2001identification}.
\begin{assumption}[Continuity] \label{ass:multi:continuity} 
Let $m(w,x,g) =\mathbb{E}[Y_i(w) \mid X_i = x, G_i=g]$ denote the group-specific conditional potential outcome function. Then, $m(w,x,g)$ is continuous in $x$ at $\tilde{c}_g$ for $g=1,\ldots,Q$ and $w=0,1$.
\end{assumption}
\noindent Under this assumption, we can identify  $\tau_g(\tilde{c}_g)$ via,
$$\tau_g(\tilde{c}_g) \ = \ \lim _{x \downarrow \tilde{c}_g} \mathbb{E}[Y_i \mid X_i=x,G_i=g]-\lim _{x \uparrow \tilde{c}_g} \mathbb{E}[Y_i \mid X_i=x,G_i=g].$$ 

Under Assumption~\ref{ass:multi:continuity} alone, however, the deterministic treatment assignment makes it impossible to identify the average treatment effect beyond the existing cutoff \citep[e.g.,][]{hahn2001identification,imbens2008regression}.
In the literature, researchers have proposed additional assumptions to extrapolate beyond the existing cutoffs.
For example, under the single-cutoff RD setting, \cite{dowd2021donuts} assume that the partial derivative of the conditional expectation function with respect to the running variable is bounded.
Under the fuzzy RD setting, \cite{bertanha2020regression} restricts heterogeneity across compliance classes.
Lastly, under the multi-cutoff RD design, \cite{cattaneo2020extrapolating} proposes a ``common-trend'' assumption. 
Building on these recent advancements, we show how to learn new and improved cutoffs from observed data in multi-cutoff RD designs.

\subsection{Policy learning under the multi-cutoff RD design}
\label{sec:notation}

We introduce our policy learning problem under the multi-cutoff RD design.
Throughout, we assume that potential outcomes depend on cutoffs only through the treatment assignment, which is implicit in our notation \citep{dong2015identifying,bertanha2020regression}.
In other words, changing the cutoff value does not alter outcomes unless it leads to a change in the treatment condition.
The assumption is violated, for example, if individuals change their behavior in response to a change in the cutoff value under the same treatment condition.

Under the multi-cutoff RD design, a policy $\pi$ is a deterministic function that maps the running variable and group indicator to a binary treatment assignment variable, $\pi:\mathcal{X} \times \mathcal{Q} \rightarrow \{0,1\}$.
Let $u(y,w)$ denote the utility for outcome $y$ under treatment $w$. 
Thus, the utility for unit $i$ with treatment $w$ and potential outcome $Y_i(w)$ is given by $u(Y_i(w),w)$.
Our goal is to find a new policy $\pi$ that has a high average value (expected utility or welfare) for the overall population.
For ease of exposition, we use the observed outcome as the utility, i.e., $u(y, w)=y$, but our discussion below readily extends to other utility functions.

For any given policy $\pi$, we define its overall value as the expected utility in the target population $\mathcal{P}$: 
\begin{equation}\label{equ:expUtility}
  V(\pi)=\mathbbm{E}\left[ Y\left(\pi(X,G)\right)\right],  
\end{equation}
where the expectation is taken with respect to the joint distribution of the running variable $X$, the group indicator $G$, and potential outcomes $Y(w)$. 
The goal of policy learning is to find the best possible policy within a pre-specified policy class $\Pi$, i.e.,
$$\pi^{*} \in \ \underset{\pi \in  \Pi}{\operatorname{argmax}}\  V(\pi).$$

Subject to the support restriction stated below, our identification and policy-learning framework applies to binary policies more generally and does not require policies to assign treatment only by thresholding the running variable. We state the population and statistical results for a prespecified supported policy class. In the empirical application, we specialize to group-specific threshold policies because the most direct way for policymakers to adjust an RD assignment rule is to change the eligibility cutoffs $c_g$ while retaining the existing monotone rule $\pi(x,g)=\mathds{1}(x\geq c_g)$.

Learning an optimal policy requires the evaluation of $V(\pi)$ for any given policy $\pi \in \Pi$.
Under RD designs, however, the average potential outcomes under cutoffs that are different from the existing ones are not identifiable without further assumptions.
Since the treatment assignment is \textit{deterministic} rather than \textit{stochastic}, the usual overlap assumption, i.e., $0 < \Pr(W_i = 1 \mid X_i = x) < 1 \ \text{for all } x \in \mathcal{X}$, is violated in RD designs.
As a result, we cannot use  standard policy learning methodology based on the inverse probability-of-treatment weighting (IPW) or its variants \citep[see, e.g.][]{zhao2012estimating, Kitagawa2018,athey2021policy}. 
Therefore, to evaluate $V(\pi)$, we must extrapolate beyond the existing threshold $\tilde{c}_g$ and infer the average unobserved potential outcomes.

\section{Safe policy learning through robust optimization}\label{sec:method:populationPL}

We now outline our basic framework of policy learning under the multi-cutoff RD designs.
The proposed framework has three steps.
We first decompose the expected utility for any given policy into identifiable and unidentifiable components. 
Second, we partially identify the unidentifiable component using a relaxed version of the ``parallel-trends'' assumption from the difference-in-differences literature \citep{manski2018right}.
Lastly, we propose a robust optimization approach that maximizes the worst-case expected utility.
We show that the resulting maximin value optimal policy has a statistical \textit{safety} guarantee; the value of deploying the learned policy is no worse than the baseline policy under the assumption used for partial identification, although it may be suboptimal relative to the infeasible policy that maximizes the true expected utility within the candidate class.

\subsection{Policy value decomposition}\label{sec:method:decom}

We begin by defining the baseline policy that generates the observed data as
$$\tilde\pi(x,g) \ := \ \mathds{1}(x\geq \tilde c_g).$$
For the safety comparison below, we require the candidate policy class to contain this baseline policy.
By comparing the deviation of a given policy $\pi$ from the baseline policy $\tilde\pi$, we decompose the expected utility in Equation~\eqref{equ:expUtility} into point-identifiable and non-identifiable terms, denoted by $\mathcal{I}_{\text{iden}}$ and $\mathcal{I}_{\text{uniden}}$, respectively:
\begin{equation}
\begin{aligned}\label{equ:general:multi:V}
V(\pi,m)&=\underbrace{ \mathbb{E}\left[\pi(X,G)\tilde\pi(X,G)Y + \left(1-\pi(X,G)\right)\left(1-\tilde\pi(X,G)\right)Y\right]}_{\mathcal{I}_{\text{iden}}}\\
&+\underbrace{ \mathbb{E}\left[  \pi(X,G)\left(1-\tilde\pi(X,G)\right) m(1,X,G)\right]+\mathbb{E}\left[  \left(1-\pi(X,G)\right)\tilde\pi(X,G) m(0,X,G)\right]}_{\mathcal{I}_{\text{uniden}}},\\
\end{aligned}
\end{equation}
where $m(w,x,g)$ is the group-specific conditional potential outcome function.

The first term in Equation~\eqref{equ:general:multi:V} corresponds to the cases where the policy $\pi$ agrees with the baseline policy $\tilde\pi$.
This component can be point-identified from the observable data by replacing the counterfactual outcome $Y(\pi(X,G))$ with the observed outcome $Y$.
The second and third terms, however, are not point-identifiable without further assumptions; the counterfactual outcome cannot be observed when the policy under consideration $\pi$ disagrees with the baseline policy $\tilde\pi$.

The key identifying restriction on the policy class concerns where a candidate policy may change treatment assignment relative to the status quo. We define the class of policies supported by the multi-cutoff design as
\begin{equation}\label{equ:Pi:supported}
\Pi_{\mathrm{S}}=\left\{\pi:\mathcal X\times\mathcal Q\rightarrow\{0,1\}\mid
\pi(x,g)=\widetilde\pi(x,g)\ \text{for }x<\widetilde c_1\ \text{or }x\geq\widetilde c_Q\right\}.
\end{equation}
Below the lowest cutoff, no treated conditional mean is observed for any group.  Similarly, at or above the highest cutoff, no control conditional mean is observed. Under the cross-group support condition stated in Assumption~\ref{ass:overlap:general} below, restricting the policy class to $\Pi_\mathrm{S}$ therefore ensures that an observed reference group is available whenever a candidate policy differs from the baseline. For $x<\widetilde c_1$ or $x\geq\widetilde c_Q$, the disagreement terms $\pi(x,g)\{1-\widetilde\pi(x,g)\}$ and $\{1-\pi(x,g)\}\widetilde\pi(x,g)$ are both zero.

Let $\Pi$ represent a prespecified class of measurable binary policies satisfying $\widetilde\pi\in\Pi\subseteq\Pi_{\mathrm S}$. The restriction $\Pi\subseteq\Pi_{\mathrm S}$ permits evaluation of every policy in the class, while inclusion of the baseline policy provides the benchmark used in the safety result below. We assume that the displayed maxima over $\Pi$ are attained.

For the empirical RD policy problem studied here, we specialize to the following group-specific threshold class:
\begin{equation}
    \begin{aligned}\label{equ:Pi:general}
        \Piol=\left\{ \pi(x,g) = \mathds{1}(x\geq c_g) \mid \ c_g \in [\tilde c_1,\tilde c_Q]\right\}.
    \end{aligned}
\end{equation}
Thus, $\Piol\subset\Pi_{\mathrm S}$ and $\widetilde\pi\in\Piol$. This threshold class preserves the program's existing monotone eligibility rule and makes each policy directly interpretable as a vector of revised administrative cutoffs. This type of institutional restriction on feasible treatment rules is common in policy learning problems \citep{Kitagawa2018}. The identification and statistical results below are stated for the general class $\Pi$ while setting $\Pi=\Piol$ gives the specialization used in the application.
 
Under the support restriction in Equation~\eqref{equ:Pi:supported}, the unidentifiable term $\mathcal{I}_{\text{uniden}}$ can be further decomposed into the conditional mean function of the observed outcome, which is point-identifiable, and an unidentifiable ``difference function,''
\begin{align}
    \tilde m(x,g) \ & := \ m(\tilde \pi (x,g), x,g) \label{equ:tildem(x,g)} = \E[Y \mid X = x, G = g] \\
d(w,x,g,g^\prime) \ & := \ m(w,x,g)- m(w,x,g^\prime). \label{equ:differencefn}
\end{align}
The difference function $d(w,x,g,g^\prime)$ represents the difference in the conditional mean function of potential outcomes between two groups at the same value of the running variable.

Note that if $d(w,x,g,g^\prime)=0$, then the conditional average treatment effect is identical between groups $g$ and $g^\prime$ given the running variable. 
This type of assumption is often invoked to combine  treatment effect estimates from multiple studies \citep{stuart2011use,dahabreh2019generalizing}, but it is typically not credible here because the running variable is the only conditioning variable.
In Section~\ref{sec:robustopt}, we instead partially identify the difference function by imposing a weaker Lipschitz restriction on its variation with the running variable.

We now formally state our decomposition result, showing that we can use the group whose existing cutoff is closest to the counterfactual cutoff of another group for extrapolation. 
\begin{proposition}[Decomposition of the unidentifiable component]\label{prop:unident}
Under the cross-group support condition in Assumption~\ref{ass:overlap:general}, for any policy $\pi \in \Pi_{\mathrm S}$, the unidentifiable term $\mathcal{I}_{\text{uniden}}$ in Equation~\eqref{equ:general:multi:V} can be further decomposed as follows:
\begin{equation}
    \mathcal{I}_{\text{uniden}}(\pi,m) \ = \ \Xi_1(\pi) + \Xi_2(d,\pi), \label{equ:general:multi:uniden}
    \end{equation}
where $\Xi_1(\pi)$ is identifiable but $\Xi_2(d,\pi)$ is not necessarily identifiable,
\begin{align*} 
 \Xi_1(\pi) \ := \ &  \mathbbm{E}\left[\sum_{g=1}^{Q} \mathds{1}\left(G=g\right)\cdot \pi(X,g)\left(1-\tilde\pi(X,g)\right) \sum_{g^\prime=1}^{g-1} \mathds{1}\left(X\in[\tilde c_{g^\prime}, \tilde c_{g^\prime+1}]\right)  \tilde m(X,g^{\prime})\right]\nonumber\\
    & + \mathbbm{E}\left[\sum_{g=1}^{Q} \mathds{1}\left(G=g\right)\cdot \left(1-\pi(X,g)\right)\tilde\pi(X,g) \sum_{g^\prime=g+1}^{Q} \mathds{1}\left(X\in[\tilde c_{g^\prime-1}, \tilde c_{g^\prime}]\right)  \tilde m(X,g^{\prime})\right] \\
 \Xi_2(d,\pi) \ := \ &  \mathbbm{E}\left[\sum_{g=1}^{Q} \mathds{1}\left(G=g\right)\cdot \pi(X,g)\left(1-\tilde\pi(X,g)\right) \sum_{g^\prime=1}^{g-1} \mathds{1}\left(X\in[\tilde c_{g^\prime}, \tilde c_{g^\prime+1}]\right) d(1,X,g,g^{\prime}) \right]\nonumber\\
    &+ \mathbbm{E}\left[\sum_{g=1}^{Q} \mathds{1}\left(G=g\right)\cdot \left(1-\pi(X,g)\right)\tilde\pi(X,g) \sum_{g^\prime=g+1}^{Q} \mathds{1}\left(X\in[\tilde c_{g^\prime-1}, \tilde c_{g^\prime}]\right)  d(0,X,g,g^{\prime}) \right]
 \end{align*}
\end{proposition}

\begin{figure}[t!]
    \centering
    \includegraphics[width=1\linewidth]{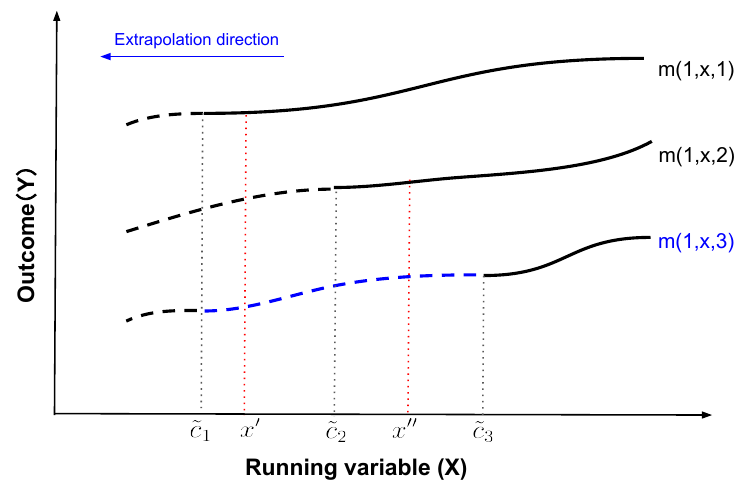}
    \caption{Illustration of the policy value decomposition in Proposition~\ref{prop:unident} in a three-group case.
    The black lines represent the group-specific conditional expectation functions under the treatment condition $m(1,x,\cdot)$, where the solid line segments can be identified by the observed conditional mean outcome functions $\tilde m(x,\cdot)$ whereas the black dashed segments are not identifiable.
   The blue dashed line indicates an example extrapolation region where we consider the counterfactual policy value under an alternative cutoff lower than the existing cutoff for Group~3.
   To infer the counterfactual policy value at $x^{\prime\prime}$ for this group, we decompose it into the sum of the observed conditional mean value for Group~2 and the unidentifiable difference between Groups~2~and~3 at $\tilde{c}_3$.
   Similarly, to infer the counterfactual policy value at $x^\prime$ for Group~3, we leverage the observed difference between Groups~1~and~3. }
    \label{fig:m1}
\end{figure}

Figure~\ref{fig:m1} illustrates the decomposition using a three-group case and focusing on the outcome under the treatment.
Suppose that we wish to learn a new cutoff for Group~3 whose existing cutoff is $\tilde{c}_3$.
To do this, we must extrapolate the conditional mean outcome $m(1,x,3)$ for this group in the region below the existing cutoff (represented by the dashed blue line). 
To evaluate the counterfactual value of the conditional mean outcome at $x^{\prime\prime} \in [\tilde{c}_2, \tilde{c}_3]$ for Group~3, we choose Group~2 whose existing cutoff is the closest to $\tilde{c}_3$.
We then decompose the counterfactual value $m(1,x^{\prime\prime},3)$ into the identifiable conditional mean for Group~2, i.e., $m(1,x^{\prime\prime},2)$, and the difference in the conditional mean function between Groups~3~and~2, i.e., $d(1,x^{\prime\prime},3,2)=m(1,x^{\prime\prime},3)-m(1,x^{\prime\prime},2)$.
Similarly, to infer $m(1, x^{\prime},3)$, we use Group~1 as a reference group and decompose this counterfactual value into the identifiable conditional mean value for Group~1, i.e., $m(1, x^\prime, 1)$, and the difference between Groups~3~and~1, i.e., $d(1,x^\prime,3,1)$.

\begin{remark}\label{remark:reference}
The decomposition in Proposition~\ref{prop:unident} does not require groups with nearby cutoffs to have similar potential outcome functions. For any group \(g'\) whose outcome under treatment \(w\) is observed at \(x\), we have \(m(w,x,g)=\widetilde m(x,g')+d(w,x,g,g')\). Thus, Proposition~\ref{prop:unident} gives one admissible decomposition, but there are several different reference groups that could be used in the decompoisition. Choosing a different reference group could result in a different bound on the missing conditional mean. In Section~\ref{sec:robustopt}, we show how to combine restrictions that arise from the different possible reference groups and decompositions to obtain sharp bounds. 
\end{remark}

In sum, Proposition~\ref{prop:unident} implies the following decomposition of the expected utility in Equation~\eqref{equ:general:multi:V}:
\begin{equation}\label{equ:gerneral:multi:Vdecom}
 V(\pi,d) \ = \ \mathcal{I}_{\text{iden}}(\pi) +\Xi_1(\pi)+\Xi_2(d,\pi).
\end{equation}
For any given policy $\pi\in\Pi_{\mathrm S}$, when the baseline policy $\tilde\pi$ and the candidate policy $\pi$ agree, the expected utility equals the identifiable component $\mathcal{I}_{\text{iden}}$.
When they disagree, the expected utility includes the sum of two additional components: the identifiable term $\Xi_{1}(\pi)$, and the partially identifiable term $\Xi_{2}(d,\pi)$.

\subsection{Partial identification and robust optimization}\label{sec:robustopt}

Using the decomposition above, we propose a robust optimization approach to policy learning.
Specifically, we first partially identify the difference function $\Xi_2$ and then use robust optimization to derive an optimal policy by maximizing the worst case value.

As illustrated in Figure~\ref{fig:m1} and discussed above, the cross-group difference in the conditional outcome needed for extrapolation is not identifiable.
For example, at $x^{\prime\prime}$, this difference function between Groups~3~and~2 cannot be point-identified.
Yet, the same difference function is identifiable in the region above the existing cutoff of Group~3, i.e., $x \ge \tilde{c}_3$.
We leverage this fact for extrapolation under the following restriction on how rapidly the cross-group difference may change with the running variable.
\begin{assumption}[Lipschitz Difference Restriction] \label{ass:general:lipdiff}
The difference function $d(w,x,g,g^\prime)$ lies in the following function class,
$$\mathcal{F}=\left\{f: \left|f(w, x,g,g^\prime)-f(w, \tilde c_{g^\ast},g,g^\prime)\right|\leq \lambda_{wgg^\prime} \left|x-\tilde c_{g^\ast}\right| \quad \forall x\in (\tilde c_{g\wedge g^\prime},\tilde c_{g\vee g^\prime}), w\in \{0,1\} \right\},$$
where ${g^\ast}=w(g\vee g^\prime)+(1-w)(g\wedge g^\prime)$ is a group index, and $\lambda_{wgg^\prime} < \infty$ is a positive constant for $w\in \{0,1\}$ and $g,g^\prime \in \mathcal{Q}$. 
\end{assumption}
In the function class $\mathcal{F}$, the group index ${g^\ast}$ depends on the treatment condition.
For example, in the case of the conditional outcome model under the treatment, ${g^\ast}$ represents the group with the higher cutoff, whereas under the control condition it corresponds to the group with the lower cutoff.

Thus, $\tilde c_{g^\ast}$ is the boundary between the region where both groups receive treatment condition $w$ and the interval over which we extrapolate.  
Assumption~\ref{ass:general:lipdiff} bounds the change in the cross-group difference relative to its value at this boundary, thereby limiting the strength of the interaction between the running variable and group membership in the conditional expected potential outcome over the extrapolation interval. Note that Assumption~\ref{ass:general:lipdiff} by itself does not require that the difference function be differentiable, and that the Lipschitz constants $\lambda_{wgg'}$ are defined for unordered group pairs, so that $\lambda_{wgg'}=\lambda_{wg'g}$.

Assumption~\ref{ass:general:lipdiff} generalizes and hence weakens the assumption of \cite{cattaneo2020extrapolating} which considers an important special case with two groups.
Specifically, they assume that the cross-group heterogeneity does not depend on the running variable at all, i.e., $\lambda_{w12}=0$ implying $d(w,x,1,2)$ is constant in $x$.
Their assumption is analogous to the parallel trend assumption under the difference-in-differences design.

Let $\mathcal{X}_{gw}=\{x\in\mathcal{X}:\tilde\pi(x,g)=w\}$ denote the region where group~$g$ receives treatment condition~$w$ under the baseline policy. Within $\mathcal{X}_{gw}\cap\mathcal{X}_{g'w}$, the difference function is identified as $\tilde d(x,g,g')=\tilde m(x,g)-\tilde m(x,g')$. We use its one-sided limit to define the boundary value that anchors the ambiguity set below:
\begin{equation}\label{equ:genmulti:anchor}
 d(w,\tilde c_{g^\ast},g,g')=
 \lim_{\substack{x'\in\mathcal{X}_{gw}\cap\mathcal{X}_{g'w}\\x'\to\tilde c_{g^\ast}}}
 \tilde d(x',g,g').
\end{equation}
We assume that this limit exists and is identified. For example, this condition holds if both potential outcome functions are continuous at the boundary and the common observed region has support approaching it. 

Under Assumption~\ref{ass:general:lipdiff} and the boundary condition above, extrapolation from the boundary gives the following class of difference functions,

\begin{equation}\label{equ:genmulti:Md}
    \mathcal{M}_d=\left\{ d: \{0,1\}  \times \mathcal{X} \times \mathcal{Q}\times \mathcal{Q} \rightarrow \mathbb{R}\mid  B_{\ell}(w,x,g,g^\prime)\leq  d(w,x,g,g^\prime) \leq  B_{u}(w,x,g,g^\prime) \right\},
\end{equation}
with the point-wise upper and lower bounds given by
\begin{equation}
     \begin{aligned}\label{equ:genmulti:bound}
    B_{\ell}(w,x,g,g^\prime)& = d(w,\tilde c_{g^\ast},g,g^\prime)
    - \lambda_{wgg^\prime}|x-\tilde c_{g^\ast}|,  \\
B_{u}(w,x,g,g^\prime) & = d(w,\tilde c_{g^\ast},g,g^\prime)
    + \lambda_{wgg^\prime}|x-\tilde c_{g^\ast}|.
     \end{aligned}
\end{equation}

Finally, we use robust optimization to find an optimal policy.
We do so by maximizing the worst case value of the expected utility $V(\pi,d)$ in Equation~\eqref{equ:gerneral:multi:Vdecom} over the policy class $\Pi$, with the difference function taking values in the ambiguity set $\mathcal{M}_d$. Because the expected utility for the baseline policy is identifiable, this is equivalent to minimizing the maximum welfare loss relative to the baseline policy:
\begin{equation}
    \begin{aligned}\label{equ:general:Vinf}
       \pi^{\inf} \in  \underset{\pi \in \Pi}{\operatorname{argmax}} \  \min_{d \in \mathcal{M}_d}\  V(\pi,d) \Longleftrightarrow  \pi^{\inf} \in  \underset{\pi \in \Pi}{\operatorname{argmin}} \  \max_{d \in \mathcal{M}_d}\left\{V(\tilde{\pi}, d) - V(\pi,d)\right\}.
    \end{aligned}
\end{equation}
Below we will focus on the maximin formulation of the problem.
As \cite{ben2021safe} shows, the resulting policy $\pi^{\inf}$ has a \textit{safety} guarantee that the overall expected utility of the learned policy is no less than that of the baseline policy.

Note that this maximin criterion does not necessarily reproduce the status quo. Because the baseline policy belongs to $\Pi$ and its value is identified, it provides a feasible benchmark whose worst-case value equals its true value. A different policy can be selected only when its lower welfare bound is at least as large as this benchmark. This can occur when the gains supported by the observed components of the policy value are large enough to dominate the most adverse extrapolation allowed by Assumption~\ref{ass:general:lipdiff}. In the threshold-policy example in Figure~\ref{fig:m1}, lowering the cutoff for Group~3 expands treatment into the blue extrapolation region. A smaller value of $\lambda_{wgg'}$, or a shorter extrapolation distance, narrows the uncertainty about this change and makes such a departure from the status quo more likely, whereas a wider bound makes the learned policy more conservative.

\begin{remark}[All-reference bound]\label{remark:allreference}
The reference group used in Proposition~\ref{prop:unident} provides one valid interval for a missing conditional mean. These intervals can also be combined across all admissible references. Let $\mathcal R_w(x)=\{g'\in\mathcal Q:\tilde\pi(x,g')=w\}$ denote the groups whose outcome under treatment condition $w$ is observed at $x$. Then we can define the all-reference bounds as the intersection of the bounds implied by each possible reference group $g' \in \mathcal{R}_w(x)$:
\begin{equation}\label{equ:genmulti:allreference}
I_{\mathrm{AR}}(w,x,g)
=\bigcap_{g'\in\mathcal R_w(x)}
\left[
\tilde m(x,g')+B_\ell(w,x,g,g'),\,
\tilde m(x,g')+B_u(w,x,g,g')
\right].
\end{equation}
When the outcome range is known, this interval can be further intersected with the corresponding range for its conditional mean. 

Under Assumptions~\ref{ass:multi:continuity} and~\ref{ass:general:lipdiff} and the boundary condition in Equation~\eqref{equ:genmulti:anchor}, Equation~\eqref{equ:genmulti:allreference} is the sharp pointwise bound under the maintained model, provided the observed data and restrictions are compatible. To see this, fix a treatment condition and an interior score. Every active pairwise restriction then relates one missing conditional mean to one observed reference; if two conditional means are both missing, their cutoffs lie on the same side of the score, so Assumption~\ref{ass:general:lipdiff} imposes no restriction on their difference there. The restrictions are therefore separable across missing conditional means, and each feasible interval is exactly the intersection in Equation~\eqref{equ:genmulti:allreference}. An additional group either provides a reference already included in this intersection or provides no active restriction at that score, so further iteration over pairwise identities cannot tighten the interval. The boundary limit in Equation~\eqref{equ:genmulti:anchor} covers cutoff points. For a fixed policy $\pi$, the sharp lower policy value is obtained by using the lower endpoint of $I_{\mathrm{AR}}(\pi(x,g),x,g)$ whenever the assigned conditional mean is missing, retaining the observed conditional mean otherwise, and then averaging over $(X,G)$.

Although the all-reference interval is weakly narrower than the interval based on a single selected reference group, for practical simplicity, we retain the reference-group construction in Proposition~\ref{prop:unident}, together with the bounds in Equation~\eqref{equ:genmulti:bound}, as our main policy-learning procedure. The resulting policy optimizes the selected-reference criterion, which need not select the same policy as maximizing the sharp all-reference lower value. This simplification may result in a more conservative learned policy, but it does not affect the validity of the bound or the safety guarantee established above.

Appendix~\ref{append:sec:realplots} holds the policy learned by our main procedure fixed and reports a supplementary one-sided plug-in comparison. We compare the estimated lower bound on its value using the selected reference with that obtained using all admissible references. Including all admissible references raises the bound, but does not change the conclusion relative to the status quo.
Optimizing the all-reference criterion with the doubly robust learning procedure below would require additional analysis because the binding reference may vary with the score and the pointwise maximum must account jointly for all reference-specific nuisance and boundary estimates. We leave this extension to future work.

\end{remark}

\begin{remark}
    Our proposed methodological framework can be extended to the RD setting with a single cutoff. Instead of relying on Assumption~\ref{ass:general:lipdiff}, one can impose a slightly stronger smoothness assumption on the conditional mean potential outcome functions themselves \citep{dowd2021donuts}. 
    A detailed discussion of this approach is provided in Appendix~\ref{appendix:singlecutoff}.
\end{remark}

\section{Learning policies from data}\label{sec:multi:empirical}

In the previous section, we define the population policy $\pi^{\inf}$ that maximizes worst-case expected utility over policies $\pi\in\Pi$.
Next, we show how to construct this maximin policy using observed data.
Recalling our decomposition in Proposition~\ref{prop:unident}, the  optimization problem of Equation~\eqref{equ:general:Vinf} can be decomposed into the sum of three components,
        $$\pi^{\inf} \in \ \underset{\pi \in \Pi}{\operatorname{argmax}}\left\{\mathcal{I}_{\text{iden}}(\pi)+\Xi_{\text{1}}(\pi)+\underset{d \in \mathcal{M}_d}{\operatorname{min}} \Xi_{2}(d,\pi)\right\}$$
where the first two components, $\mathcal{I}_{\text{iden}}(\pi)$ and $\Xi_{\text{1}}(\pi)$, are point-identifiable, while the last term $\Xi_{2}(d,\pi)$ is not.
We introduce our proposed estimator by examining each of these three terms in turn.

\subsection{Doubly robust estimation of the identifiable component}

Since the first term $\mathcal{I}_{\text{iden}}$ in Equation~\eqref{equ:general:multi:V} is identifiable, we use its sample analog for unbiased estimation:
\begin{equation}\label{equ:multi:emp1}
\widehat{\mathcal{I}}_{\text{iden}}(\pi) = \frac{1}{n} \sum_{i=1}^{n} \left[\pi(X_i,G_i)\tilde\pi(X_i,G_i)Y_i+(1-\pi(X_i,G_i))(1-\tilde\pi(X_i,G_i)) Y_i\right].
\end{equation}

The second term $\Xi_{1}$ in Equation~\eqref{equ:general:multi:uniden} is also nonparametrically identified under an overlap condition, but this term contains an unknown nuisance component, corresponding to the conditional outcome regression for each group $g$, $\tilde{m}(X, g)$.
Specifically, we make the following overlap assumption
\begin{assumption}[Overlap between groups based on the running variable]\label{ass:overlap:general}
Define $e_g(x):=\mathbb{P}\left[G_i=g\mid X_i=x\right]$ as the conditional probability of group membership $g$ given the value of the running variable $X_i = x$. There exists some $\eta >0$ such that  $\eta \leq e_g(x)\leq 1 - \eta$ for all $x\in[\tilde c_1,\tilde c_Q]$ and $g\in \mathcal{Q}$. 
\end{assumption}

Assumption~\ref{ass:overlap:general} concerns common support of the running variable across groups, as opposed to overlap in the deterministic treatment assignment.
Identification of $\Xi_1(\pi)$ requires only that the reference-group probability $e_{g'}(x)$ be positive almost everywhere on each comparison interval where it is used. We maintain the stronger uniform overlap condition for estimation and statistical guarantees below.
This overlap condition also clarifies that treatment assignment for different units depends on their group status, which is determined before treatment assignment. This differs from a setting where units are all subject to a common collection of cutoffs and groups are subsequently constructed from ranges of the running variable.
The uniform condition above can be weakened to only hold for the particular focal--reference comparisons that appear for a given restricted policy class, provided that the reference-group probability is bounded away from zero on each comparison interval and both group probabilities are bounded away from zero near the corresponding boundary. In that case, unsupported comparisons must be omitted from the policy class and estimator.

For a fixed policy, an orthogonal doubly robust representation of $\Xi_1$ yields the following equivalent expression,
\begin{equation}
\begin{aligned}\label{equ:general:multi:db}
  \Xi_{\text{DR}}(\pi) \ := \ & \mathbbm{E}\left[\sum_{g=1}^{Q} \pi(X,g)\left(1-\tilde\pi(X,g)\right)
   \sum_{g^\prime=1}^{g-1} \right. \left. \mathds{1}\left(X\in[\tilde c_{g^\prime},\tilde c_{g^\prime+1}]\right) \cdot \right.\\
   & \hspace{1.5in}\left. \left\{ \mathds{1}\left(G=g\right)\tilde m(X,g^{\prime})+ \mathds{1}\left(G=g^{\prime}\right)\frac{e_g(X)}{e_{g^{\prime}}(X)}(Y-\tilde m(X,g^{\prime})) \right\} \right]
  \\
  & 
  + \mathbbm{E}\left[\sum_{g=1}^{Q}  \left(1-\pi(X,g)\right)\tilde\pi(X,g) \sum_{g^\prime=g+1}^{Q} \right. \left.  \mathds{1}\left(X\in[\tilde c_{g^\prime-1}, \tilde c_{g^\prime}]\right)\cdot\right.\\
  & \hspace{1.5in} \left. \left\{\mathds{1}\left(G=g\right)\tilde m(X,g^{\prime})+\mathds{1}\left(G=g^{\prime}\right)\frac{e_g(X)}{e_{g^{\prime}}(X)}(Y-\tilde m(X,g^{\prime})) \right\} \right].
 \end{aligned}
\end{equation}

Although Equation~\eqref{equ:general:multi:db} resembles the standard doubly robust construction of the average treatment effect, there is a key difference: the nuisance model $e_g(x)$ represents the conditional probability of \emph{group membership} given the running variable, rather than the conditional probability of treatment assignment.
This difference has important implications.
First, our framework neither requires covariate overlap nor unconfoundedness for treatment assignment.  In fact, we also do not require group membership to be unconfounded. Second, the nuisance model is a function of a single-dimensional running variable, making it substantially easier to estimate than typical propensity score models in high dimensional observational causal inference settings.

We adopt a fully nonparametric approach to estimating these nuisance components $\tilde m(\cdot,\cdot)$ and $e_g(\cdot)$ and use cross-fitting \citep{chernozhukov2016locally,athey2021policy}.
Since the running variable is one dimensional, many potential nonparametric estimators can be used.
In our empirical analyses (Section~\ref{sec:app}), we fit a local linear regression using the \texttt{nprobust} package \citep{calonico2019nprobust} with the integrated-MSE bandwidth selector for $\tilde m(\cdot,\cdot)$ and a multinomial log-linear model for $e_g(\cdot)$.

We cross-fit our estimators, randomly splitting the data into $K$ disjoint folds, where each fold contains a subset of data from all $Q$ groups.
For each fold $k$, we use the remaining $K-1$ folds to estimate the nuisance components and then compute the estimates on fold $k$.
We denote these estimates by $\hat{\tilde m}^{-k[i]}(X_i,g)$ and $\hat e_g^{-k[i]}(X_i)$, where we use the superscript $-k[i]$ to denote the model fitted to all but the fold that contains the $i$th observation. 
Now, we can write the proposed cross-fitted doubly-robust estimator as,
\begin{equation}\label{equ:general:multi:empDR}
\begin{aligned}
\widehat\Xi_{\text{DR}}(\pi) = & \ \frac{1}{n}\sum_{i=1}^{n}
\left[\sum_{g=1}^{Q} \pi(X_i,g)\left(1-\tilde\pi(X_i,g)\right) 
   \sum_{g^\prime=1}^{g-1} \right. \left. \mathds{1}\left(X_i\in[\tilde c_{g^\prime},\tilde c_{g^\prime+1}]\right) \cdot \right.\\
   & \hspace{1in}\left. \left\{ \mathds{1}\left(G_i=g\right)\hat{\tilde m}^{-k[i]}(X_i,g^{\prime})+ \mathds{1}\left(G_i=g^{\prime}\right)\frac{\hat{e}^{-k[i]}_g(X_i)}{\hat{e}^{-k[i]}_{g^{\prime}}(X_i)}(Y_i-\hat{\tilde m}^{-k[i]}(X_i,g^{\prime})) \right\} \right.
  \\
  & \left.
  \qquad\qquad+ \sum_{g=1}^{Q}  \left(1-\pi(X_i,g)\right)\tilde\pi(X_i,g) \sum_{g^\prime=g+1}^{Q} \right. \left.  \mathds{1}\left(X_i\in[\tilde c_{g^\prime-1}, \tilde c_{g^\prime}]\right)\cdot\right.\\
  & \hspace{1in} \left. \left\{\mathds{1}\left(G_i=g\right)\hat{\tilde m}^{-k[i]}(X_i,g^{\prime})+\mathds{1}\left(G_i=g^{\prime}\right)\frac{\hat{e}^{-k[i]}_g(X_i)}{\hat{e}^{-k[i]}_{g^{\prime}}(X_i)}(Y_i-\hat{\tilde m}^{-k[i]}(X_i,g^{\prime})) \right\} \right].
\end{aligned}
\end{equation}

For a fixed policy, this representation permits root-$n$ estimation of the identifiable component under the $L^2$ product-rate conditions we describe below in Assumption~\ref{ass:rate:general}. This statement concerns $\Xi_1$ only. The full worst-case criterion also contains nonparametric boundary estimates used to construct the partial-identification bounds, and their generally slower convergence rate is accounted for separately in Section~\ref{sec:theory}.

\subsection{Estimating the bounds}\label{sec:multi:empirical:E2}

We now turn to estimating the worst-case value of $\Xi_{2}$ in Equation~\eqref{equ:general:multi:uniden}. 
We first define the sample analog of $\Xi_{2}$ as follows:
\begin{equation}\label{equ:general:multi:empXi2}
\begin{aligned}
     \widehat\Xi_2(d,\pi):=  \frac{1}{n}\sum_{i=1}^{n}\Big[&\sum_{g=1}^{Q} \mathds{1}\left(G_i=g\right)\cdot \pi(X_i,g)\left(1-\tilde\pi(X_i,g)\right) \sum_{g^\prime=1}^{g-1} \mathds{1}\left(X_i\in[\tilde c_{g^\prime}, \tilde c_{g^\prime+1}]\right) d(1,X_i,g,g^{\prime}) \\
    &+\sum_{g=1}^{Q} \mathds{1}\left(G_i=g\right)\cdot \left(1-\pi(X_i,g)\right)\tilde\pi(X_i,g) \sum_{g^\prime=g+1}^{Q} \mathds{1}\left(X_i\in[\tilde c_{g^\prime-1}, \tilde c_{g^\prime}]\right)  d(0,X_i,g,g^{\prime}) \Big].
\end{aligned}
\end{equation}
Unfortunately, it is not possible to directly optimize $\widehat\Xi_2(d,\pi)$ over the restricted model class $\mathcal{M}_d$ because a model in this class contains an unknown nuisance component ${\tilde{d}}(\cdot,g,g^\prime)$.
This term appears in the point-wise upper and lower bounds, $ B_{\ell}(w,x,g,g^\prime)$ and $ B_{u}(w,x,g,g^\prime)$.
In Appendix~\ref{appendix:add:DRlearner}, we describe a two-stage doubly robust estimator $\hat{\tilde{d}}^{DR}(\cdot,g,g^\prime)$, adapted from \citet{kennedy2020towards}, that regresses a cross-fitted pseudo-outcome on the running variable. Because the bounds depend on the value of this function at a boundary point, its convergence rate is generally nonparametric rather than root-$n$ as explicitly stated in Section~\ref{sec:theory}.

We construct an \textit{empirical} model class $\widehat{\mathcal{M}}_d$ by plugging in the estimates in place of their true values given in Equation~\eqref{equ:genmulti:bound}.
The empirical bounds, $\hat{B}_{\ell}(w,x,g,g^\prime)$ and $\hat{B}_{u}(w,x,g,g^\prime)$, for any value of $x$, are given by:
 \begin{equation}
     \begin{aligned}\label{equ:genmulti:bound:DB}
    \hat{B}_{\ell}(w,x,g,g^\prime)& =  \underset{\begin{subarray}{c}
  x^{\prime} \in \mathcal{X}_{gw} \cap \mathcal{X}_{g^\prime w}   \\
    x^\prime \to \tilde c_{{g^\ast}}
  \end{subarray}}{\operatorname{lim}}
  \hat{\tilde{d}}(x^\prime,g,g^\prime) - \lambda_{wgg^\prime} |x-x^\prime|  \\
 \hat{B}_{u}(w,x,g,g^\prime) & = \underset{\begin{subarray}{c}
  x^{\prime} \in \mathcal{X}_{gw} \cap \mathcal{X}_{g^\prime w}   \\
    x^\prime \to \tilde c_{{g^\ast}}
  \end{subarray}}{\operatorname{lim}}   \hat{\tilde{d}}(x^\prime,g,g^\prime) + \lambda_{wgg^\prime} |x-x^\prime|.
     \end{aligned}
 \end{equation}
We discuss choosing values of the Lipschitz constant $\lambda_{wgg^\prime}$ in Section~\ref{sec:multi:smoothpara} below.

In our implementation, we estimate the boundary levels $\tilde d(\tilde c_{g^\ast},g,g^\prime)$ using a local linear regression (polynomial order $p=1$) of the cross-fitted pseudo-outcome, implemented with the \texttt{nprobust} package \citep{calonico2019nprobust} and its direct plug-in MSE bandwidth selector. In our empirical setting (see Section~\ref{sec:app} below) the outcome is binary, so we additionally project these estimated differences onto $[-1,1]$.
For the outcome-regression nuisance functions entering the integrated doubly robust score, we use local linear regression with the integrated-MSE selector \texttt{imse-dpi} \citep{calonico2019nprobust}. These bandwidths are used for point estimation rather than to construct pointwise confidence intervals for the boundary levels. Section~\ref{sec:theory} formalizes the resulting boundary-estimation error, including both smoothing bias and stochastic error.

With these three components in place, we obtain an \textit{empirical} optimal policy by maximizing the worst-case in-sample value across policies $\pi \in \Pi$:
\begin{equation}\label{equ:multi:estor}
\hat{\pi} \in \ {\operatorname{argmax}}_{\pi \in \Pi} \left\{ \widehat{\mathcal{I}}_{\text{iden}}(\pi) + \widehat\Xi_{\text{DR}}(\pi) + \min _{d \in \widehat{\mathcal{M}}_{d}} \widehat\Xi_2(d,\pi)\right\}.
\end{equation}
In the empirical application, we set $\Pi=\Piol$ though the criterion and statistical results below do not otherwise require this threshold class.

\subsection{Choosing the Lipschitz constants}
\label{sec:multi:smoothpara}

Lipschitz restrictions are often used to obtain partial identification \citep[e.g.,][]{manski1997monotone,kim2018identification}. Assumption~\ref{ass:general:lipdiff} requires an \emph{a priori} choice of the corresponding Lipschitz constants $\lambda_{wgg^\prime}$.
Larger values of $\lambda_{wgg^\prime}$ widen the ambiguity set and weakly reduce the worst-case value of policies that rely on the corresponding extrapolation.
In principle, the choice of $\lambda_{wgg^\prime}$ should be guided by subject-matter expertise.
We supplement such knowledge with a data-driven baseline calibration and a transparent sensitivity analysis.

Because $\lambda_{wgg'}$ is indexed by the treatment state and an unordered pair of groups, the model contains $2\binom{Q}{2}=Q(Q-1)$ pair-state constants. In our empirical setting in Section~\ref{sec:app}, $Q=23$, yielding 506 constants. Rather than selecting them separately, we estimate a baseline value $\widehat\lambda^{\mathrm{base}}_{wgg'}$ for each pair and apply a common multiplier $M$,
\[
\lambda_{wgg'}(M)=M\widehat\lambda^{\mathrm{base}}_{wgg'}.
\]
Varying $M$ therefore provides a one-dimensional sensitivity analysis while retaining the estimated relative heterogeneity across treatment states and group pairs.

To obtain $\widehat\lambda^{\mathrm{base}}_{wgg'}$, we use a two-stage doubly robust procedure that we describe more fully in Appendix~\ref{appendix:add:DRlearner}: we construct a cross-fitted pseudo-outcome for the observed pairwise difference, fit a local quadratic regression in the common-treatment region using the \texttt{nprobust} package with the direct plug-in MSE bandwidth selector \citep{calonico2019nprobust}, and evaluate the absolute first derivative at the identifying boundary. In the empirical analysis, we report results for $M\in\{0,1,2,3,4\}$.
While this provides a data-driven way to calibrate the scale of the Lipschitz constants, it is important to note that the formal guarantees we develop in Section~\ref{sec:theory} hold for a given choice of Lipschitz constants.
We therefore interpret $M$ as a sensitivity parameter and recommend using substantive knowledge to assess the range of plausible values \citep[as suggested by, e.g.,][across a range of methodologogical contexts]{rambachan2023more,kim2018identification,imbens2019optimized}.

\section{Theoretical guarantees}
\label{sec:theory}

To better understand the statistical properties of the empirical policy $\hat{\pi}$, we compute an asymptotic bound on the expected utility loss relative to the baseline policy and the oracle optimal policy within the policy class.
Before presenting the results, we introduce additional assumptions necessary for characterizing the estimation error of the policy value and the regret bound.
These assumptions are stronger than the minimal assumptions required for identification.

\subsection{Assumptions for estimation}

First, as in much of the existing policy learning literature \citep[e.g.,][]{zhou2018offline,nie2021learning,zhan2021policy}, we assume that the potential outcomes are bounded.
\begin{assumption}[Bounded outcome]\label{ass:gen:bo}
There exists a $\Lambda \geq 0$ such that $\left|Y(w)\right|\leq \Lambda$.
\end{assumption}
\noindent This bounded assumption can be generalized to unbounded but sub-Gaussian random variables \citep{zhou2018offline,athey2021policy}.

Second, we consider the following conditions on the convergence rates for the estimators of the nuisance components in $\widehat{\Xi}_{\text{DR}}$ (Equation~\eqref{equ:general:multi:empDR}).
These conditions are satisfied for many semi-parametric or non-parametric estimators (see \cite{zhou2018offline} for a brief summary).

\begin{assumption}[$L^2$ convergence rates of estimated nuisance functions]\label{ass:rate:general}
For every fold $k\in\{1,2,\ldots,K\}$ and every $g,g^\prime\in\mathcal Q$ with $g\neq g^\prime$, let
\[
\widehat r^{-k}_{gg^\prime}(x)=\frac{\widehat e^{-k}_g(x)}{\widehat e^{-k}_{g^\prime}(x)},
\qquad
r_{gg^\prime}(x)=\frac{e_g(x)}{e_{g^\prime}(x)}.
\]
For an independent test observation $X$, conditional on the training sample for fold $k$,
\begin{equation}
\left\|\widehat{\tilde m}^{-k}(\cdot,g)-\tilde m(\cdot,g)\right\|_2=o_p(1),
\qquad
\left\|\widehat r^{-k}_{gg^\prime}-r_{gg^\prime}\right\|_2=o_p(1),
\end{equation}
and
\begin{equation}
\label{equ:rate:general:equ2}
\left\|\widehat{\tilde m}^{-k}(\cdot,g^\prime)-\tilde m(\cdot,g^\prime)\right\|_2
\left\|\widehat r^{-k}_{gg^\prime}-r_{gg^\prime}\right\|_2
=o_p\left(n^{-1/2}\right),
\end{equation}
where $\|f\|_2^2=\mathbb E[f(X)^2]$. In addition, the fitted outcome regressions are clipped to the range $[-\Lambda,\Lambda]$ implied by Assumption~\ref{ass:gen:bo}, and the fitted group-probability vector is projected onto the probability simplex with every component bounded below by a constant $0<\underline\eta<\eta$ before the ratios are formed.
\end{assumption}

Assumption~\ref{ass:rate:general} imposes $L^2$ consistency for each nuisance function, along with a product-rate requirement in Equation~\eqref{equ:rate:general:equ2}. This is the same product-rate structure used in doubly robust policy learning \citep{zhou2018offline,athey2021policy} and semiparametric estimation \citep[e.g.,][]{newey2018cross,farrell2015robust,chernozhukov2018double}. Because the running variable $X$ is one-dimensional here, these conditions can be satisfied by standard nonparametric estimators under suitable regularity conditions. For example, the local polynomial regression we use to estimate $\tilde m$ requires suitable smoothness and bandwidth choice. Note that these rate requirements concern point estimation rather than conventional pointwise inference in RD settings. Although bandwidth selection remains relevant for controlling smoothing bias and stochastic error, we do not construct confidence intervals for these local polynomial fits and therefore do not require the additional bias correction or under-smoothing used for such inference. Finally, the clipping and trimming conditions can be imposed directly after fitting and lead to bounded values in the doubly robust estimators of the identified components. Thus, Assumption~\ref{ass:rate:general} can be viewed as a high-level rate condition that accommodates different nuisance-function estimators.

Finally, we quantify the effect of using the estimated empirical bounds (i.e., using the plug-in estimates of ${\tilde{d}}(x,g,g^\prime)$) on the regret of the learned policy by invoking the following assumption that characterizes the estimation error rate of  $\hat{\tilde{d}}(\cdot,g,g^\prime)$.
\begin{assumption}[Estimation Error of $\hat{\tilde{d}}(\cdot,g,g^\prime)$]\label{ass:bdrate}
For some sequence $\rho_n\rightarrow\infty$, assume that 
$$ \left|\underset{\begin{subarray}{c}
  x^{\prime} \in \mathcal{X}_{gw} \cap \mathcal{X}_{g^\prime w}   \\
    x^\prime \to \tilde c_{{g^\ast}}
  \end{subarray}}{\operatorname{lim}}
  \tilde{d}(x^\prime,g,g^\prime) 
  -\underset{\begin{subarray}{c}
  x^{\prime} \in \mathcal{X}_{gw} \cap \mathcal{X}_{g^\prime w}   \\
    x^\prime \to \tilde c_{{g^\ast}}
  \end{subarray}}{\operatorname{lim}}
  \hat{\tilde{d}}(x^\prime,g,g^\prime)\right|=O_p\left(\rho_n^{-1}\right)$$
for all $w \in \{0,1\}$ and $g,g^\prime\in \mathcal{Q}$.
\end{assumption}

This assumption concerns one-sided estimation of the \textit{conditional} cross-group difference function ${\tilde{d}}(\cdot,g,g^\prime)$ at the existing cutoff values $x=\tilde{c}_{1},\ldots,\tilde{c}_{Q}$. The rate $\rho_n^{-1}$ includes both smoothing bias and stochastic error.
For example, if  $\tilde d(\cdot,g,g^\prime)$ is locally H\"older of order $\gamma$
and we fit a local polynomial of sufficient order with bandwidth $h$, the boundary error has the usual form \citep{fan2018local}
\[
O_p\!\left(h^\gamma+(nh)^{-1/2}+r_n\right),
\]
where $r_n$ denotes the conditional bias remainder generated by replacing the oracle pseudo-outcome with its cross-fitted estimate. The two-stage doubly robust analysis of \citet{kennedy2020towards} gives conditions under which this remainder is negligible relative to the oracle regression rate. In particular, if $r_n=o_p(n^{-\gamma/(2\gamma+1)})$ and $h\asymp n^{-1/(2\gamma+1)}$, then $\rho_n^{-1}\asymp n^{-\gamma/(2\gamma+1)}$. Thus, unlike the integrated identified component, the boundary component is generally not root-$n$ rate. Assumption~\ref{ass:bdrate} keeps the main regret results agnostic about the particular local-polynomial order and bandwidth while requiring their total error, including bias, to vanish.

From the definitions of the lower bounds in Equations~\eqref{equ:genmulti:bound}~and~\eqref{equ:genmulti:bound:DB},
we can write the estimation error of the empirical (lower) bound $\hat{B}_\ell(w, x, g, g')$ in terms of the estimation error of $\hat{\tilde{d}}(\cdot,g,g^\prime)$,
\begin{equation}\label{equ:B&d}
   \begin{aligned}
  &\frac{1}{n}\sum_{i=1}^{n} \max _{w \in \{0,1\},\ g,g^\prime\in \mathcal{Q}} \left | {B}_{\ell}(w,X_i,g,g^\prime)-\hat{B}_{\ell}(w,X_i,g,g^\prime)\right|   \\
  = &\
  \max _{w \in \{0,1\},\ g,g^\prime\in \mathcal{Q}} \left |  \underset{\begin{subarray}{c}
  x^{\prime} \in \mathcal{X}_{gw} \cap \mathcal{X}_{g^\prime w}   \\
    x^\prime \to \tilde c_{{g^\ast}}
  \end{subarray}}{\operatorname{lim}}
  \tilde{d}(x^\prime,g,g^\prime) 
  -\underset{\begin{subarray}{c}
  x^{\prime} \in \mathcal{X}_{gw} \cap \mathcal{X}_{g^\prime w}   \\
    x^\prime \to \tilde c_{{g^\ast}}
  \end{subarray}}{\operatorname{lim}}
  \hat{\tilde{d}}(x^\prime,g,g^\prime)\ 
  \right|. 
  \end{aligned}
\end{equation}
Therefore, under Assumption~\ref{ass:bdrate}, the estimation error of the empirical lower bound is also governed by $\rho_n^{-1}$:
$$ \frac{1}{n}\sum_{i=1}^{n} \max _{w \in \{0,1\},\ g,g^\prime\in \mathcal{Q}} \left |  {B}_{\ell}(w,X_i,g,g^\prime)-\hat{B}_{\ell}(w,X_i,g,g^\prime)\right|=O_p\left(\rho_n^{-1}\right).$$
Here, we focus on the lower bounds since the inner optimization problem of Equation~\eqref{equ:multi:estor} depends only on the worst value of $\hat{\tilde{d}}(x,g,g^\prime)$, i.e, ${B}_{\ell}(w,x,g,g^\prime)$; however, this rate applies to the upper bound as well.
In the subsequent sections, we will examine how this error affects the final performance of the learned policy.

\subsection{Performance relative to the baseline policy}

Under the above assumptions, we compare the learned policy $\hat{\pi}$ to the baseline policy as well as the oracle policy that maximizes the expected utility. 
To simplify the statement of the main theorem, we first define $ {\Gamma}^{(1)}_{ig} $ and $ {\Gamma}^{(0)}_{ig} $ as the doubly-robust scores for Group~$g$ used to extrapolate to the left (from the treated to the control) and right (from the control to the treated) of the current cutoff, respectively: 
\begin{equation}
\begin{aligned}\label{equ:thm:DRscore}    {\Gamma}^{(1)}_{ig} &= \mathds{1}\left(G_i=g\right) \sum_{g^\prime=1}^{g-1}  \mathds{1}\left(X_i\in[\tilde c_{g^\prime},\tilde c_{g^\prime+1}]\right)\tilde m(X_i,g^{\prime})\\
& \qquad +    \sum_{g^\prime=1}^{g-1}\mathds{1}\left(G_i=g^{\prime}\right) \mathds{1}\left(X_i\in[\tilde c_{g^\prime},\tilde c_{g^\prime+1}]\right)  \frac{e_g(X_i)}{e_{g^{\prime}}(X_i)}(Y_i-\tilde m(X_i,g^{\prime})), \\
{\Gamma}^{(0)}_{ig}& = \mathds{1}\left(G_i=g\right)\sum_{g^\prime=g+1}^{Q}  \mathds{1}\left(X_i\in[\tilde c_{g^\prime-1}, \tilde c_{g^\prime}]\right)\tilde m(X_i,g^{\prime})\\
& \qquad + \sum_{g^\prime=g+1}^{Q} \mathds{1}\left(G_i=g^{\prime}\right) \mathds{1}\left(X_i\in[\tilde c_{g^\prime-1}, \tilde c_{g^\prime}]\right)\frac{e_g(X_i)}{e_{g^{\prime}}(X_i)}(Y_i-\tilde m(X_i,g^{\prime})).
    \end{aligned}
\end{equation}

For each group $g$, let $\Pi_g=\{x\mapsto\pi(x,g):\pi\in\Pi\}$ denote the corresponding group-specific policy class, and define
\[
\nu=VC(\Pi),\qquad \overline\nu=\max_{g\in\mathcal Q}VC(\Pi_g).
\]
We treat $Q$ and the prespecified policy class as fixed as $n$ grows. We also define
\[
V^*_{g,\Pi}=\sup_{\pi_a,\pi_b\in\Pi}\mathbbm{E}\left[
\{\pi_a(X_i,g)-\pi_b(X_i,g)\}^2
\left\{(1-\tilde\pi(X_i,g))(\Gamma^{(1)}_{ig})^2
+\tilde\pi(X_i,g)(\Gamma^{(0)}_{ig})^2\right\}
\right].
\]

To characterize the safety property of our learned policy $\hat{\pi}$, we first bound the expected utility loss relative to the baseline policy $\tilde \pi$. 
The appendix gives the corresponding finite-sample statement with explicit constants.
\begin{theorem}[Statistical safety guarantee]\label{thm:multi:safe:general}
Let $\Pi$ be a prespecified measurable policy class satisfying $\tilde\pi\in\Pi\subseteq\Pi_{\mathrm S}$, with $\nu<\infty$ and $\overline\nu<\infty$. Under Assumptions~\ref{ass:general:lipdiff},~\ref{ass:overlap:general},~\ref{ass:gen:bo},~\ref{ass:rate:general}~and~\ref{ass:bdrate}, the expected utility loss of the learned safe policy $\hat{\pi}$ relative to the baseline policy $\tilde{\pi}$ is,
 \begin{equation*}
     \begin{aligned}
V(\tilde{\pi})-V(\hat{\pi})\leq O_p\left(\frac{\Lambda\sqrt{1+\nu}+Q\sqrt{1+\overline\nu}\sqrt{\max_{g\in\mathcal{Q}}V^\ast_{g,\Pi}}}{\sqrt{n}}\vee\rho_n^{-1}\right),
\end{aligned}
 \end{equation*}
\end{theorem}

For the threshold class $\Pi=\Piol$, $\nu=Q$ and $\overline\nu=1$. Thus, the application-specific form of Theorem~\ref{thm:multi:safe:general} is, up to universal constants,
\[
V(\tilde{\pi})-V(\hat{\pi})\leq O_p\left(\frac{\Lambda\sqrt Q+Q\sqrt{\max_{g\in\mathcal Q}V^*_{g,\Piol}}}{\sqrt n}\vee\rho_n^{-1}\right),
\]
The linear factor $Q$ results from summing the group-specific score deviations and is a conservative upper bound rather than an intrinsic lower bound on the difficulty of using more cutoffs.

Theorem~\ref{thm:multi:safe:general} guarantees that the learned policy does not yield worse overall expected utility than the baseline policy, up to vanishing error terms. The two rates have different sources. The $n^{-1/2}$ term comes from empirical averages of the identifiable components, including the orthogonal score for $\Xi_{\mathrm{DR}}$. The term $\rho_n^{-1}$ comes from estimating the boundary anchors used in the partial-identification bounds and is generally slower than  $n^{-1/2}$.
Therefore, the rate will generally be slower than the optimal $O_p(n^{-1/2})$ regret bound achieved in the standard policy learning settings with point identification \citep{kallus2018confounding,Kitagawa2018,zhan2021policy,athey2021policy} due to the partial identification under multi-cutoff RD designs. Such a penalty is generally unavoidable under partial identification settings without more structure or stronger assumptions \citep{kallus2018confounding,khan2023off}.

Nevertheless, in RD settings with a one-dimensional running variable and smooth difference functions, $\rho_n^{-1}$ has the standard one-dimensional nonparametric point-estimation rate described above.
In addition, the inferential target here differs from conventional RD estimation. We are interested in the quality of the empirical policy $\hat{\pi}$ relative to the baseline policy $\tilde{\pi}$, accounting for sampling uncertainty and estimation error, rather than being able to perform statistical inference on values such as the optimal cutoffs.

\subsection{Performance relative to the oracle policy}

Lastly, we  compare our learned policy to the infeasible, oracle optimal policy $\pi^\ast$ in the same policy class, defined as:
$$\pi^{\ast} \in \ \underset{\pi \in  \Pi}{\operatorname{argmax}}\  V(\pi).$$
As before, we report a simpler version of the bound here, while its complete version is found in Appendix~\ref{app:proof:thm2}.
\begin{theorem}[Empirical optimality gap under the multi-cutoff regression discontinuity designs]\label{thm:multi:optimalgap:general}
Under the conditions of Theorem~\ref{thm:multi:safe:general}, the regret of the learned safe policy $\hat{\pi}$ relative to the optimal policy $\pi^{*}$ in $\Pi$ is,
\begin{equation*}
    \begin{aligned}
V(\pi^\ast)-V(\hat{\pi})\leq {}&O_p\left(\frac{\Lambda\sqrt{1+\nu}+Q\sqrt{1+\overline\nu}\sqrt{\max_{g\in\mathcal{Q}}V^\ast_{g,\Pi}}}{\sqrt{n}}\vee\rho_n^{-1}\right)\\
&{}+ \frac{2}{n}\sum_{i=1}^{n} \max _{w \in \{0,1\},g,g^\prime\in\mathcal{Q}}
              \lambda_{wgg^\prime} |X_i-\tilde c_{{g^\ast}}|,
    \end{aligned}
\end{equation*}
 where ${g^\ast}=w(g\vee g^\prime)+(1-w)(g\wedge g^\prime)$ denotes the nearest cutoff index used to extrapolate the difference function $d(w,x,g,g^\prime)$.
\end{theorem}

Compared to Theorem~\ref{thm:multi:safe:general}, there is an additional term proportional to the Lipschitz constant $\lambda_{wgg^\prime}$ and the average distance to the nearest cutoff.
This term measures the sample average of the length of the point-wise partial identification bound for $d(w,X_i,g,g^\prime)$ under Assumption~\ref{ass:general:lipdiff}.
While the robust optimization procedure ensures that our learned policy will be no worse than the baseline policy, asymptotically, Theorem~\ref{thm:multi:optimalgap:general}
shows that this safety guarantee comes at the cost of obtaining a potentially suboptimal policy.
If the cross-group difference deviates substantially from a parallel trend, i.e., $\lambda_{wgg^\prime}>0$ is larger, then the interval will be wider, leading to a potentially larger optimality gap.
However, if the trends are all parallel and all Lipschitz constants $\lambda_{wgg'} = 0$, then Theorem~\ref{thm:multi:optimalgap:general} implies that the regret relative to the optimal policy will converge to zero.

Finally, Theorems~\ref{thm:multi:safe:general} and~\ref{thm:multi:optimalgap:general} treat the Lipschitz constants $\lambda_{wgg^\prime}$ as fixed values that satisfy Assumption~\ref{ass:general:lipdiff}. While we offer a heuristic procedure to calibrate their values in Section~\ref{sec:multi:smoothpara} along with a method for conducting a sensitivity analysis, these results do not account for selecting that scale from the same data.
In different contexts with partial identification, \cite{kallus2018confounding} and \cite{khan2023off} take an approach similar to the one discussed in Section~\ref{sec:multi:smoothpara} by treating such hyper-parameters as known and conducting a sensitivity analysis.

\section{\texorpdfstring{Simulation studies}{Simulation studies}}\label{sec:simu:main}

We examine the finite-sample performance of the proposed method in a two-group, two-cutoff design calibrated to the Colombian ACCES student-loan program analyzed in Section~\ref{sec:app}. In that program, eligibility was determined by a test-score cutoff that varied across regions and years. Following the subset of the ACCES data considered by \citet{cattaneo2020extrapolating}, we set the two status-quo cutoffs to $-850$ and $-571$ and calibrate the remaining features of the data-generating process to the application.

We consider sample sizes $n\in\{1{,}000,2{,}000,4{,}000,8{,}000,16{,}000\}$ and conduct 1,000 replications for each specification. We use the group-specific cutoff-policy class, cross-fitted nuisance estimates, and empirical lower-bound objective described above. The baseline Lipschitz constant is chosen using the procedure in Section~\ref{sec:multi:smoothpara}; we then vary the multiplier $M$ applied to this value. Scenario~A sets the status-quo policy equal to the oracle optimal policy, providing a direct assessment of the safety property. In Scenario~B, the status quo is suboptimal, allowing us to examine both improvement over the status quo and the remaining gap from the oracle. Appendix~\ref{sec:simu} gives the complete data-generating process and estimation details. We implement the method in the \texttt{R} package \texttt{rdlearn}, available at \url{https://cran.r-project.org/package=rdlearn}.

Under Scenario~A, the upper panels of Figure~\ref{fig:simulations} show that the loss relative to the oracle status quo decreases toward zero as the sample size grows. A larger $M$ produces a wider ambiguity set and a more conservative policy, which remains closer to the status quo. This result illustrates the statistical safety guarantee in Theorem~\ref{thm:multi:safe:general}.

The lower panels of Figure~\ref{fig:simulations} report the results under Scenario~B. Regret relative to the status quo is negative in the left panel, indicating that the learned policy improves on the status quo on average for every reported sample size and value of $M$. The improvement generally becomes larger with the sample size and is greater for smaller values of $M$, when the extrapolation bound is more informative. The right panel shows that regret relative to the oracle remains positive but generally declines with the sample size. Because the counterfactual outcomes are only partially identified, this optimality gap need not vanish; it is also larger for higher values of $M$. These patterns illustrate the distinction between safety relative to the status quo and the oracle gap characterized in Theorems~\ref{thm:multi:safe:general} and~\ref{thm:multi:optimalgap:general}.

\begin{figure}[t]
    \centering
    \includegraphics[width=0.82\linewidth]{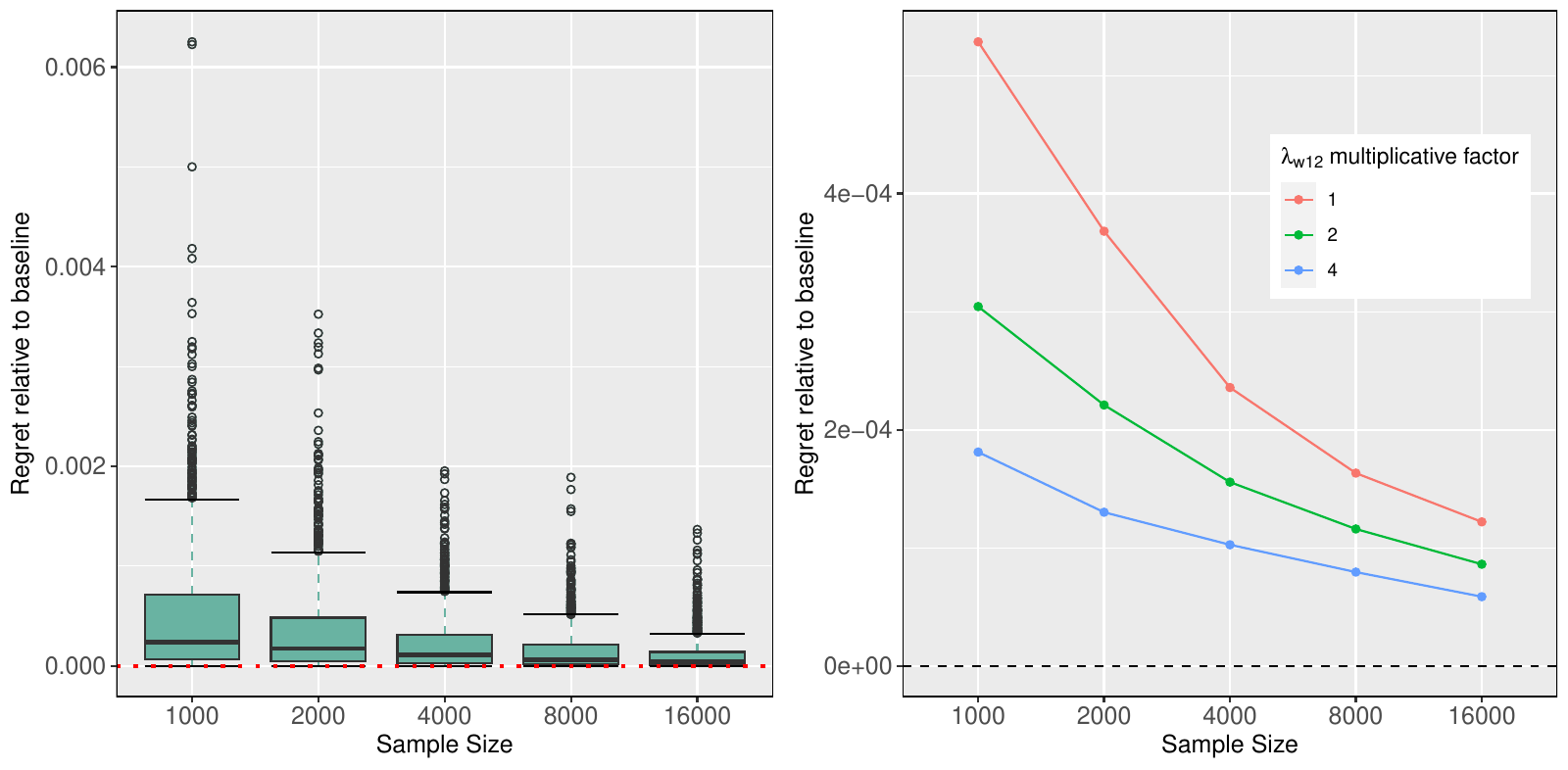}
    \par\smallskip
    \includegraphics[width=0.82\linewidth]{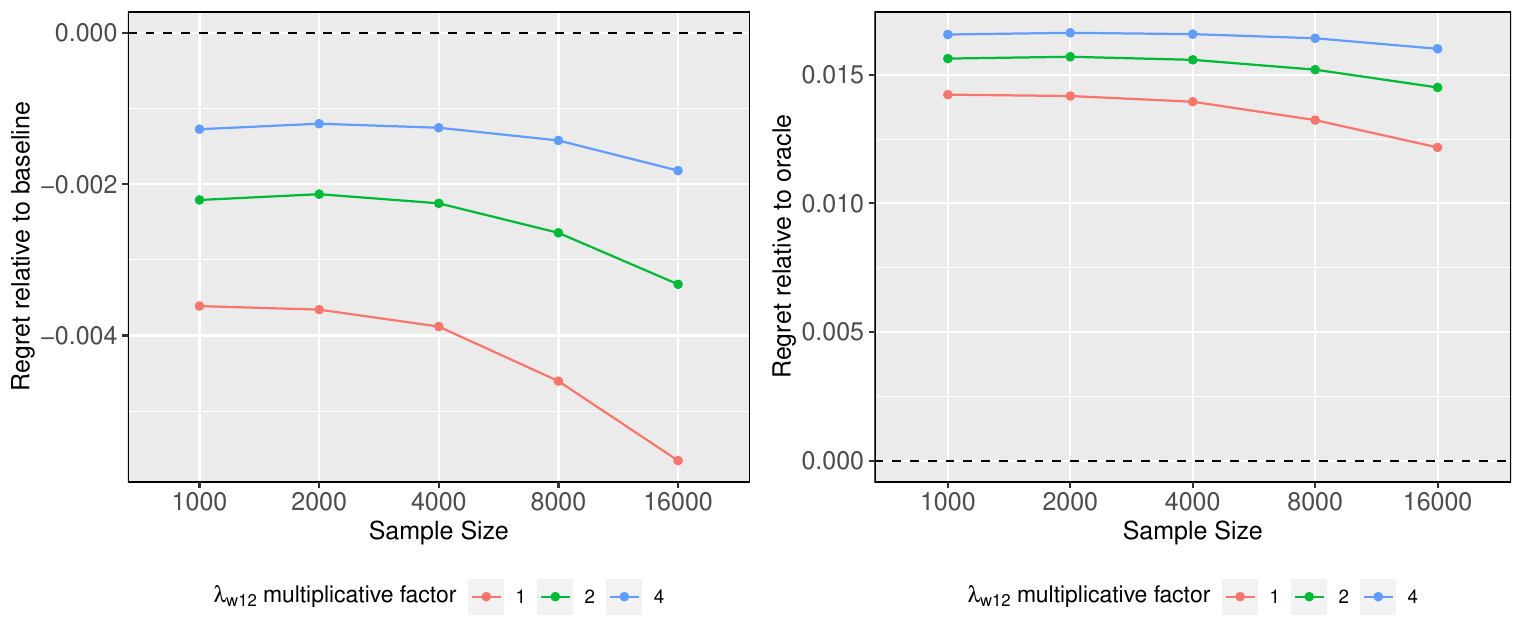}
    \caption{Finite-sample performance under Scenarios~A (upper panels) and~B (lower panels). Under Scenario~A, the left panel presents boxplots of regret relative to the oracle status quo across 1,000 simulations using the data-driven estimate of $\lambda_{w12}$; the right panel presents mean regret as the sample size $n$ and multiplier $M$ vary. Under Scenario~B, the lower-left and lower-right panels present mean regret relative to the status quo and the oracle policy, respectively.}
    \label{fig:simulations}
\end{figure}

\section{Empirical application}\label{sec:app}

In this section, we apply the proposed methodology to the ACCES program, a national-level subsidized loan program in Colombia.
Starting in 2002, the ACCES program has provided financial aid to low-income students.
Eligibility for the program is determined by a test score cutoff, which varies across different regions of Colombia and time periods.
The original analysis by \cite{melguizo2016credit} estimated a single RD treatment effect by pooling and normalizing all 33 cutoffs, while \cite{cattaneo2020extrapolating} focused on two cutoffs from different time periods and extrapolated the RD treatment effect for each one.
In contrast to these prior studies, we jointly analyze department-specific cutoffs in 2010 and learn a revised cutoff policy by leveraging the multi-cutoff design.

\subsection{The background}\label{sec:app:background}

We begin by providing some background information about the ACCES program (see \cite{melguizo2016credit} for more details).
To be eligible for the ACCES program, a student must achieve a high score on the national high school exit exam, the SABER 11, administered by the Colombian Institute for the Promotion of Postsecondary Education (ICFES).
Each semester of every year, once students take the SABER 11, the ICFES authorities compute a position score that ranges from 1 to 1,000, based on each student's position in terms of 1,000 quantiles. 
For example, the students who are assigned a position score of 1 should be in the top 0.1\%, while those with the test scores in the bottom 0.1\% receive a position score of 1000.
The SABER 11 position scores are calculated for each semester, and then pooled for each year.
Throughout our analysis, as in \cite{cattaneo2020extrapolating}, we multiply the position score by $-1$ so that the values of the running variable above a cutoff lead to the program eligibility. 

To select ACCES beneficiaries, the Colombian Institute for Educational Loans and Studies Abroad establishes an eligibility cutoff: an applicant is eligible when the original position rank is less than or equal to the cutoff. After multiplying the rank and cutoff by $-1$, this rule is $X_i\geq\tilde c_{G_i}$.
The eligibility cutoff is defined separately for each department (or region) of Colombia, and has changed over time.
Between 2002 and 2008, the cutoff was $-850$ for all Colombian departments.
After 2008, however, the ICEFES used varying cutoffs across different years and departments.

The 2010 data cover 33 departments with department-specific cutoffs.  These departments are pre-existing administrative groups, and so different students will be subject to different cutoffs depending on their department, irrespective of the value of their position scores.
This means that the overlap condition in Assumption~\ref{ass:overlap:general} requires that in every department there are students with position scores across the range between the lowest and largest cutoffs (i.e. between $-828$ and $-559$).

We conduct our analysis using individual applicant records. We remove 41 records whose recorded eligibility is inconsistent with the deterministic cutoff rule and retain departments with at least 200 distinct running-variable values, yielding 11,635 applicants from 23 departments. Although the running variable is discrete, it takes a large number of distinct values; we implement the method on this finite score grid using local smoothing as an approximation to the continuous-score framework. Following \cite{melguizo2016credit} and \cite{cattaneo2020extrapolating}, the outcome is an indicator for enrollment in a postsecondary institution.
Appendix Table~\ref{tab:emp:summary} reports the department-level sample sizes, baseline cutoffs, and estimated RD treatment effects at the baseline cutoffs.

\subsection{Empirical findings}\label{sec:app:multi}

We estimate the empirical improvement of the expected utility, i.e., $\hat{V}(\pi)-\hat{V}(\tilde \pi)$, for each candidate policy in the policy class $\Piol$, restricting the candidate cutoffs for each department to be within the range of the original lowest and highest cutoffs, i.e., $c_g \in [-828, -559]$.
We first examine the unrestricted policy across the sensitivity grid, then report the capacity-neutral and order-preserving policies, and finally summarize two supplementary benchmarks: a depth-2 policy-tree analysis and a fixed-policy comparison of the selected- and all-reference lower envelopes.

We select a value of the Lipschitz constant $\lambda_{wgg^\prime}$ for each pairwise difference function $d(w,x,g,g^\prime)$ using the procedure described in Section~\ref{sec:multi:smoothpara}. To optimize over the policy class $\Piol$, we note that there are no constraints that hold across departments, or jointly across different values of the cutoffs. In addition, the worst-case expected utility is additive across the departments, so the optimization problem separates into a single-dimensional problem for each department. We solve each optimization problem with a simple grid search: over the finite grid of possible scores, we evaluate the department-specific component of the estimated worst-case utility if the department-specific cutoff was set to that score value. We then select the cutoff for each department that has the largest lower-bound.

Finally, as discussed in Section~\ref{sec:multi:smoothpara}, we also perform a sensitivity analysis where we multiply the baseline values of the Lipschitz constants by a common positive constant $M \in \{0,1,2,3,4\}$, where $M = 0$ corresponds to setting all of the Lipschitz constants equal to zero and so leading to point-identification of the difference function $d(q,\cdot, g, g')$.

\begin{figure}[!t]
    \centering
  \includegraphics[width=0.88\linewidth]{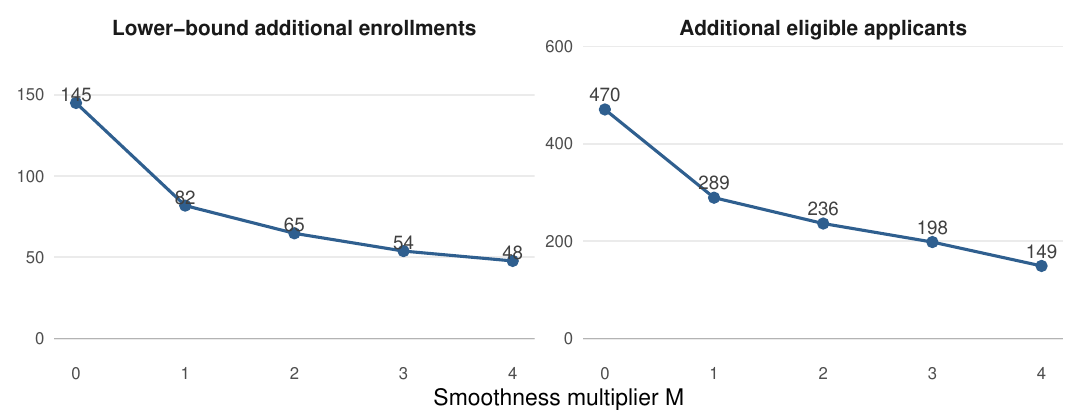}
    \caption{Sensitivity of the unrestricted learned policy to the Lipschitz multiplier $M$. The left panel reports the estimated lower bound on additional expected enrollments in the sample, calculated by multiplying the lower-bound enrollment-rate gain by the sample size of 11,635 applicants. The right panel reports the increase in the number of eligible applicants.}
    \label{fig:realapp:compare}
\end{figure}

 Figure~\ref{fig:realapp:compare} shows two ways to quantify the new policies: (i) the estimated lower bound on additional expected enrollments in the sample and (ii) the number of additional applicants that would be made eligible by the new policy. We evaluate these two quantities as the multiplicative factor $M$ changes.
When $M=0$ and the cutoff differences are assumed to be exactly constant, i.e., $\lambda_{wgg^\prime}=0$, there is the largest estimated lower-bound enrollment gain and the largest expansion in eligibility.
Next, we consider the data-driven Lipschitz specification corresponding to $M=1$ in Figure~\ref{fig:realapp:compare}.
The 506 pair-state baseline values range from $4.1\times 10^{-6}$ to $1.3\times 10^{-1}$, with a median of $5.1\times 10^{-3}$. Relative to the constant-difference benchmark, the $M=1$ policy yields a smaller estimated improvement and a smaller expansion in eligibility because we have moved from point identification to partial identification.
As $M$ increases, both estimated gains decline gradually and the learned policy moves toward the status quo, as expected when the ambiguity set becomes wider.

Increasing the number of eligible students could lead to a violation of overall capacity constraints. We therefore consider a capacity-neutral analysis that restricts the total number eligible to be no greater than under the status quo. We also impose a monotonicity constraint by requiring the learned cutoffs to preserve the ordering of the status-quo department cutoffs. Appendix~\ref{append:sec:realplots} describes both finite-grid optimization problems.

Table~\ref{tab:realapp:welfare} reports the results at $M=1$. Under the status quo, the observed enrollment rate is 0.7944, and 8,971 out of the 11,635 applicants are eligible. The unrestricted learned policy has an estimated lower-bound enrollment rate of 0.8014, which corresponds to an estimated lower bound of approximately 82 additional expected enrollments in this sample, and increases eligibility by 289 applicants.

The capacity-neutral policy has an estimated lower-bound value of 0.7974 with an estimated lower bound of approximately 34 additional expected enrollments. Figure~\ref{fig:realapp:capacitycompare} compares its department-specific cutoff and eligibility changes with those of the unrestricted policy. Relative to the unrestricted policy, the capacity-neutral policy selects the same cutoff in three departments and a higher cutoff in the remaining twenty. Relative to the status quo, it reallocates 49 eligibility slots: eligibility increases in eight departments, decreases in eleven, and is unchanged in four. The largest increases in eligibility occur in Cauca, Santander, and Quindío, while the largest reductions occur in Huila, Sucre, and Putumayo. Among these departments, gross higher-education coverage in 2010 was approximately 48\% in Santander and 50\% in Quindío, compared with 26\% in Huila, 17\% in Sucre, and 12\% in Putumayo \citep{mineducacion2011}. This suggests that the learned policy may partly shift eligibility toward settings with a broader existing higher-education sector, where additional loan eligibility may translate more readily into enrollment. The pattern is not systematic, however, as Cauca receives the largest increase despite having substantially lower coverage.

\begin{figure}[t]
    \centering
    \includegraphics[width=0.98\linewidth]{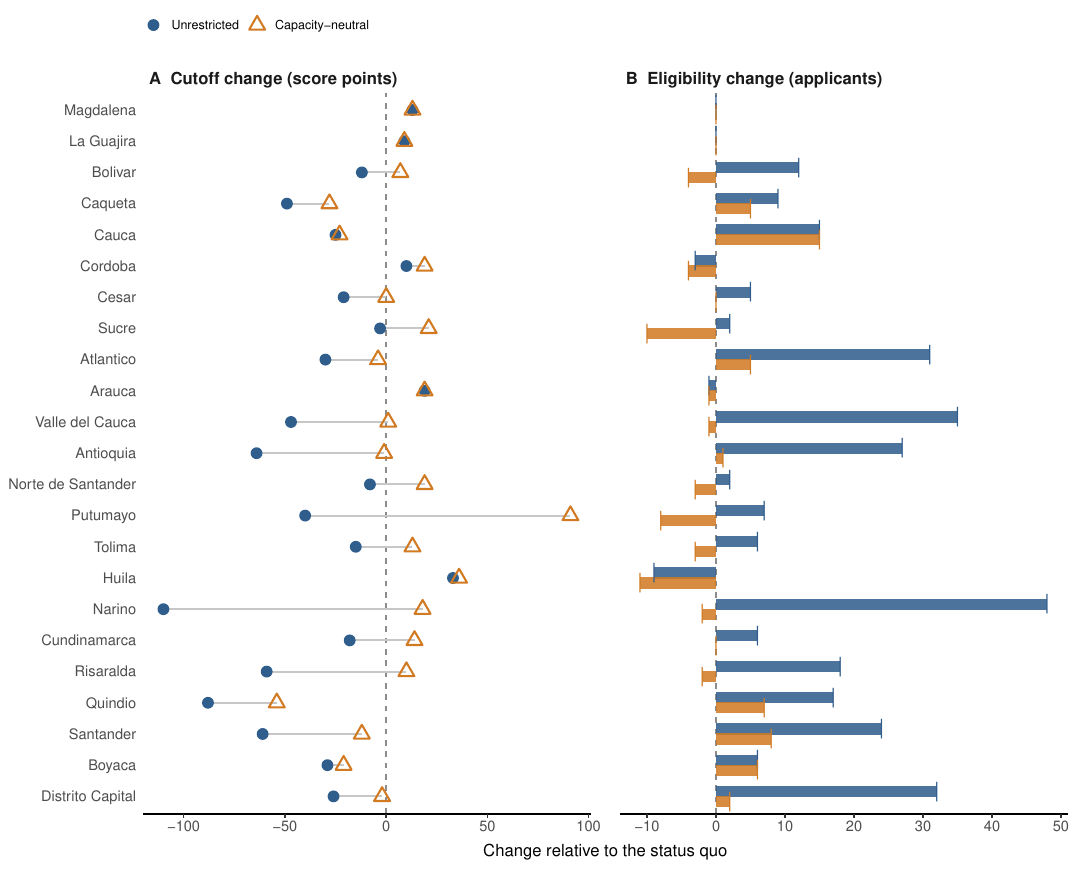}
    \caption{Department-specific changes under the unrestricted and capacity-neutral policies at $M=1$. Panel A reports each learned cutoff relative to the status-quo cutoff, $\widehat c_g-\widetilde c_g$; negative values indicate a less stringent eligibility threshold. Panel B reports the corresponding change in the number of eligible applicants. The unrestricted policy expands eligibility by 289 applicants in total, whereas the changes under the capacity-neutral policy sum to zero.}
    \label{fig:realapp:capacitycompare}
\end{figure}

The order-preserving policy has an estimated lower-bound value of 0.8001, an estimated lower bound of approximately 67 additional expected enrollments, and makes 329 additional applicants eligible.

Appendix~\ref{append:sec:tree} also applies the same criterion to a depth-2 tree class that permits non-threshold assignment rules.
This decision policy groups together departments and applies a common adjustment to the thresholds for each group of departments.
Roughly, the selected tree reduces the threshold by a greater amount for the half of departments with larger status quo cutoffs (and so more stringent requirements) than for the half of departments with lower status quo cutoffs.
This policy has an estimated lower-bound value of 0.7984, above the status-quo value of 0.7944 but below the value of 0.8014 for the main department-specific cutoff policy.

Finally, we assess the conservativeness of selecting one reference group by holding the unrestricted $M=1$ policy fixed. Its estimated lower value under our main cross-fitted doubly robust procedure is 0.8014. In a matched plug-in comparison, the selected-reference and all-reference lower envelopes give values of 0.8008 and 0.8050, respectively, a difference of about 0.42 percentage points. Both plug-in evaluations remain above the status-quo value of 0.7944. Appendix~\ref{append:sec:realplots} explains that this analysis evaluates a fixed policy rather than reoptimizes under the all-reference criterion.

\begin{table}[t]
\centering
\caption{Estimated policy values and eligibility under the applicant-level analysis}

\begin{tabular}{lcccc}
\hline
Policy & $M$ & Worst-case value & Improvement & Eligible applicants \\
\hline
Status quo & -- & 0.7944 & 0 & 8,971 \\
Unrestricted learned policy & 1 & 0.8014 & 0.0070 & 9,260 \\
Capacity-neutral learned policy & 1 & 0.7974 & 0.0029 & 8,971 \\
Order-preserving learned policy & 1 & 0.8001 & 0.0057 & 9,300 \\
\hline
\end{tabular}

\label{tab:realapp:welfare}
\end{table}

\section{Concluding remarks}\label{sec:discussion}

Despite the popularity of regression discontinuity (RD) 
designs for program evaluation with observational data, little attention has been given to the question of how to learn new and better treatment assignment cutoffs.
In this paper, we propose a methodology for safe policy learning and establish population and statistical safety guarantees relative to the existing status quo.
We partially identify the expected utility of candidate policies and use robust optimization to select the policy with the highest worst-case value.
We stress, however, that extrapolation is required to learn new policies.
Our approach leverages the presence of multiple cutoffs and uses a nonparametric Lipschitz restriction on cross-group differences in order to conduct this extrapolation.

Although the application focuses on department-specific cutoff policies, the identification argument applies to any binary policy that agrees with the status quo outside the range of existing cutoffs, and the statistical guarantees extend to prespecified finite-VC subclasses of this supported class that contain the status quo.

There are several directions for future research.
First, the partial-identification assumption involves researcher-specified Lipschitz constants.
While we have suggested ways that the data can inform these hyperparameters, developing rigorous selection criteria for this choice will be an important aspect of future work.
Second, we can generalize Assumption~\ref{ass:general:lipdiff} to the case where it holds only after conditioning on the other available covariates in RD design.
Consequently, our proposed doubly robust estimator can be modified to include not only the single running variable, but also other observed covariates.

Third, an important extension is to incorporate constraints defined by population quantities, such as a hard budget limit $\E[\pi(X,G)]\leq\delta$, while accounting for estimation uncertainty in the constraint. Structural restrictions such as monotonicity, simplicity, or a common cutoff can already be imposed through the policy class; the application illustrates several such cases. If a restricted class excludes the status quo, however, the safety comparison to the status quo no longer applies.
Related work on empirically constrained policy learning may provide a useful starting point for this extension \citep{sun2021empirical}.

Fourth, while our current methodology primarily addresses scenarios with continuous running variables, future work should extend the proposed framework to handle discrete running variables, which introduce partial identification challenges similar to those posed by extrapolation. As noted by \cite{kolesar2018inference} and \cite{imbens2019optimized}, explicit smoothness assumptions on the outcome or difference functions might help address these challenges.

Fifth, we can change the target population of interest and expand the notion of safety.
In particular, it is of interest to establish a statistical safety guarantee for a specific group $g$  or other well-defined subpopulation so that a learned policy does not perform worse in this subpopulation than the status quo policy \citep{jia2023}.
This differs from the current guarantees that hold over the entire population on average, and so can lead to worse outcomes for certain subpopulations.
Finally, one may consider the use of safe policy learning in other study designs.
For example, the tie-breaker design of \cite{li2022general} combines treatment assignment based on both randomization and thresholding.

\bigskip

\bibliographystyle{chicago}
\bibliography{library}

\clearpage
\appendix

\renewcommand\thefigure{\thesection.\arabic{figure}}    
\setcounter{figure}{0}

\renewcommand\thefigure{\thesection.\arabic{figure}}
\renewcommand\thetable{\thesection.\arabic{table}}
\renewcommand\thetheorem{\thesection.\arabic{theorem}}
\renewcommand\thecorollary{\thesection.\arabic{corollary}}
\renewcommand\thelemma{\thesection.\arabic{lemma}}
\renewcommand\theproposition{\thesection.\arabic{proposition}}
\renewcommand\theequation{\thesection.\arabic{equation}}
\renewcommand\theHequation{appendix.\thesection.\arabic{equation}}
\renewcommand\theassumption{\thesection.\arabic{assumption}}
\renewcommand\theHassumption{appendix.\thesection.\arabic{assumption}}
\renewcommand\theremark{\thesection.\arabic{remark}}

\setcounter{assumption}{0}
\setcounter{theorem}{0}
\setcounter{lemma}{0}
\setcounter{proposition}{0}
\setcounter{table}{0}
\setcounter{equation}{0}
\setcounter{remark}{0}

\begin{center}
 \huge Supplementary Appendix for ``Safe Policy Learning under Regression Discontinuity Designs with Multiple Cutoffs''
\end{center}

\etocdepthtag.toc{mtappendix}
\etocsettagdepth{mtchapter}{none}
\etocsettagdepth{mtappendix}{subsection}

\section{Safe Policy Learning under Single-cutoff Designs}\label{appendix:singlecutoff}

Building on the proposed robust optimization framework in the main text, we demonstrate how to learn a new and improved cutoff under single-cutoff RD designs.

\subsection{Problem setup and notation}
\label{sec:singlecutoff:notation}

Suppose that for each unit $i=1,\ldots,n$, we observe a univariate running variable $X_i \in \mathcal{X} \subset \R$, a binary treatment assignment $W_i \in \{0,1\}$, and an outcome variable $Y_i \in \R$.
We assume that $X$ has a continuous positive density $f_X$ on the support $\mathcal{X}$. 
Under the sharp RD design with a single cutoff, the treatment assignment $W_i$ is completely determined by the value of $X_i$ being on an either side of a fixed, known cutoff $\tilde{c}$, i.e., $W_i=\mathds{1}(X_i\geq \tilde{c})$ with $\mathds{1}\left(\cdot\right)$ denoting the indicator function.

In this section, we consider a class of linear threshold policies: 
\begin{equation}
\Pisg \ = \ \left\{\pi(x)=\mathds{1}(x\geq c ) \mid c, x \in\mathcal{X}\right\}, \label{eq:single_class}
\end{equation}
where each cutoff value $c$ in the support of the running variable corresponds to a candidate policy in the policy class. 
Then, the optimal policy $\pi^*$ within $\Pisg$ is defined as,
$\pi^{*} \ = \ \underset{\pi \in \Pisg}{\operatorname{argmax}}\  V(\pi)$.
We also define the baseline policy applied in the observed data as follows: $$\tilde\pi(X_i):=\mathds{1}(X_i\geq \tilde c),$$ 
Importantly, the baseline $\tilde\pi$ is contained in the linear threshold policy class $\Pisg$.
As seen below, this is needed in order for the resulting optimal policy to have a safety guarantee, i.e., its performance being at least as good as that of the baseline policy.

\subsection{Policy learning through robust optimization}

A major limitation of the RD designs is the unidentifiability of the counterfactual outcomes under the cutoffs different from the existing one. 
Let $m(w,x)=\mathbbm{E}\left[Y(w)\mid X=x\right]$ denote the conditional expectation of the potential outcomes. 
By comparing the deviation of a given policy $\pi$ from the baseline policy $\tilde\pi$, we can decompose the expected utility into point-identifiable and non-identifiable terms, respectively:
\begin{equation}
\begin{aligned}\label{equ:sing:V}
        V(\pi,m)\ = \ & \mathbbm{E}\left[ \left\{\pi(X)\tilde\pi(X)+(1-\pi(X))(1-\tilde\pi(X))\right\} Y\right]+\\
        &\mathbbm{E}\left[  \pi(X)(1-\tilde\pi(X)) m(1,X)\right]+\mathbbm{E}\left[  (1-\pi(X))\tilde\pi(X) m(0,X)\right]. 
\end{aligned}
\end{equation}
The first term of Eqn~\eqref{equ:sing:V} corresponds to the cases where the policy $\pi$ agrees with the baseline policy $\tilde\pi$.
This component can be identified from observable data.
The terms in the second line of Eqn~\eqref{equ:sing:V}, however, are not identifiable without further assumptions since the counterfactual outcome cannot be observed when the policy under consideration $\pi$ disagrees with the baseline policy $\tilde\pi$.

The evaluation of $V(\pi,m)$, therefore, requires some methods to extrapolate $m(w,x)$ beyond the baseline cutoff. 
The lack of identifiability means that multiple models of $m(w,x)$ can fit equally well to the observed data.
To address this problem, we first show that a standard smoothness assumption used for estimation in RD designs implies that $m(w,x)$ belongs to a particular model class $\mathcal{F}$ (see Section~\ref{subsec:smooth}).
We then combine this observation with the fact that when the policy $\pi$ agrees with the baseline policy $\tilde{\pi}$, the conditional expectation function of potential outcome $m(w,x)$ is identifiable, i.e.,
$$\tilde m(x) \ := \ m(\tilde \pi (x), x)= \mathbbm{E}\left[Y \mid X=x\right].$$
This leads to the following restricted model class,
\begin{equation}\label{equ:sing:M}
    \mathcal{M}=\{f \in \mathcal{F} \mid f(\tilde{\pi}(x), x)=\tilde{m}(x) \; \forall x \in \mathcal{X}\},
\end{equation}
which partially identifies a range of potential values for $m(w,x)$ given a policy $\pi$.

Lastly, we use robust optimization to find an optimal policy that maximizes the worst case value of the expected utility $V(\pi,m)$, with $m(w,x)$ taking values in the ambiguity set $\mathcal{M}$:
\begin{equation}
    \begin{aligned}\label{equ:Vinf}
       \pi^{\inf} \in  \underset{\pi \in \Pi_{\mathrm{Single}}}{\operatorname{argmax}} \  \min_{m \in \mathcal{M}}\  V(\pi,m) \Longleftrightarrow \pi^{\inf } \in \underset{\pi \in \Pi}{\operatorname{argmax}} \min _{m \in \mathcal{M}}\{V(\pi, m)-V(\tilde{\pi})\}.
    \end{aligned}
\end{equation}

As \cite{ben2021safe} show, the overall expected utility of this learned policy $\pi^{\inf}$ is never less than that of the baseline policy.
This safety property holds because we only optimize over the unknown components, while the worst-case value and the true value coincide for the baseline policy, i.e., $ \min_{m \in \mathcal{M}}\  V(\tilde{\pi},m)=V(\tilde{\pi})$.
Thus, so long as the baseline policy is in the policy class, the resulting policy never performs worse than the baseline policy.
We now turn to the details of the proposed methodology and its properties under the single-cutoff RD designs.

\subsection{Partial identification via smoothness}
\label{subsec:smooth}

For partial identification, we rely on a smoothness assumption that is similar to --- but weaker than --- the assumptions commonly used in estimating the CATE at the cutoff under the RD designs.
Specifically, we assume that the conditional expectation function of potential outcomes is approximately linear.
\cite{armstrong2018optimal} consider the same model class and propose a method for minimax linear estimation of the CATE at the cutoff. 
\begin{assumption}[Approximately linear model \citep{sacks1978linear}] \label{ass:bddsndderiv}
The conditional expectation function of potential outcomes $m(w,x)$ lies in the following model class,\\
$$\mathcal{F}=\left\{f: \left|f(w, x)-f(w, \tilde c)-f_x^{(1)}(w, \tilde c)(x-\tilde c)\right|\leq S_w(x-\tilde c)^2 \ \forall x\in \mathcal{X}\right\},$$
where $f_x^{(1)}$ is the first derivative with respect to $x$, and $S_w < \infty$ is a positive constant for $w\in \{0,1\}$.
\end{assumption}

Assumption~\ref{ass:bddsndderiv}  states that the approximation error due to the first-order Taylor expansion of $m(w,x)$ around the existing cutoff $\tilde c$ is bounded by $S_w(x-\tilde c)^2$, uniformly over $x\in\mathcal{X}$.
This is closely related to the restriction on the second derivative of the conditional expectation function that has been commonly used in the RD design literature.
For instance, a locally bounded second derivative assumption --- $m_x^{(2)}(w,x)$ exists and is bounded in a neighborhood of $\tilde c$ --- is often used to justify the application of local linear regression \citep[e.g.][]{cheng1997automatic,imbens2012optimal}.
In addition, \citet{imbens2019optimized} rely on a more stringent assumption of a bounded global second derivative to obtain the optimal minimax linear estimator.
Others consider the same and similar smoothness restrictions \citep[e.g.][]{noack2019bias,armstrong2018optimal}.
While Assumption~\ref{ass:bddsndderiv} is similar to the bounded second-derivative assumption, it is technically a weaker condition since it does not impose any smoothness away from the cutoff.

Under Assumption~\ref{ass:bddsndderiv}, we extrapolate $m(w,x)$ beyond the cutoff $\tilde{c}$ and partially identify the last two components of Eqn~\eqref{equ:sing:V}.
We first compute the restricted model class $\mathcal{M}$ for $m(w,x)$:
\begin{equation}\label{equ:sing:modelclass}
    \mathcal{M}=\left\{ f: \{0,1\}\times \mathcal{X} \rightarrow \mathbb{R} \mid  B_\ell(w,x) \leq  f(w,x) \leq B_{u}(w,x) \right\}.
\end{equation}
Here, the potential values of $m(w,x)$ are characterized by the following set of point-wise upper and lower bounds extrapolated from the cutoff $\tilde c$ from both the left and right sides: 
 \begin{equation}
     \begin{aligned}\label{equ:sing:boundfn}
     B_{\ell}(w,x)& =\underset{\begin{subarray}{l}
   x^\prime \in \mathcal{X}_w  \\
    x^\prime \rightarrow \tilde{c}
  \end{subarray}}{\operatorname{lim}} \;  \tilde m(x^\prime) + \tilde m^{(1)}(x^\prime)(x-x^\prime)-S_w(x-x^\prime)^2 \\
B_{u}(w,x)& =\underset{\begin{subarray}{l}
   x^\prime \in \mathcal{X}_w  \\
    x^\prime \rightarrow \tilde c
  \end{subarray}}{\operatorname{lim}} \;\tilde m(x^\prime) + \tilde m^{(1)}(x^\prime)(x-x^\prime)+S_w(x-x^\prime)^2,
     \end{aligned}
 \end{equation}
where $\mathcal{X}_w=\{x \in \mathcal{X} \mid \tilde{\pi}(x)=w\}$ denotes the set of running variable values with the baseline policy giving the treatment assignment $w$.
Note that the limits are taken as the value of the running variable $x'$ approaches the baseline cutoff $\tilde{c}$ from within $\mathcal{X}_w$.

Under this setup, the inner optimization problem of Eqn~\eqref{equ:Vinf} has a closed form solution,
$\min_{m \in \mathcal{M}}\  V(\pi,m) = V(\pi, B_{\ell})$,
which finds the worst-case value at the lower bound $B_{\ell}$. 
The maximin problem can then be written as a single maximization problem,
\begin{equation}
    \begin{aligned}\label{equ:sing:VB}
       \pi^{\inf} \in \underset{\pi \in \Pi_{\mathrm{Single}}}{\operatorname{argmax}} \  V(\pi,B_{\ell}).
    \end{aligned}
\end{equation}

\subsection{Learning policies from data}\label{sec:sing:empirPL}

The estimation of the optimal policy defined in Eqn~\eqref{equ:sing:VB} requires inferring $V(\pi,B_{\ell})$ from observed data for any specific policy $\pi$.
We can estimate this expected utility via its empirical analog,
\begin{equation}
\begin{aligned}\label{equ:sing:Voracle}
     \widehat V(\pi,B_{\ell})\ = \ \frac{1}{n}\sum_{i=1}^{n} & \left\{\pi(X_i)\tilde\pi(X_i)+(1-\pi(X_i))(1-\tilde\pi(X_i))\right\} Y_i\\
     &+\pi(X_i)(1-\tilde\pi(X_i)) B_{\ell}(1,X_i)+ (1-\pi(X_i))\tilde\pi(X_i) B_{\ell}(0,X_i).
\end{aligned}
\end{equation}

Unfortunately, it is not possible to directly optimize Eqn~\eqref{equ:sing:Voracle} because it involves the lower bound $B_{\ell}(w,x)$, which contains unknown nuisance components $\tilde m(\cdot)$ and its first derivative $\tilde m^{(1)}(\cdot)$.
Therefore, we propose using the estimates of these nuisance functions, $\hat{\tilde{m}}(\cdot)$ and $\hat{\tilde{m}}^{(1)}(\cdot)$, in place of the true values in Eqn~\eqref{equ:sing:boundfn}, yielding the empirical lower bound, $\widehat B_{\ell}(w,x)$. 
In practice, these nuisance estimates can be obtained using the R package \texttt{nprobust}.
In addition, we must choose the bound on the second derivative $S_w$ in the bound since it is unknown. We discuss this estimation problem in the next section.

Finally, we can obtain the empirical safe policy for the single-cutoff RD by solving,
\begin{equation}
\begin{aligned}\label{equ:sing:emp:estor}
\hat{\pi} \in  \underset{\pi \in \Pi_{\mathrm{Single}}}{\operatorname{argmax}}\  \widehat V(\pi, \widehat B_{ \ell}).
    \end{aligned}
\end{equation}

\subsection{Choosing smoothness parameters $S_w$}\label{sec:sing:smoothpara}

The application of our methodology requires choosing the value of smoothness parameter $S_w$, where a greater value of $S_w$ yields a more conservative policy.
In principle, this choice should be guided by subject knowledge, but in practice, having a good default data-driven method is useful.
In the RD design literature, it is common to specify the smoothness parameter by fitting a global polynomial regression on either side of the cutoff.
For example, \cite{imbens2019optimized} propose to estimate $m(w,x)$ using a global quadratic regression, and then multiply maximal curvature of the fitted model by a constant, e.g., between 2 and 4.

We instead propose a nonparametric method for the specification of the smoothness parameter. 
We first fit a local quadratic polynomial regression \citep{calonico2019nprobust} separately for the observations on either side of the cutoff, and obtain the estimated second derivative of the conditional mean potential outcome function at a grid of values within the empirical support of the running variable.
We then take the largest local curvature as $S_w$. 
In practice, we also recommend a sensitivity analysis based on a range of plausible smoothness parameter values.

\section{A Doubly Robust Estimator of Heterogeneous Cross-group Differences}\label{appendix:add:DRlearner}
Due to the discontinuity at the cutoffs, we need to estimate the conditional differences between observed treatment groups or between observed control groups using data above or below the threshold, respectively. Below we discuss how to estimate differences across treatment groups using a portion of the data; estimates for the control group can be obtained similarly using the another portion of the data.
\begin{Algorithm}[Estimating observed conditional differences across treatment groups]
For group $g$ and its comparison group $g^\prime$, let $\left(D_1^n, D_2^n\right)$ denote two independent samples of observations with $W_i=1$ from the full collection of groups. Define $e_h(x)=\Pr(G_i=h\mid X_i=x)$ for every group $h$. \\

\noindent\textit{Step 1. Nuisance training}:\\
(a) Construct estimates $\hat{e}_g$ of the group propensity score $e_g$ using $D_1^n$.\\
(b) Construct estimates $\left(\hat{\tilde m}(\cdot,g), \hat{\tilde m}(\cdot,g^\prime)\right)$ of the group-specific regression functions $\left({\tilde m}(\cdot,g), {\tilde m}(\cdot,g^\prime)\right)$ using $D_1^n$.\\

\noindent\textit{Step 2. Pseudo-outcome regression: Construct the pseudo-outcome}
$$
\widehat{\varphi}(Z;g,g^\prime)= \hat{\tilde m}(X,g)-    \hat{\tilde m}(X,g^\prime)+\frac{ \mathds{1}(G=g)}{ \hat e_g(X)}\left(Y-  \hat{\tilde m}(X,g)\right)-\frac{ \mathds{1}(G=g^\prime)}{\hat e_{g^\prime}(X)}\left(Y-  \hat{\tilde m}(X,g^\prime)\right)
$$
and regress it on covariates $X\geq \operatorname{max}\{\tilde{c}_g,\tilde{c}_{g^\prime}\}$
in the test sample $D_2^n$, yielding
$$
\hat{\tilde{d}}^{DR}(x,g,g^\prime)=\widehat{\mathbb{E}}_n\{\widehat{\varphi}(Z;g,g^\prime) \mid X=x\}
$$
for $x\geq \operatorname{max}\{\tilde{c}_g,\tilde{c}_{g^\prime}\}$, i.e., the treatment groups.
Here, $\widehat{\mathbb{E}}_n$ can represent any generic nonparametric regression methods.\\
 
\noindent\textit{Step 3. Cross-fitting (optional)}: 
Repeat \textit{Step 1-2}, swapping the roles of $D_1^n$ and $D_2^n$ so that $D_2^n$ is used for nuisance training and $D_1^n$ as the test sample. Use the average of the resulting two estimators as a final estimate of $\tilde{d}$. $K$-fold variants are also possible.

\noindent Similarly, to estimate the conditional difference between observed control outcomes, we use observations with $W_i=0$ from the full collection of groups and fit the pseudo-outcome regression in Step 2 using $X<\operatorname{min}\{\tilde{c}_g,\tilde{c}_{g^\prime}\}$.
\end{Algorithm}

\section{Simulation details}\label{sec:simu}

This section provides the complete setup and estimation details for the simulation studies summarized in Section~\ref{sec:simu:main}. We compare the utility loss of the learned policy relative to the baseline and oracle policies in a multi-cutoff RD design with two groups, $g\in\mathcal{Q}=\{1,2\}$.

\subsection{The setup}\label{sec:simu:setup}

We consider a realistic simulation setting based on the ACCES loan program data analyzed in Section~\ref{sec:app}.
Specifically, the data generating process is given by,
\begin{equation*}
    \begin{aligned}
        &X_i \sim \text{Unif}(-1000,-1)\\
        &G_i = 2 \cdot\mathds{1}(0.01\cdot X_i+5+\epsilon_i>0)+ \mathds{1}(0.01\cdot X_i+5+\epsilon_i\leq0),\quad \epsilon_i\sim \mathcal{N}(0,\sigma_\epsilon^2)\\
        &W_i = \sum_{g=1}^{2} \mathds{1}(G_i =g)\mathds{1}(X_i \geq \tilde c_g)\\
    \end{aligned}
\end{equation*}
with the sample size $n\in\{1000,2000,4000,8000,16000\}$.
Following \citet{cattaneo2020extrapolating} who analyzed a subset of the ACCES data, we focus on two sub-populations with cutoff values $\tilde c_2=-571$ and $\tilde c_1=-850$.

Under this setting, the group indicator $G_i$ is determined by the running variable $X$ and a Normal random variable $\epsilon$, with mean $0$ and variance $\sigma_\epsilon^2$.
The value of $\sigma_\epsilon^2$ controls the degree of overlap between the two groups with a higher value leading to a greater overlap.
We set $\sigma_\epsilon^2=10^2$ so that the conditional probability of belonging to group $G=2$ is roughly between 0.3 and 0.7 for all samples, which gives a strong overlap between the two groups.

The outcome variable is generated according to the following model:
\begin{equation}\label{equ:sim:Ym}
       Y=\gamma_0^\top g(X) (1-W)+\gamma_1^\top g(X)W - d(W,X)\cdot \mathds{1}(G =1)+ \epsilon_Y,
\end{equation}
where the error term $\epsilon_Y$ follows a normal distribution with mean 0 and standard deviation of $0.3$. We consider two data generating scenarios.
For both Scenarios, we choose $g(X)$ to be a third-order polynomial, and select 
$(\gamma_0,\gamma_1)$ by first fitting a third-order polynomial model using the data from our empirical application, and then further calibrating the fitted coefficients such that the baseline policy $\tilde\pi$ equals the oracle policy $\pi^*$ under Scenario A, and that the baseline policy $\tilde\pi$ differs from the oracle policy $\pi^*$ under Scenario B.
The resulting specifications are:
\begin{equation}\label{sim:multi:scenario}
    \begin{aligned}
        (A)\quad 
          & d(w,x):=d(w,x,2,1)= 0.2 + e^{0.01x}+ 0.3(1-w)\\ 
          & g(X)=\left(1,(X-164.43),(X-164.43)^2,(X-164.43)^3\right)^\intercal\\
          &\gamma_0=( -1.99 ,-1.00e(-02) ,-1.20e(-05) ,-4.59e(-09))^\intercal\\
        &\gamma_1=(0.96, 5.31e(-04) ,1.10e(-06) , 1.15e(-09))^\intercal\\
        (B) \quad &   d(w,x):=d(w,x,2,1)= 0.2 + e^{0.01x}-0.1(1-w)\\ 
         & g(X)=\left(1,X,X^2,X^3\right)^\intercal\\
          &\gamma_0=( -1.94 ,-1.00e(-02) ,-1.20e(-05) ,-4.59e(-09))^\intercal\\
        &\gamma_1=(0.96, 5.31e(-04) ,1.10e(-06) , 1.15e(-09))^\intercal.\\
    \end{aligned}
\end{equation}
We note that in both scenarios, the true difference in conditional potential outcome functions between the two groups $d(w,x)$ is Lipschitz continuous on $[-1000,-1]$, and it is even smoother than the approximately constant cutoff assumption (Assumption~\ref{ass:general:lipdiff}).

The policy class of interest is parameterized as 
$$\Piol=\left\{ \pi(x,g) = \mathds{1}(x\geq c_g) |\  c_g\in [-850, -571] \right\}.$$ 
Our method learns an optimal policy over the class $\Piol$ under Assumption~\ref{ass:general:lipdiff} based on Equation~\eqref{equ:multi:estor}.
We estimate all the nuisance components in Equation~\eqref{equ:general:multi:empDR} via nonparametric methods using the cross-fitting strategy.
Specifically, we use a local polynomial regression as implemented in the R package \texttt{nprobust} for $\tilde m(x,g)$ and a generalized linear model with logistic link function and a smoothing spline for $e_g(\cdot)$.
To construct the empirical point-wise upper and lower bounds, $ \widehat{B}_{\ell}(w,x,g,g^\prime)$ and $ \widehat{B}_{u}(w,x,g,g^\prime)$, we use our proposed two-stage doubly robust estimator in Appendix~\ref{appendix:add:DRlearner} to estimate
the unknown nuisance component ${\tilde{d}}(x,g,g^\prime)$.
Finally, we choose the Lipschitz constant $\lambda_{w12}$ through the procedure discussed in Section~\ref{sec:multi:smoothpara}.

\section{Proofs and Derivations}
To guide the proofs in the main paper, we first introduce several definitions and lemmas used throughout this section. Let $\Pi$ be the prespecified measurable policy class in the main text, with $\widetilde\pi\in\Pi\subseteq\Pi_{\mathrm S}$, joint VC dimension $\nu=VC(\Pi)$, and group-specific classes $\Pi_g=\{x\mapsto\pi(x,g):\pi\in\Pi\}$. We write $\overline\nu=\max_{g\in\mathcal Q}VC(\Pi_g)$. The threshold-policy result follows by setting $\Pi=\Piol$, for which $\nu=Q$ and $\overline\nu=1$.

\subsection{Definitions and lemmas for terms $\mathcal{I}_{iden}$ and term $\Xi_2$}\label{sec:pr:general:def:1}
\begin{definition}[Rademacher complexity]
Let $\mathcal{F}$ be a family of functions mapping $X \mapsto \mathbb{R}$. The population Rademacher complexity of $\mathcal{F}$ is defined as
$$\mathcal{R}_{n}(\mathcal{F}) \equiv \mathbb{E}_{X, \varepsilon}\left[\sup _{f \in \mathcal{F}}\left|\frac{1}{n} \sum_{i=1}^{n} \varepsilon_{i} f\left(X_{i}\right)\right|\right]$$
where $\varepsilon_{i}$'s are i.i.d. Rademacher random variables, i.e., $\operatorname{Pr}\left(\varepsilon_{i}=1\right)=\operatorname{Pr}\left(\varepsilon_{i}=-1\right)=1 / 2$, and the expectation is taken over both the Rademacher variables $\varepsilon_{i}$ and the i.i.d. covariates $X_{i}$. 
\end{definition}

\begin{definition}\label{def:J}
 Define the following population quantity,
\begin{equation*}
\begin{aligned}
      J(\pi)=&\mathbbm{E}\left[ \pi(X,G)\tilde\pi(X,G)Y + \left(1-\pi(X,G)\right)\left(1-\tilde\pi(X,G)\right)Y\right]\\
&+ \mathbbm{E}\left[\sum_{g=1}^{Q} \mathds{1}\left(G=g\right)\cdot \pi(X,g)\left(1-\tilde\pi(X,g)\right) \sum_{g^\prime=1}^{g-1} \mathds{1}\left(X\in[\tilde c_{g^\prime}, \tilde c_{g^\prime+1}]\right)  d(1,X,g,g^{\prime}) \right]\\
    &+ \mathbbm{E}\left[\sum_{g=1}^{Q} \mathds{1}\left(G=g\right)\cdot \left(1-\pi(X,g)\right)\tilde\pi(X,g) \sum_{g^\prime=g+1}^{Q} \mathds{1}\left(X\in[\tilde c_{g^\prime-1}, \tilde c_{g^\prime}]\right)   d(0,X,g,g^{\prime})\right],
\end{aligned}
\end{equation*}

We also define the following empirical quantity,
\begin{equation*}
\begin{aligned}
     \hat J(\pi)=\frac{1}{n}\sum_{i=1}^{n} & \ 
\pi(X_i,G_i)\tilde\pi(X_i,G_i)Y_i+\left\{1-\pi(X_i,G_i)\right\}\left\{1-\tilde\pi(X_i,G_i)\right\} Y_i \\
+& \sum_{g=1}^{Q} \mathds{1}\left(G_i=g\right)\cdot \pi(X_i,g)\left(1-\tilde\pi(X_i,g)\right) \sum_{g^\prime=1}^{g-1} \mathds{1}\left(X_i\in[\tilde c_{g^\prime}, \tilde c_{g^\prime+1}]\right)  d(1,X_i,g,g^{\prime}) \\
    + &  \sum_{g=1}^{Q} \mathds{1}\left(G_i=g\right)\cdot \left(1-\pi(X_i,g)\right)\tilde\pi(X_i,g) \sum_{g^\prime=g+1}^{Q} \mathds{1}\left(X_i\in[\tilde c_{g^\prime-1}, \tilde c_{g^\prime}]\right)   d(0,X_i,g,g^{\prime}),
\end{aligned}
\end{equation*}
supposing the true counterfactual difference function $d(w,x,g,g^{\prime})$ were known.
\end{definition}

\begin{lemma}\label{lemma:general:T1}
Under Assumption~\ref{ass:gen:bo}, let $\Pi$ have finite VC dimension $\nu$. Then, for any $\delta>0$,
$$
\sup _{\pi \in \Pi}|\hat{J}(\pi)-J(\pi)| \leq \frac{ \Lambda(2+3u\sqrt{\nu})}{\sqrt{n}}+\delta
$$
with probability at least $1-\exp \left(-\frac{n \delta^{2}}{8 \Lambda^{2}}\right)$, where $u$ is a universal constant.
\end{lemma}

\begin{proof}[Proof of Lemma \ref{lemma:general:T1}]
  Define the function class with functions $f(x,y,g)$ in 
\begin{equation*}
\begin{aligned}
\mathcal{F}=  & \Biggl\{ \pi(X_i,G_i)\tilde\pi(X_i,G_i)Y_i+\left\{1-\pi(X_i,G_i)\right\}\left\{1-\tilde\pi(X_i,G_i)\right\} Y_i \\
     & \left. + \pi(X_i,G_i)\left(1-\tilde\pi(X_i,G_i)\right) \sum_{g^\prime=1}^{G_i-1} \mathds{1}\left(X_i\in[\tilde c_{g^\prime}, \tilde c_{g^\prime+1}]\right)  d(1,X_i,G_i,g^{\prime}) \right. \\
     &  +  \left(1-\pi(X_i,G_i)\right)\tilde\pi(X_i,G_i) \sum_{g^\prime=G_i+1}^{Q} \mathds{1}\left(X_i\in[\tilde c_{g^\prime-1}, \tilde c_{g^\prime}]\right)   d(0,X_i,G_i,g^{\prime}) \mid \pi \in \Pi \Biggr\}.
\end{aligned}
\end{equation*}
Notice that
$$
\sup _{\pi \in \Pi}|\hat{J}(\pi)-J(\pi)|=\sup _{f \in \mathcal{F}}\left|\frac{1}{n} \sum_{i=1}^{n} f\left(X_{i}, Y_{i},G_i\right)-\mathbb{E}[f(X, Y,G)]\right|
$$
The class $\mathcal{F}$ is uniformly bounded by $2\Lambda$ because the outcomes and the difference functions are bounded, so by Theorem~$4.10$ in \cite{wainwright2019high}
$$
\sup_{f \in \mathcal{F}} \left|\frac{1}{n} \sum^{n}_{i=1} f\left(X_{i}, Y_{i},G_i\right)-\mathbb{E}[f(X, Y,G)]\right| \leq 2 \mathcal{R}_{n}(\mathcal{F})+\delta,
$$
with probability at least $1-\exp \left(-\frac{n \delta^{2}}{8 \Lambda^{2}}\right)$.  
Finally, the Rademacher complexity for $\mathcal{F}$ is bounded by
\begin{align*}
& \mathcal{R}_{n}(\mathcal{F}) \\
\leq \ & \mathbb{E}_{X, Y, G, \epsilon}\left[\left| \frac{1}{n} \sum_{i=1}^{n}\left\{\left(1-\tilde{\pi}\left(X_{i},G_i\right)\right) Y_{i}+\tilde{\pi}\left(X_{i},G_i\right) \sum_{g^\prime=G_i+1}^{Q} \mathds{1}\left(X_i\in[\tilde c_{g^\prime-1}, \tilde c_{g^\prime}]\right)   d(0,X_i,G_i,g^{\prime})\right\}\epsilon_{i}\right|\right] \\
& + \sup _{\pi \in \Pi} \mathbb{E}_{X, Y, G, \varepsilon}\left[\left| \frac{1}{n} \sum_{i=1}^{n}\left\{\tilde{\pi}\left(X_{i}\right)\left(Y_{i}-\sum_{g^\prime=G_i+1}^{Q} \mathds{1}\left(X_i\in[\tilde c_{g^\prime-1}, \tilde c_{g^\prime}]\right)   d(0,X_i,G_i,g^{\prime})\right)\right. \right.\right.\\
& \hspace{80pt} \left.\left.\left.+\left(1-\tilde{\pi}\left(X_{i}\right)\right)\left(\sum_{g^\prime=1}^{G_i-1} \mathds{1}\left(X_i\in[\tilde c_{g^\prime}, \tilde c_{g^\prime+1}]\right)  d(1,X_i,G_i,g^{\prime})-Y_{i}\right)\right\} \pi(X_{i},G_i) \epsilon_{i}\right|\right] \\
\ \leq \ & \frac{2\Lambda}{n} \mathbb{E}_{\varepsilon}\left[\left|\sum_{i=1}^{n} \varepsilon_{i}\right|\right]
+\sup _{\pi \in \Pi} 3\Lambda \ \mathbb{E}_{X, Y, G, \varepsilon}\left[\left|\frac{1}{n} \sum_{i=1}^{n} \pi(X_{i},G_i) \varepsilon_{i}\right|\right] \\
\ \leq \ & \frac{2\Lambda}{\sqrt{n}}
+3 \Lambda \mathcal{R}_{n}(\Pi) \\
\ \leq \ & \frac{2\Lambda}{\sqrt{n}}
+\frac{3\Lambda u\sqrt{\nu}}{\sqrt{n}},
\end{align*}
where the constant term in the second inequality is due to $|Y_i |\leq \Lambda$ and $|d(\cdot,\cdot,\cdot,\cdot)|\leq 2\Lambda$,
and the last inequality follows from the fact that for a function class $\mathcal{G}$ with finite VC dimension $\nu<\infty$, the Rademacher complexity is bounded by, $\mathcal{R}_{n}(\mathcal{G}) \leq u \sqrt{\nu / n}$, for some universal constant $u$ \citep[][$\S 5.3$]{wainwright2019high}. 
\end{proof}

\begin{lemma}[Bounded estimation error]\label{lemma:multi:noise}
Under Assumption~\ref{ass:bdrate},
\begin{equation*}
     \sup_{\pi\in\Pi}\left|\min_{d \in \widehat{\mathcal{M}}_{d}} \hat{J}(\pi,d)-\min_{d \in {\mathcal{M}_d}}\hat{J}(\pi,d)\right|= O_p \left(\rho_n^{-1}\right).
\end{equation*}
\end{lemma}

\begin{proof}[Proof of Lemma \ref{lemma:multi:noise}]
  \begin{equation*}
  \begin{aligned}
       & \left|\min_{d \in \widehat{\mathcal{M}}_{d}} \hat{J}(\pi,d)-\min_{d \in {\mathcal{M}_d}}\hat{J}(\pi,d)\right| \\
        =& \ \left| \frac{1}{n}\sum_{i=1}^{n}  \pi(X_i,G_i)\left(1-\tilde\pi(X_i,G_i)\right) \right.\\
        & \qquad \times \left. \sum_{g^\prime=1}^{G_i-1} \mathds{1}\left(X_i\in[\tilde c_{g^\prime}, \tilde c_{g^\prime+1}]\right)  \left[\min_{d(1,\cdot,\cdot,g^\prime) \in \widehat{\mathcal{M}}_{d}} d(1,X_i,G_i,g^{\prime})-\min_{d(1,\cdot,\cdot,g^\prime) \in {\mathcal{M}}_{d}} d(1,X_i,G_i,g^{\prime})\right]\right. \\
        & +
    \left(1-\pi(X_i,G_i)\right)\tilde\pi(X_i,G_i)\\
    & \qquad \times \left. \sum_{g^\prime=G_i+1}^{Q} \mathds{1}\left(X_i\in[\tilde c_{g^\prime-1}, \tilde c_{g^\prime}]\right)   \left[\min_{d(0,\cdot,\cdot,g^\prime) \in \widehat{\mathcal{M}}_{d}} d(0,X_i,G_i,g^{\prime})-\min_{d(0,\cdot,\cdot,g^\prime) \in {\mathcal{M}}_{d}} d(0,X_i,G_i,g^{\prime})\right]  \right|\\
    =& \ \left| \frac{1}{n}\sum_{i=1}^{n}  \pi(X_i,G_i)\left(1-\tilde\pi(X_i,G_i)\right) \right. \times \left. \sum_{g^\prime=1}^{G_i-1} \mathds{1}\left(X_i\in[\tilde c_{g^\prime}, \tilde c_{g^\prime+1}]\right)  \left[\hat{B}_{\ell}(1,X_i,G_i,g^\prime)-{B}_{\ell}(1,X_i,G_i,g^\prime)\right]\right. \\
        & +
    \left(1-\pi(X_i,G_i)\right)\tilde\pi(X_i,G_i) \times \left. \sum_{g^\prime=G_i+1}^{Q} \mathds{1}\left(X_i\in[\tilde c_{g^\prime-1}, \tilde c_{g^\prime}]\right)   \left[\hat{B}_{\ell}(0,X_i,G_i,g^\prime)-{B}_{\ell}(0,X_i,G_i,g^\prime)\right]  \right|\\
    \leq & \  \frac{1}{n}\sum_{i=1}^{n}\  \max _{w \in \{0,1\},\ g,g^\prime\in \mathcal{Q}} \left | \hat{B}_{\ell}(w,X_i,g,g^\prime)-{B}_{\ell}(w,X_i,g,g^\prime)\right|.
  \end{aligned}
\end{equation*}
The final bound does not depend on $\pi$. Taking the supremum over $\pi\in\Pi$ and applying Equation~\eqref{equ:B&d} and Assumption~\ref{ass:bdrate} therefore gives the stated $O_p(\rho_n^{-1})$ result directly.
\end{proof}

\subsection{Definitions and lemmas for $\Xi_{DR}$}\label{sec:pr:general:def:2}

We now turn our attention to the doubly-robust construction in Equation~\eqref{equ:general:multi:db}.
Our derivation is based upon the proof strategy in \cite{zhou2018offline}. 
Throughout this section, $\Pi$ and $\Pi_g$ retain the general definitions given above.

\begin{definition}\label{def:K}
Corresponding to Equation~\eqref{equ:general:multi:db}, we redefine the term $\Xi_{\text{DR}}$ as the following population quantity, explicitly under a given policy $\pi$,
{\footnotesize
\begin{align*}
    K(\pi)
    = & 
    \mathbbm{E}\Bigg[\sum_{g=1}^{Q} \pi(X_i,g)\left(1-\tilde\pi(X_i,g)\right) \cdot
    \\
    & \underbrace{\Bigg\{ \mathds{1}\left(G_i=g\right)\cdot \sum_{g^\prime=1}^{g-1} \mathds{1}\left(X_i\in[\tilde c_{g^\prime},\tilde c_{g^\prime+1}]\right)\tilde m(X_i,g^{\prime}) + 
    \sum_{g^\prime=1}^{g-1}\mathds{1}\left(G_i=g^{\prime}\right) \mathds{1}\left(X_i\in[\tilde c_{g^\prime},\tilde c_{g^\prime+1}]\right) \frac{e_g(X_i)}{e_{g^{\prime}}(X_i)}(Y_i-\tilde m(X_i,g^{\prime})) \Bigg\}}_{:= {\Gamma}^{(1)}_{ig}} \Bigg]
    \\
    & + \mathbbm{E}\Bigg[\sum_{g=1}^{Q} \left(1-\pi(X_i,g)\right)\tilde\pi(X_i,g) \cdot
    \\
    & \underbrace{\Bigg\{ \mathds{1}\left(G_i=g\right)\sum_{g^\prime=g+1}^{Q} \mathds{1}\left(X_i\in[\tilde c_{g^\prime-1}, \tilde c_{g^\prime}]\right)\tilde m(X_i,g^{\prime}) + 
    \sum_{g^\prime=g+1}^{Q} \mathds{1}\left(G_i=g^{\prime}\right) \mathds{1}\left(X_i\in[\tilde c_{g^\prime-1}, \tilde c_{g^\prime}]\right)\frac{e_g(X_i)}{e_{g^{\prime}}(X_i)}(Y_i-\tilde m(X_i,g^{\prime})) \Bigg\}}_{:= {\Gamma}^{(0)}_{ig}} \Bigg],
\end{align*}
}
where $ {\Gamma}^{(1)}_{ig} $ and $ {\Gamma}^{(0)}_{ig} $ is defined as the true double-robust score for group $g$ when we extrapolate the current cutoff to the left and to the right, respectively.

To simply the notation, we further define
$$
\begin{aligned}
 {K}(\pi) &= \ \sum_{g=1}^{Q} \underbrace{\mathbbm{E}\Big[
 \pi(X_i,g)\left(1-\tilde\pi(X_i,g)\right) 
   {\Gamma}^{(1)}_{ig} \Big]}_{:= {K}_1(\pi,g) } +  \sum_{g=1}^{Q}\underbrace{\mathbbm{E}\Big[
 \left(1-\pi(X_i,g)\right)\tilde\pi(X_i,g) 
 {\Gamma}^{(0)}_{ig} \Big]}_{:= {K}_0(\pi,g) } .
\end{aligned}
$$

We also redefine the empirical quantity  $\widehat\Xi_{\text{DR}}$ in Equation~\eqref{equ:general:multi:empDR}, 
$$
\begin{aligned}
 \hat{K}(\pi) &= \ \sum_{g=1}^{Q}\underbrace{\frac{1}{n}\sum_{i=1}^{n}
 \pi(X_i,g)\left(1-\tilde\pi(X_i,g)\right) 
  \hat {\Gamma}^{(1)}_{ig} }_{:= \hat{K}_1(\pi,g) } +  \sum_{g=1}^{Q}\underbrace{\frac{1}{n}\sum_{i=1}^{n}
 \left(1-\pi(X_i,g)\right)\tilde\pi(X_i,g) 
  \hat {\Gamma}^{(0)}_{ig} }_{:= \hat{K}_0(\pi,g) } ,
\end{aligned}
$$
where $\hat {\Gamma}^{(1)}_{ig} $ and $\hat {\Gamma}^{(0)}_{ig} $ are the empirical double robust scores using the cross-fitted nuisance estimates constructed from the $K$ folds of the data:

\begin{align*}
    \hat{\Gamma}^{(1)}_{ig }
   \  = \ & \mathds{1}\left(G_i=g\right) \sum_{g^\prime=1}^{g-1}  \mathds{1}\left(X_i\in[\tilde c_{g^\prime},\tilde c_{g^\prime+1}]\right)\tilde m^{-k[i]}(X_i,g^{\prime}) \\
   &  \hspace{.5in} +    \sum_{g^\prime=1}^{g-1}\mathds{1}\left(G_i=g^{\prime}\right) \mathds{1}\left(X_i\in[\tilde c_{g^\prime},\tilde c_{g^\prime+1}]\right)  \frac{e^{-k[i]}_g(X_i)}{e^{-k[i]}_{g^{\prime}}(X_i)}(Y_i-\tilde m^{-k[i]}(X_i,g^{\prime}))\\
    \hat{\Gamma}^{(0)}_{ig } \ = \ &  \mathds{1}\left(G_i=g\right)\sum_{g^\prime=g+1}^{Q}  \mathds{1}\left(X_i\in[\tilde c_{g^\prime-1}, \tilde c_{g^\prime}]\right)\tilde m^{-k[i]}(X_i,g^{\prime})\\
    & \hspace{.5in} + \sum_{g^\prime=g+1}^{Q} \mathds{1}\left(G_i=g^{\prime}\right) \mathds{1}\left(X_i\in[\tilde c_{g^\prime-1}, \tilde c_{g^\prime}]\right)\frac{e^{-k[i]}_g(X_i)}{e^{-k[i]}_{g^{\prime}}(X_i)}(Y_i-\tilde m^{-k[i]}(X_i,g^{\prime}))
\end{align*}
Finally, we define the oracle (infeasible) estimator with the estimated $\hat {\Gamma}^{(1)}_{ig} $ and $\hat {\Gamma}^{(0)}_{ig} $ in $\hat{K}(\pi)$  substituted by the true value $ {\Gamma}^{(1)}_{ig} $ and $ {\Gamma}^{(0)}_{ig} $,
$$
\begin{aligned}
 \tilde{K}(\pi) &= \ \sum_{g=1}^{Q} \underbrace{\frac{1}{n}\sum_{i=1}^{n}
 \pi(X_i,g)\left(1-\tilde\pi(X_i,g)\right) 
 {\Gamma}^{(1)}_{ig} }_{:=  \tilde{K}_1(\pi,g) } +  \sum_{g=1}^{Q}\underbrace{\frac{1}{n}\sum_{i=1}^{n}
 \left(1-\pi(X_i,g)\right)\tilde\pi(X_i,g) 
  {\Gamma}^{(0)}_{ig} }_{:=  \tilde{K}_0(\pi,g) } ,
\end{aligned}
$$
It can be easily shown that $\mathbbm{E}\left[\tilde{K}(\pi)\right]={K}(\pi)$.

\end{definition}

To proceed with the proof, we adopt several useful definitions from \cite{zhou2018offline}.
\begin{definition}
Given the feature domain $\mathcal{X}$, a policy class $\Pi$, a set of $n$ points $\left\{x_{1}, \ldots, x_{n}\right\} \subset$ $\mathcal{X}$, define:
\begin{enumerate}
    \item Hamming distance between any two policies $\pi_a$ and $\pi_b$ in $\Pi: H\left(\pi_a, \pi_b\right)=\frac{1}{n} \sum_{j=1}^{n} \mathds{1}\left(\pi_a\left(x_{j}\right) \neq \pi_b\left(x_{j}\right)\right)$.
    \item $\epsilon$-Hamming covering number of the set $\left\{x_{1}, \ldots, x_{n}\right\}$ :
$N_{H}\left(\epsilon, \Pi,\left\{x_{1}, \ldots, x_{n}\right\}\right)$ is the smallest number $K$ of policies $\left\{\pi_1, \ldots, \pi_{K}\right\}$ in $\Pi$, such that $\forall \pi \in$ $\Pi, \exists \pi_{i}, H\left(\pi, \pi_{i}\right) \leq \epsilon .$
   \item $\epsilon$-Hamming covering number of $\Pi: N_{H}(\epsilon, \Pi)=\sup \left\{N_{H}\left(\epsilon, \Pi,\left\{x_{1}, \ldots, x_{m}\right\}\right) \mid m \geq 1, x_{1}, \ldots, x_{m} \in\right.$ $\mathcal{X}\} .$
   \item Entropy integral: $\kappa(\Pi)=\int_{0}^{1} \sqrt{\log N_{H}\left(\epsilon^{2}, \Pi\right)} d \epsilon$.
\end{enumerate}
\end{definition}

\begin{definition}\label{def:delta}
 The influence difference function $\tilde{\Delta}(\cdot , \cdot): \Pi \times \Pi \rightarrow \mathbf{R}$ and the policy value difference function $\Delta(\cdot , \cdot): \Pi \times \Pi \rightarrow \mathbf{R}$ are defined respectively as follows: $$
 \begin{aligned}
     \tilde{\Delta}\left(\pi_a, \pi_b\right)&:=\tilde{K}(\pi_a)-\tilde{K}(\pi_b)\\
     & = \sum_{g=1}^{Q} \underbrace{\tilde{K_1}(\pi_a,g)-\tilde{K_1}(\pi_b,g)}_{:=\tilde{\Delta}_1\left(\pi_a(\cdot,g), \pi_b(\cdot,g)\right)} +\sum_{g=1}^{Q}\underbrace{\tilde{K_0}(\pi_a,g)-\tilde{K_0}(\pi_b,g)}_{:=\tilde{\Delta}_0\left(\pi_a(\cdot,g), \pi_b(\cdot,g)\right)}\\
     \Delta\left(\pi_a,\pi_b\right)&:=K(\pi_a)-K(\pi_b)\\
     &=\sum_{g=1}^{Q} \underbrace{{K_1}(\pi_a,g)-{K_1}(\pi_b,g)}_{:={\Delta}_1\left(\pi_a(\cdot,g), \pi_b(\cdot,g)\right)}+\sum_{g=1}^{Q}\underbrace{{K_0}(\pi_a,g)-{K_0}(\pi_b,g)}_{:={\Delta}_0\left(\pi_a(\cdot,g), \pi_b(\cdot,g)\right)}
 \end{aligned}$$
\end{definition}

\begin{lemma}\label{lemma:oracle}
Under Assumptions~\ref{ass:gen:bo} and~\ref{ass:overlap:general}, suppose that $\kappa(\Pi_g)<\infty$ for every $g\in\mathcal Q$. For $\delta>0$, define
\[
A_{g,\Pi}(\delta)=54.4\sqrt{2}\,\kappa(\Pi_g)+435.2+\sqrt{2\log\frac{2Q}{\delta}}
\]
and let
\[
V^\ast_{g,\Pi}=  \sup _{f_a,f_b \in \Pi_g}\mathbbm{E}\left[  (f_a(X_i)-f_b(X_i))^2\cdot (\left(1-\tilde\pi(X_i,g)\right)
( {\Gamma}^{(1)}_{ig} )^2 + \tilde\pi(X_i,g)
(  {\Gamma}^{(0)}_{ig}  )^2)\right]
\]
denote the worst-case variance of evaluating the difference between two policies in $\Pi$ for group $g$. Then, with probability at least $1-\delta$,
$$
\sup _{\pi_a, \pi_b \in \Pi}\left|\tilde{\Delta}\left(\pi_a, \pi_b\right)-\Delta\left(\pi_a, \pi_b\right)\right| \leq
\sum_{g=1}^{Q}A_{g,\Pi}(\delta)\sqrt{\frac{V^\ast_{g,\Pi}}{n}}
+o\left(\frac{1}{\sqrt n}\right).
$$
In particular, if $\overline\nu=\max_g VC(\Pi_g)<\infty$, then $\kappa(\Pi_g)\lesssim\sqrt{\overline\nu}$ and the right-hand side is
\[
O_p\left(\frac{Q\sqrt{1+\overline\nu}\sqrt{\max_{g\in\mathcal Q}V^\ast_{g,\Pi}}}{\sqrt n}\right).
\]

\end{lemma}

\begin{proof}
Recall that $\Pi_g=\{x\mapsto\pi(x,g):\pi\in\Pi\}$. Following Definitions~\ref{def:K} and~\ref{def:delta}, we can decompose the process into
\[
\begin{aligned}
    \tilde{\Delta}\left(\pi_a, \pi_b\right)
    &=\sum_{g=1}^{Q}
\tilde{K_1}(\pi_a,g)+\tilde{K_0}(\pi_a,g)-\tilde{K_1}(\pi_b,g) -\tilde{K_0}(\pi_b,g)
\end{aligned}
\]
and 
\[{\Delta}\left(\pi_a, \pi_b\right)=
\sum_{g=1}^{Q} {K_1}(\pi_a,g)+{K_0}(\pi_a,g)-{K_1}(\pi_b,g) -{K_0}(\pi_b,g).
\]
 We notice that
$$
\begin{aligned}
    \sup _{\pi_a, \pi_b \in \Pi}\left|\tilde{\Delta}\left(\pi_a, \pi_b\right)-\Delta\left(\pi_a, \pi_b\right)\right| \leq & 
    \sum_{g=1}^{Q} \sup _{\pi_a(\cdot,g), \pi_b(\cdot,g) \in \Pi_g}\left|\tilde{\Delta}\left(\pi_a(\cdot,g), \pi_b(\cdot,g)\right)-\Delta\left(\pi_a(\cdot,g), \pi_b(\cdot,g)\right)\right| \\
\end{aligned}
$$
So it suffices to bound 
\begin{equation}\label{equ:delta1}
   \sup _{\pi_a(\cdot,g), \pi_b(\cdot,g) \in \Pi_g}\left|\tilde{\Delta}\left(\pi_a(\cdot,g), \pi_b(\cdot,g)\right)-\Delta\left(\pi_a(\cdot,g), \pi_b(\cdot,g)\right)\right| 
\end{equation}
for each group $g$.

By Assumptions~\ref{ass:gen:bo}~and~\ref{ass:overlap:general}, $\left\{({\Gamma}^1_{i,g},{\Gamma}^0_{i,g})\right\}_{i=1}^{n}$ are i.i.d. random vectors with bounded support. Note that
\[
\tilde{\Delta}\left(\pi_a(\cdot,g), \pi_b(\cdot,g)\right)= \frac{1}{n}  \sum_{i=1}^{n}
\left\langle \begin{bmatrix}
            \pi_a(X_i,g)- \pi_b(X_i,g) \\
        1- \pi_a(X_i,g)-( 1- \pi_b(X_i,g))
         \end{bmatrix} , 
          \begin{bmatrix}
            \left(1-\tilde\pi(X_i,g)\right) 
 {\Gamma}^{(1)}_{ig}  \\
  \tilde\pi(X_i,g) 
  {\Gamma}^{(0)}_{ig}   \end{bmatrix}  \right\rangle
 \]
Thus, we can apply Lemma~2 of \cite{zhou2018offline} by specializing it to the 2-dimensional case to Equations~\eqref{equ:delta1} for each group $g$ and use the union bound to obtain: 
$\forall\delta>0$, with probability at least $1-2 Q\delta$,
    \begin{align*}
         &  \sum_{g=1}^{Q} \sup _{\pi_a(\cdot,g), \pi_b(\cdot,g) \in \Pi_g}\left|\tilde{\Delta}\left(\pi_a(\cdot,g), \pi_b(\cdot,g)\right)-\Delta\left(\pi_a(\cdot,g), \pi_b(\cdot,g)\right)\right| \\
         \leq & 
         \sum_{g=1}^{Q}\left(54.4 \sqrt{2} \kappa(\Pi_g)+435.2+\sqrt{2 \log \frac{1}{\delta}}\right)\cdot \\
        & \quad \sqrt{\frac{1}{n}\sup_{\pi_a(\cdot,g), \pi_b(\cdot,g) \in \Pi_g} \mathbbm{E}\left[\left\langle \begin{bmatrix}
            \pi_a(X_i,g)- \pi_b(X_i,g) \\
        1- \pi_a(X_i,g)-( 1- \pi_b(X_i,g))
         \end{bmatrix} , 
          \begin{bmatrix}
            \left(1-\tilde\pi(X_i,g)\right) 
 {\Gamma}^{(1)}_{ig}  \\
  \tilde\pi(X_i,g) 
  {\Gamma}^{(0)}_{ig}   \end{bmatrix}  \right\rangle^{2}\right]}
         +o\left(\frac{1}{\sqrt{n}}\right)\\
         \leq & 
         \sum_{g=1}^{Q}\left(54.4 \sqrt{2} \kappa(\Pi_g)+435.2+\sqrt{2 \log \frac{1}{\delta}}\right)\cdot \\
        & \qquad \sqrt{
        \frac{1}{n}
        \sup_{\pi_a(\cdot,g), \pi_b(\cdot,g) \in \Pi_g} 
        \mathbbm{E}\left[\left\langle \begin{bmatrix}
            \pi_a(X_i,g)- \pi_b(X_i,g) \\
        1- \pi_a(X_i,g)-( 1- \pi_b(X_i,g))
         \end{bmatrix} , 
          \begin{bmatrix}
            \left(1-\tilde\pi(X_i,g)\right) 
 {\Gamma}^{(1)}_{ig}  \\
  \tilde\pi(X_i,g) 
  {\Gamma}^{(0)}_{ig}   \end{bmatrix}  \right\rangle^{2}\right]}
         +o\left(\frac{1}{\sqrt{n}}\right)\\
        = & 
         \sum_{g=1}^{Q}\left(54.4 \sqrt{2} \kappa(\Pi_g)+435.2+
         \sqrt{2 \log \frac{1}{\delta}}\right)\cdot 
        \\
         & \qquad  \sqrt{ \frac{1}{n} 
         \sup_{\pi_a(\cdot,g), \pi_b(\cdot,g) \in \Pi_g}\mathbbm{E}\left[  (\pi_a(X_i,g)- \pi_b(X_i,g))^2\cdot (\left(1-\tilde\pi(X_i,g)\right)
( {\Gamma}^{(1)}_{ig} )^2 + \tilde\pi(X_i,g) 
(  {\Gamma}^{(0)}_{ig}  )^2)\right]
}+o\left(\frac{1}{\sqrt{n}}\right)
    \end{align*}
where the second to last equality holds because
$$
\begin{aligned}
& \Big\langle \begin{bmatrix}
            \pi_a(X_i,g)- \pi_b(X_i,g) \\
        1- \pi_a(X_i,g)-( 1- \pi_b(X_i,g))
         \end{bmatrix} , 
          \begin{bmatrix}
            \left(1-\tilde\pi(X_i,g)\right) 
 {\Gamma}^{(1)}_{ig}  \\
  \tilde\pi(X_i,g) 
  {\Gamma}^{(0)}_{ig}   \end{bmatrix}  \Big\rangle^{2} \\
  =& (\pi_a(X_i,g)- \pi_b(X_i,g))^2(  \left(1-\tilde\pi(X_i,g)\right) 
 {\Gamma}^{(1)}_{ig}  - \tilde\pi(X_i,g) 
  {\Gamma}^{(0)}_{ig} )^2\\
  = & (\pi_a(X_i,g)- \pi_b(X_i,g))^2 (\left(1-\tilde\pi(X_i,g)\right)
( {\Gamma}^{(1)}_{ig} )^2 + \tilde\pi(X_i,g) 
(  {\Gamma}^{(0)}_{ig}  )^2)
 \end{aligned}
$$
By the definition
\[
V^\ast_{g,\Pi}=  \sup _{\pi_a(\cdot,g), \pi_b(\cdot,g) \in \Pi_g}\mathbbm{E}\left[  (\pi_a(X_i,g)- \pi_b(X_i,g))^2\cdot (\left(1-\tilde\pi(X_i,g)\right)
( {\Gamma}^{(1)}_{ig} )^2 + \tilde\pi(X_i,g) 
(  {\Gamma}^{(0)}_{ig}  )^2)\right]
\],
the last display is the sum in the statement with $\sqrt{2\log(1/\delta)}$ in each summand. Replacing $\delta$ by $\delta/(2Q)$ gives the asserted probability and logarithmic term. Finally, Theorem~1 of \cite{haussler1995sphere} implies that the Hamming covering number of a binary VC class satisfies
$N_H(\epsilon,\Pi_g)\leq e\{VC(\Pi_g)+1\}(2e/\epsilon)^{VC(\Pi_g)}$.
It follows that $\kappa(\Pi_g)\leq 2.5\sqrt{VC(\Pi_g)}\leq 2.5\sqrt{\overline\nu}$, which gives the stated finite-VC order.

\end{proof}

\begin{definition}
 Define $\hat{\Delta}(\cdot , \cdot): \Pi \times \Pi \rightarrow \mathbf{R}$ with 
 $$
  \begin{aligned}
    \hat{\Delta}\left(\pi_a, \pi_b\right) &=\hat K(\pi_a)-\hat K(\pi_b)\\
     & = \sum_{g=1}^{Q} \underbrace{\hat{K_1}(\pi_a,g)-\hat{K_1}(\pi_b,g)}_{:=\hat{\Delta}_1\left(\pi_a(\cdot,g), \pi_b(\cdot,g)\right)} +\sum_{g=1}^{Q}\underbrace{\hat{K_0}(\pi_a,g)-\hat{K_0}(\pi_b,g)}_{:=\hat{\Delta}_0\left(\pi_a(\cdot,g), \pi_b(\cdot,g)\right)}\\
 \end{aligned}$$
\end{definition}

\begin{lemma}\label{lemma:feasible}
Assume the number of cross-fitting folds $K>2$ is finite and $\kappa(\Pi_g)<\infty$ for every $g\in\mathcal Q$. Under Assumptions~\ref{ass:gen:bo},~\ref{ass:overlap:general}~and~\ref{ass:rate:general},
$$\sup _{\pi_a, \pi_b \in \Pi}\left|\hat{\Delta}\left(\pi_a, \pi_b\right)-\tilde{\Delta}\left(\pi_a, \pi_b\right)\right|=o_{p}\left(\frac{1}{\sqrt{n}}\right)$$.
\end{lemma}

\begin{proof}
  Taking any two policies $\pi_a,\pi_b \in \Pi$, it suffices to prove the following result:
\begin{align}
    \sup _{\pi_a(\cdot,g), \pi_b(\cdot,g) \in \Pi_g}\left|\hat{\Delta}_1\left(\pi_a(\cdot,g), \pi_b(\cdot,g)\right)-\tilde\Delta_1\left(\pi_a(\cdot,g), \pi_b(\cdot,g)\right)\right| & =o_{p}\left(\frac{1}{\sqrt{n}}\right) \label{equ:deltahat1}\\
 \sup _{\pi_a(\cdot,g), \pi_b(\cdot,g) \in \Pi_g}\left|\hat{\Delta}_0\left(\pi_a(\cdot,g), \pi_b(\cdot,g)\right)-\tilde\Delta_0\left(\pi_a(\cdot,g), \pi_b(\cdot,g)\right)\right|& =o_{p}\left(\frac{1}{\sqrt{n}}\right) \label{equ:deltahat0}
\end{align}
Then, we only need to prove Equation~\eqref{equ:deltahat1}.  Equation~\eqref{equ:deltahat0} can be proven using the same argument.
We have the following decomposition:
\begin{align*}
           & \hat{\Delta}_1\left(\pi_a(\cdot,g), \pi_b(\cdot,g)\right)-\tilde\Delta_1\left(\pi_a(\cdot,g), \pi_b(\cdot,g)\right) \\
           = &  \frac{1}{n} \sum_{i=1}^{n} \left(\pi_a(X_i,g)-\pi_b(X_i,g)\right)\{1-\tildepi(X_i,g)\} \sum_{g^\prime=1}^{g-1}  \mathds{1}\left(X_i\in[\tilde c_{g^\prime},\tilde c_{g^\prime+1}]\right) \left(\hat{\tilde m}^{-k[i]}(X_i,g^\prime)-{\tilde m}(X_i,g^\prime)\right)\cdot \\
           & \hspace{3.5in} \left(\mathds{1}\left(G_i=g\right)-\mathds{1}\left(G_i=g^{\prime}\right)\frac{{e}_g(X_i)}{{e}_{g^{\prime}}(X_i)}\right)
           \\
           & +  \frac{1}{n} \sum_{i=1}^{n} \left(\pi_a(X_i,g)-\pi_b(X_i,g)\right)\{1-\tildepi(X_i,g)\} \sum_{g^\prime=1}^{g-1}  \mathds{1}\left(X_i\in[\tilde c_{g^\prime},\tilde c_{g^\prime+1}]\right) \mathds{1}\left(G_i=g^{\prime}\right)\\
           & \hspace{3.5in}\left(\frac{\hat{e}^{-k[i]}_g(X_i)}{\hat{e}^{-k[i]}_{g^{\prime}}(X_i)}-\frac{{e}_g(X_i)}{{e}_{g^{\prime}}(X_i)}\right) \left(Y_i-{\tilde m}(X_i,g^\prime) \right)\\
           & +  \frac{1}{n} \sum_{i=1}^{n} \left(\pi_a(X_i,g)-\pi_b(X_i,g)\right)\{1-\tildepi(X_i,g)\}\times \sum_{g^\prime=1}^{g-1}  \mathds{1}\left(X_i\in[\tilde c_{g^\prime},\tilde c_{g^\prime+1}]\right) \mathds{1}\left(G_i=g^{\prime}\right)\\
           & \hspace{3.5in} \left(\frac{\hat{e}^{-k[i]}_g(X_i)}{\hat{e}^{-k[i]}_{g^{\prime}}(X_i)}-\frac{{e}_g(X_i)}{{e}_{g^{\prime}}(X_i)}\right)\left({\tilde m}(X_i,g^\prime)-\hat{\tilde m}^{-k[i]}(X_i,g^\prime)\right)\\
            := & S_1\left(\pi_a, \pi_b\right) + S_2\left(\pi_a, \pi_b\right) + S_3\left(\pi_a, \pi_b\right).
      \end{align*}
We will bound each of these three terms.
Define the following:
\begin{align*}
       S_1^k \left(\pi_a, \pi_b\right) & =  \frac{1}{n} \sum_{i\in \mathcal{I}_k}  \left(\pi_a(X_i,g)-\pi_b(X_i,g)\right)\{1-\tildepi(X_i,g)\}\cdot\\
      & \quad\qquad \sum_{g^\prime=1}^{g-1}  \mathds{1}\left(X_i\in[\tilde c_{g^\prime},\tilde c_{g^\prime+1}]\right) \left(\hat{\tilde m}^{-k[i]}(X_i,g^\prime)-{\tilde m}(X_i,g^\prime)\right)\cdot \left(\mathds{1}\left(G_i=g\right)-\mathds{1}\left(G_i=g^{\prime}\right)\frac{{e}_g(X_i)}{{e}_{g^{\prime}}(X_i)}\right) \\
          S_2^k \left(\pi_a, \pi_b\right) & = \frac{1}{n} \sum_{i\in \mathcal{I}_k} \left(\pi_a(X_i,g)-\pi_b(X_i,g)\right)\{1-\tildepi(X_i,g)\}\cdot\\ 
          & \quad\qquad \sum_{g^\prime=1}^{g-1}  \mathds{1}\left(X_i\in[\tilde c_{g^\prime},\tilde c_{g^\prime+1}]\right) \mathds{1}\left(G_i=g^{\prime}\right)\left(\frac{\hat{e}^{-k[i]}_g(X_i)}{\hat{e}^{-k[i]}_{g^{\prime}}(X_i)}-\frac{{e}_g(X_i)}{{e}_{g^{\prime}}(X_i)}\right) \left(Y_i-{\tilde m}(X_i,g^\prime) \right),
    \end{align*}
where $ \mathcal{I}_k =\left\{i\mid k[i]=k\right\}$ denotes the $k$-th fold data and $ \mathcal{I}_{-k} =\left\{i\mid k[i]\neq k\right\}$ represents the data of the remaining $K-1$ folds. Note that $S_1\left(\pi_a, \pi_b\right)=\sum_{k=1}^{K}S_1^k \left(\pi_a, \pi_b\right)$ and $S_2\left(\pi_a, \pi_b\right)=\sum_{k=1}^{K}S_2^k \left(\pi_a, \pi_b\right)$.

Because $\hat{\tilde{ m}}^{-k[i]}(\cdot,\cdot)$ is trained using the remaining $K-1$ folds of the data, when we condition on the data of the remaining $K-1$ folds, $S_1^k$ is a sum of i.i.d. random variables and has mean zero:
    \begin{align*}
       &  \mathbb{E}\left[ \sum_{i\in \mathcal{I}_k} \left(\pi_a(X_i,g)-\pi_b(X_i,g)\right)\{1-\tildepi(X_i,g)\} \sum_{g^\prime=1}^{g-1}  \mathds{1}\left(X_i\in[\tilde     c_{g^\prime},\tilde c_{g^\prime+1}]\right) \cdot\right.\\
       &\hspace{10em}
       \left(\hat{\tilde m}^{-k[i]}(X_i,g^\prime)-{\tilde m}(X_i,g^\prime)\right)
        \left.\left(\mathds{1}\left(G_i=g\right)-\mathds{1}\left(G_i=g^{\prime}\right)\frac{{e}_g(X_i)}{{e}_{g^{\prime}}(X_i)}\right)\right]\\
       = \ & 
         \mathbb{E}\left[ \sum_{i\in \mathcal{I}_k}\mathbb{E}\left[ \left(\pi_a(X_i,g)-\pi_b(X_i,g)\right)\{1-\tildepi(X_i,g)\} \sum_{g^\prime=1}^{g-1}  \mathds{1}\left(X_i\in[\tilde c_{g^\prime},\tilde c_{g^\prime+1}]\right) \cdot \right.\right.\\
         &\hspace{10em}\left.
       \left(\hat{\tilde m}^{-k[i]}(X_i,g^\prime)-{\tilde m}(X_i,g^\prime)\right)
        \left.\left(\mathds{1}\left(G_i=g\right)-\mathds{1}\left(G_i=g^{\prime}\right)\frac{{e}_g(X_i)}{{e}_{g^{\prime}}(X_i)}\right)\mid  \mathcal{I}_{-k}, X_i \right]\right] \\
        = \ &\mathbb{E}\left[\sum_{i\in \mathcal{I}_k} \left(\pi_a(X_i,g)-\pi_b(X_i,g)\right)\{1-\tildepi(X_i,g)\} \sum_{g^\prime=1}^{g-1}  \mathds{1}\left(X_i\in[\tilde c_{g^\prime},\tilde c_{g^\prime+1}]\right) 
        \left(\hat{\tilde m}^{-k[i]}(X_i,g^\prime)-{\tilde m}(X_i,g^\prime)\right)\cdot \right.\\
         &\hspace{15em}\left.
       \mathbb{E}\left[\sum_{i\in \mathcal{I}_k}\left(\mathds{1}\left(G_i=g\right)-\mathds{1}\left(G_i=g^{\prime}\right)\frac{{e}_g(X_i)}{{e}_{g^{\prime}}(X_i)}\right)\mid  \mathcal{I}_{-k}, X_i\right]\right] \\
         = \ &\mathbb{E}\left[\sum_{i\in \mathcal{I}_k} \left(\pi_a(X_i,g)-\pi_b(X_i,g)\right)\{1-\tildepi(X_i,g)\} \sum_{g^\prime=1}^{g-1}  \mathds{1}\left(X_i\in[\tilde c_{g^\prime},\tilde c_{g^\prime+1}]\right) \left(\hat{\tilde m}^{-k[i]}(X_i,g^\prime)-{\tilde m}(X_i,g^\prime)\right)\cdot \right.\\
         &\left.\hspace{30em}
         \left\{e_g(X_i)-e_{g^\prime}(X_i)\frac{e_g(X_i)}{e_{g^\prime}(X_i)}\right\}\right]\\
         = \ & 0.
    \end{align*}
For simplicity, we assume the training data is divided into $K$ evenly-sized folds, then we can rewrite the bound on the first term as    
\begin{align*}
         & \sup_{\pi_a,\pi_b \in \Pi}|S_1^k\left(\pi_a, \pi_b\right)|=
         \sup_{\pi_a,\pi_b \in \Pi}\left|S_1^k\left(\pi_a, \pi_b\right)-\mathbb{E}\left[S_1^k\left(\pi_a, \pi_b\right)\right]\right|\\
         = & \frac{1}{K} \sup_{\pi_a,\pi_b \in \Pi}\left| \frac{1}{\frac{n}{K}} \sum_{i\in \mathcal{I}_k}
         \left(\pi_a(X_i,g)-\pi_b(X_i,g)\right)
     \cdot\right.
         \{1-\tildepi(X_i,g)\} \\
         & \hspace{1in}\sum_{g^\prime=1}^{g-1}  \mathds{1}\left(X_i\in[\tilde     c_{g^\prime},\tilde c_{g^\prime+1}]\right)   \left(\hat{\tilde m}^{-k[i]}(X_i,g^\prime)-{\tilde m}(X_i,g^\prime)\right)
        \left.\left(\mathds{1}\left(G_i=g\right)-\mathds{1}\left(G_i=g^{\prime}\right)\frac{{e}_g(X_i)}{{e}_{g^{\prime}}(X_i)}\right)\right.\\
         & -\mathbb{E} \Big[\left(\pi_a(X_i,g)-\pi_b(X_i,g)\right)\{1-\tildepi(X_i,g)\} \\
        & \left.\hspace{.8in} \sum_{g^\prime=1}^{g-1}  \mathds{1}\left(X_i\in[\tilde     c_{g^\prime},\tilde c_{g^\prime+1}]\right) \cdot\right.
       \left(\hat{\tilde m}^{-k[i]}(X_i,g^\prime)-{\tilde m}(X_i,g^\prime)\right)
        \left.\left(\mathds{1}\left(G_i=g\right)-\mathds{1}\left(G_i=g^{\prime}\right)\frac{{e}_g(X_i)}{{e}_{g^{\prime}}(X_i)}\right)\Big]\right|.
\end{align*}

By clipping the fitted outcome regressions and trimming the fitted group probabilities as stated in Assumption~\ref{ass:rate:general}, and using Assumptions~\ref{ass:gen:bo} and~\ref{ass:overlap:general},
\[
\{1-\tildepi(X_i,g)\} \sum_{g^\prime=1}^{g-1}  \mathds{1}\left(X_i\in[\tilde     c_{g^\prime},\tilde c_{g^\prime+1}]\right) \cdot
       \left(\hat{\tilde m}^{-k[i]}(X_i,g^\prime)-{\tilde m}(X_i,g^\prime)\right)
        \left(\mathds{1}\left(G_i=g\right)-\mathds{1}\left(G_i=g^{\prime}\right)\frac{{e}_g(X_i)}{{e}_{g^{\prime}}(X_i)}\right)
        \]
        has a bounded envelope. We can therefore apply Lemma 2 of \cite{zhou2018offline}, conditional on $\mathcal{I}_{-k}$, to the group-specific binary policy class $\Pi_g$.
Then, we have the following result $\forall \delta>0$, with probability at least $1-2 \delta$,
\begin{align}
         & K\sup_{\pi_a,\pi_b \in \Pi}|S_1^k\left(\pi_a, \pi_b\right)|=
         K\sup_{\pi_a,\pi_b \in \Pi}\left|S_1^k\left(\pi_a, \pi_b\right)-\mathbb{E}\left[S_1^k\left(\pi_a, \pi_b\right)\right]\right|\nonumber\\
        \leq & o\left(\frac{1}{\sqrt{n}}\right)+\left(54.4 \sqrt{2} \kappa(\Pi_g)+435.2+\sqrt{2 \log \frac{1}{\delta}}\right)\times \nonumber \\
        & \left\{\frac{1}{n/K}\sup_{\pi_a,\pi_b \in \Pi} \mathbbm{E}\Biggl[\left(\pi_a(X_i,g)-\pi_b(X_i,g)\right)^2
        (1-\tildepi(X_i,g))\right.\nonumber\\
        & \quad \left.\left.\sum_{g^\prime=1}^{g-1}  \mathds{1}\left(X_i\in[\tilde     c_{g^\prime},\tilde c_{g^\prime+1}]\right)  \left(\hat{\tilde m}^{-k[i]}(X_i,g^\prime)-{\tilde m}(X_i,g^\prime)\right)^2\left(\mathds{1}\left(G_i=g\right)-\mathds{1}\left(G_i=g^{\prime}\right)\frac{{e}_g(X_i)}{{e}_{g^{\prime}}(X_i)}\right)^2\mid\mathcal{I}_{-k} \right]\right\}^{1/2}\nonumber\\ 
         \leq & o\left(\frac{1}{\sqrt{n}}\right)+\left(54.4 \sqrt{2} \kappa(\Pi_g)+435.2+\sqrt{2 \log \frac{1}{\delta}}\right)\times\nonumber\\
         &\left\{\frac{1}{n/K}\mathbbm{E}\left[   (1-\tildepi(X_i,g))\sum_{g^\prime=1}^{g-1}  \mathds{1}\left(X_i\in[\tilde     c_{g^\prime},\tilde c_{g^\prime+1}]\right) \left(\hat{\tilde m}^{-k[i]}(X_i,g^\prime)-{\tilde m}(X_i,g^\prime)\right)^2\right.\right.\nonumber \\
         & \hspace{2in}\left.\left.\left(\mathds{1}\left(G_i=g\right)-\mathds{1}\left(G_i=g^{\prime}\right)\frac{{e}_g(X_i)}{{e}_{g^{\prime}}(X_i)}\right)^2 \mid\mathcal{I}_{-k} \right]\right\}^{1/2}\nonumber\\ 
         \leq & o\left(\frac{1}{\sqrt{n}}\right)+\frac{1}{\eta^2}\left(54.4 \sqrt{2} \kappa(\Pi_g)+435.2+\sqrt{2 \log \frac{1}{\delta}}\right)\nonumber \\
         & \sqrt{\frac{1}{n/K}\mathbbm{E}\left[ (1-\tildepi(X_i,g))\sum_{g^\prime=1}^{g-1}  \mathds{1}\left(X_i\in[\tilde     c_{g^\prime},\tilde c_{g^\prime+1}]\right) \left(\hat{\tilde m}^{-k[i]}(X_i,g^\prime)-{\tilde m}(X_i,g^\prime)\right)^2\mid\mathcal{I}_{-k} \right]}, \label{equ:S1}
    \end{align}
where the last inequality follows from Assumption \ref{ass:overlap:general}. Assumption~\ref{ass:rate:general} gives, conditional on the training sample,
$$
\mathbbm{E}\left[ \left(\hat{\tilde m}^{-k[i]}(X_i,g^{\prime})-{\tilde m}(X_i,g^{\prime})\right)^2\mid \mathcal{I}_{-k} \right]=o_p(1).
$$
Combining this result with Equation~\eqref{equ:S1} yields
$$\sup_{\pi_a,\pi_b \in \Pi}|S_1^k\left(\pi_a, \pi_b\right)|=o_p\left(n^{-1/2}\right).$$
Therefore, 
\begin{equation}
        \sup_{\pi_a,\pi_b \in \Pi}|S_1\left(\pi_a, \pi_b\right)|=\sup_{\pi_a,\pi_b \in \Pi}|\sum_{k=1}^{K}S_1^k\left(\pi_a, \pi_b\right)|\leq\sum_{k=1}^{K}\sup_{\pi_a,\pi_b \in \Pi}|S_1^k\left(\pi_a, \pi_b\right)|=o_p\left(\frac{1}{\sqrt{n}}\right).
\end{equation}
By similar argument, we can also show that 
$$\sup_{\pi_a,\pi_b \in \Pi}|S_2^k\left(\pi_a, \pi_b\right)|= o_p\left(\frac{1}{\sqrt{n}}\right)$$ 
and hence, 
$$\sup_{\pi_a,\pi_b \in \Pi}|S_2\left(\pi_a, \pi_b\right)|= o_p\left(\frac{1}{\sqrt{n}}\right).$$

Next, we bound the third component $S_3$ as follows:
    \begin{align*}
        & \sup_{\pi_a,\pi_b \in \Pi}|S_3\left(\pi_a, \pi_b\right)| \\
        &= \frac{1}{n} \sup_{\pi_a,\pi_b \in \Pi}\left|
           \sum_{i=1}^n \left(\pi_a(X_i,g)-\pi_b(X_i,g)\right)\{1-\tildepi(X_i,g)\} \sum_{g^\prime=1}^{g-1}  \mathds{1}\left(X_i\in[\tilde c_{g^\prime},\tilde c_{g^\prime+1}]\right)\times\right.\\
           &\left.\hspace{10em}\mathds{1}\left(G_i=g^{\prime}\right)\left(\frac{\hat{e}^{-k[i]}_g(X_i)}{\hat{e}^{-k[i]}_{g^{\prime}}(X_i)}-\frac{{e}_g(X_i)}{{e}_{g^{\prime}}(X_i)}\right)\left({\tilde m}(X_i,g^\prime)-\hat{\tilde m}^{-k[i]}(X_i,g^\prime)\right)\right|\\
        & \leq \sum_{g^\prime=1}^{g-1} \frac{1}{n} \sum_{i=1}^n \left|\left(\frac{\hat{e}^{-k[i]}_g(X_i)}{\hat{e}^{-k[i]}_{g^{\prime}}(X_i)}-\frac{{e}_g(X_i)}{{e}_{g^{\prime}}(X_i)}\right)\left({\tilde m}(X_i,g^\prime)-\hat{\tilde m}^{-k[i]}(X_i,g^\prime)\right)\right|\\
        & \leq \sum_{g^\prime=1}^{g-1} \sqrt{ \frac{1}{n} \sum_{i=1}^n \left(\frac{\hat{e}^{-k[i]}_g(X_i)}{\hat{e}^{-k[i]}_{g^{\prime}}(X_i)}-\frac{{e}_g(X_i)}{{e}_{g^{\prime}}(X_i)}\right)^2}
        \sqrt{ \frac{1}{n} \sum_{i=1}^n \left({\tilde m}(X_i,g^\prime)-\hat{\tilde m}^{-k[i]}(X_i,g^\prime)\right)^2},
    \end{align*}
where the last inequality follows from the Cauchy-Schwartz inequality. Taking expectation on both sides yields: 
    \begin{align*}
        & \mathbb{E}\left[\sup_{\pi_a,\pi_b \in \Pi}|S_3\left(\pi_a, \pi_b\right)|\right] \\
         \leq & \sum_{g^\prime=1}^{g-1}\mathbb{E}\left[\sqrt{ \frac{1}{n} \sum_{i=1}^n \left(\frac{\hat{e}^{-k[i]}_g(X_i)}{\hat{e}^{-k[i]}_{g^{\prime}}(X_i)}-\frac{{e}_g(X_i)}{{e}_{g^{\prime}}(X_i)}\right)^2}
        \sqrt{ \frac{1}{n} \sum_{i=1}^n \left({\tilde m}(X_i,g^\prime)-\hat{\tilde m}^{-k[i]}(X_i,g^\prime)\right)^2}\right]\\
         \leq & \sum_{g^\prime=1}^{g-1}\sqrt{\mathbb{E}\left[{ \frac{1}{n} \sum_{i=1}^n \left(\frac{\hat{e}^{-k[i]}_g(X_i)}{\hat{e}^{-k[i]}_{g^{\prime}}(X_i)}-\frac{{e}_g(X_i)}{{e}_{g^{\prime}}(X_i)}\right)^2}
        \right]}
        \sqrt{\mathbb{E}\left[{ \frac{1}{n} \sum_{i=1}^n \left({\tilde m}(X_i,g^\prime)-\hat{\tilde m}^{-k[i]}(X_i,g^\prime)\right)^2}\right]}\\
         = & \sum_{g^\prime=1}^{g-1}
        \sqrt{\mathbb{E}\left[{ \left(\frac{\hat{e}^{-k[i]}_g(X_i)}{\hat{e}^{-k[i]}_{g^{\prime}}(X_i)}-\frac{{e}_g(X_i)}{{e}_{g^{\prime}}(X_i)}\right)^2}
        \right]} \sqrt{\mathbb{E}\left[ {  \left({\tilde m}(X_i,g^\prime)-\hat{\tilde m}^{-k[i]}(X_i,g^\prime)\right)^2}\right]}\\
         \leq & (Q-1)\sqrt{ \frac{o(1)}{\left(\frac{K-1}{K} n\right)}}  = o\left(\frac{1}{\sqrt{n}}\right),
    \end{align*}

where the second inequality again follows from Cauchy-Schwartz, and the third inequality follows from Assumption \ref{ass:rate:general} and the fact that $\hat{\tilde m}^{-k[i]}(\cdot,\cdot)$ and $\hat e^{-k[i]}(\cdot)$ are trained on $\frac{K-1}{K}n$ data points.
Consequently, by Markov's inequality, 
$$\sup_{\pi_a,\pi_b \in \Pi}|S_3\left(\pi_a, \pi_b\right)|= o_p\left(\frac{1}{\sqrt{n}}\right).$$  
Summing up all our bounds completes the proof.
\end{proof}
\subsection{Proof of Theorem \ref{thm:multi:safe:general}}\label{app:proof:thm1}

Given two policies, $\pi_1$ and $\pi_2$ in the same policy class, we define the \textit{utility loss} (or \textit{regret}) of $\pi_1$ relative to $\pi_2$ as the difference in the expected utility between them:
$$R(\pi_1,\pi_2) \ = \ V(\pi_2)-V(\pi_1).$$
Below we state the full version of the result in Theorem~\ref{thm:multi:safe:general}.
\begin{theorem*}[Statistical safety guarantee under multi-cutoff regression discontinuity designs]
Under Assumptions~\ref{ass:general:lipdiff},~\ref{ass:overlap:general},~\ref{ass:gen:bo},~\ref{ass:rate:general}~and~\ref{ass:bdrate},
 let $\Pi$ be a prespecified measurable policy class satisfying $\tilde\pi\in\Pi\subseteq\Pi_{\mathrm S}$, $\nu<\infty$, and $\overline\nu<\infty$, and define $\hat\pi$ as in Equation~\eqref{equ:multi:estor}. If the true difference model $d(w,x,g,g^\prime)$ belongs to $\mathcal{M}_d$, then, for any $\delta,h_1,h_2>0$, there exist $C_0>0$, $0<N(\delta,h_1,h_2)<\infty$, and a universal constant $u_1$ such that, for all $n\geq N(\delta,h_1,h_2)$, with probability at least $1-4\delta$,
     \begin{align*}
& V(\tilde{\pi})-V(\hat{\pi}) \\
\leq & \frac{4\Lambda}{\sqrt n}\left(1+\frac{3}{2}u_1\sqrt\nu+\sqrt{2\log\frac{1}{\delta}}\right)
+\sum_{g=1}^Q A_{g,\Pi}(\delta)\sqrt{\frac{V^\ast_{g,\Pi}}{n}}\\
&\quad+\frac{h_1+h_2}{\sqrt n}+\frac{C_0}{\rho_n},
\end{align*}
where $V^\ast_{g,\Pi}$ is defined in Section~\ref{sec:theory} and $\Gamma^{(1)}_{ig}$ and $\Gamma^{(0)}_{ig}$ are defined in Equation~\eqref{equ:thm:DRscore}.

\end{theorem*} 

\begin{proof}

By Definitions~\ref{def:J}~and~\ref{def:K}, the regret of $\hat\pi$ relative to $\tilde\pi$ can be decomposed into 
 \begin{equation*}
     \begin{aligned}
R(\hat{\pi}, \tilde{\pi}) = & V(\tilde{\pi})-V(\hat{\pi}) \\
=& V(\tilde{\pi})-\hat{V}(\tilde{\pi})+\hat{V}(\tilde{\pi})-\hat{V}(\hat{\pi})+\hat{V}(\hat{\pi})-V(\hat{\pi}) \\
= & J(\tilde{\pi})-\hat{J}(\tilde{\pi})+ K(\tilde{\pi})-\hat{K}(\tilde{\pi}) +\hat{V}(\tilde{\pi})-\hat{V}(\hat{\pi})+\hat{J}(\hat{\pi})-J(\hat{\pi})+\hat{K}(\hat{\pi})-K(\hat{\pi})\\
\leq & \underbrace{2 \sup _{\pi \in \Pi}\left|\hat{J}(\pi)-{J}(\pi)\right|}_{\cir{1}}
+\underbrace{\hat{V}(\tilde{\pi})-\hat{V}(\hat{\pi})}_{\cir{2}}
+ \underbrace{ K(\tilde{\pi})-\hat{K}(\tilde{\pi})+\hat{K}(\hat{\pi})-K(\hat{\pi})}_{\cir{3}},
\end{aligned}
 \end{equation*}
 where the third equality is because $V(\pi)=J(\pi)+K(\pi)$, $\hat V(\pi)=\hat J(\pi)+\hat K(\pi)$.
 
First, due to Lemma \ref{lemma:general:T1}, we can bound the first term $\cir{1}$, i.e., for any $\delta>0$,
{
\begin{equation}
 2\sup _{\pi \in \Pi}\left|\hat{J}(\pi)-{J}(\pi)\right|
      \leq  \frac{4\Lambda}{\sqrt{n}}\left(1+ \frac{3}{2}u_1 \sqrt{\nu}+  \sqrt{2\log \frac{1}{\delta}}\right),
\end{equation}
}
with probability at least $1-\delta$, where $u_1$ is a universal constant.

Now we look at the second term $\cir{2}$, $\hat{V}(\tilde{\pi})-\hat{V}(\hat{\pi})$. We can further decompose it as
$$
\begin{aligned}
\hat{V}(\tilde{\pi})-\hat{V}(\hat{\pi}) = 
\hat{V}(\tilde{\pi})-\min_{d \in {\mathcal{M}_d}} \hat{V}(\tilde{\pi},d)
+\min_{d \in {\mathcal{M}_d}} \hat{V}(\hat{\pi},d)-\hat{V}(\hat{\pi})
+\min_{d \in {\mathcal{M}_d}} \hat{V}(\tilde{\pi},d)-\min_{d \in {\mathcal{M}_d}} \hat{V}(\hat{\pi},d).
\end{aligned}
$$
Because the true difference function $d(w,x,g,g^\prime) \in \mathcal{M}_d$, we have $\hat{V}(\hat{\pi}) \geq \min_{d \in {\mathcal{M}_d}} \hat{V}(\hat{\pi},d)$; also, we note that no extrapolation is needed for the baseline policy being used, so that the worst case value equals the true value for $\tilde{\pi}$, i.e., $\hat{V}(\tilde{\pi})=\min_{d \in {\mathcal{M}_d}} \hat{V}(\tilde{\pi},d)$. Combining these two, we have  
$$
\hat{V}(\tilde{\pi})-\min_{d \in {\mathcal{M}_d}} \hat{V}(\tilde{\pi},d)+\min_{d \in {\mathcal{M}_d}} \hat{V}(\hat{\pi},d)-\hat{V}(\hat{\pi})\leq 0.$$
Also, we can see that
\begin{align*}
&\min_{d \in {\mathcal{M}_d}} \hat{V}(\tilde{\pi},d)
-\min_{d \in {\mathcal{M}_d}} \hat{V}(\hat{\pi},d)\\
=
&\min_{d \in {\mathcal{M}_d}} \hat{V}(\tilde{\pi},d)
-
\min_{d \in \widehat{\mathcal{M}}_{d}}  \hat{V}(\tilde{\pi},d)
+
\min_{d \in \widehat{\mathcal{M}}_{d}}  \hat{V}(\hat{\pi},d)-\min_{d \in {\mathcal{M}_d}} \hat{V}(\hat{\pi},d)
+
\min_{d \in \widehat{\mathcal{M}}_{d}}  \hat{V}(\tilde{\pi},d)
-
\min_{d \in \widehat{\mathcal{M}}_{d}}   \hat{V}(\hat{\pi},d)\\
\leq &  
 \   \left|\min_{d \in \widehat{\mathcal{M}}_{d}} \hat{V}(\tilde{\pi},d)-\min_{d \in {\mathcal{M}_d}}\hat{V}(\tilde{\pi},d)\right|+ \left|\min_{d \in \widehat{\mathcal{M}}_{d}} \hat{V}(\hat{\pi},d)-\min_{d \in {\mathcal{M}_d}}\hat{V}(\hat{\pi},d)\right|\\
=  &
 \   \left|\min_{d \in \widehat{\mathcal{M}}_{d}} \hat{J}(\tilde{\pi},d)-\min_{d \in {\mathcal{M}_d}}\hat{J}(\tilde{\pi},d)\right|+ \left|\min_{d \in \widehat{\mathcal{M}}_{d}} \hat{J}(\hat{\pi},d)-\min_{d \in {\mathcal{M}_d}}\hat{J}(\hat{\pi},d)\right|=  O_p \left(\rho_n^{-1}\right),
\end{align*}
where the second to last inequality is because $\hat{\pi}$ maximizes $\min_{d \in \widehat{\mathcal{M}}_{d}}  \hat{V}(\pi,d)$ and the last equality is due to Lemma~\ref{lemma:multi:noise}.  Therefore, we have for any $\delta>0$, there exists $C_0,N(\delta)>0$, if we collect at least $n\geq N(\delta)$ samples, then, with probability at least $1-\delta$,
$$\hat{V}(\tilde{\pi})-\hat{V}(\hat{\pi})\leq\frac{C_0}{\rho_n}.$$

Finally, we bound the remaining term $\cir{3}$ as follows:
\begin{equation}
    \begin{aligned}
        K(\tilde{\pi})-\hat{K}(\tilde{\pi})+\hat{K}(\hat{\pi})-K(\hat{\pi}) &= \hat\Delta\left(\hat\pi,\tilde\pi\right)- \Delta\left(\hat\pi,\tilde\pi\right)\\
        &\leq \left|\hat\Delta\left(\hat\pi,\tilde\pi\right)- \Delta\left(\hat\pi,\tilde\pi\right)\right|\\
        & \leq \sup _{\pi_a, \pi_b \in \Pi}\left|\hat{\Delta}\left(\pi_a, \pi_b\right)-{\Delta}\left(\pi_a, \pi_b\right)\right|\\
        & \leq \sup _{\pi_a, \pi_b \in \Pi}\left|\tilde{\Delta}\left(\pi_a, \pi_b\right)-{\Delta}\left(\pi_a, \pi_b\right)\right| + \sup _{\pi_a, \pi_b \in \Pi}\left|\hat{\Delta}\left(\pi_a, \pi_b\right)-\tilde{\Delta}\left(\pi_a, \pi_b\right)\right|.
    \end{aligned}
\end{equation}
Following Lemma~\ref{lemma:oracle}, for any $\delta,h_1>0$, there exists $0<N(\delta,h_1)<\infty$ such that, if $n\geq N(\delta,h_1)$, then, with probability at least $1-\delta$,
{
$$
\sup _{\pi_a, \pi_b \in \Pi}\left|\tilde{\Delta}\left(\pi_a, \pi_b\right)-\Delta\left(\pi_a, \pi_b\right)\right| \leq
\sum_{g=1}^Q A_{g,\Pi}(\delta)\sqrt{\frac{V^\ast_{g,\Pi}}{n}}+\frac{h_1}{\sqrt n}.
$$
}
Also, following Lemma \ref{lemma:feasible}, we have for any $\delta,h_2>0$, there exists $0<N(\delta,h_2)<\infty$, if we collect at least $n\geq N(\delta,h_2)$ samples, then, with probability at least $1-\delta$,
$$\sup _{\pi_a, \pi_b \in \Pi}\left|\hat{\Delta}\left(\pi_a, \pi_b\right)-\tilde{\Delta}\left(\pi_a, \pi_b\right)\right|\leq \frac{h_2}{\sqrt{n}}.$$

Then by applying the union bound, for any $\delta,h_1,h_2>0$, there exists $C_0>0$, $0<N(\delta,h_1,h_2)<\infty$ and a universal constant $u_1$ such that, with probability at least $1-4\delta$,
{
 \begin{equation*}
     \begin{aligned}
R(\hat{\pi}, \tilde{\pi}) \leq
\frac{4\Lambda}{\sqrt{n}}\left(1+ \frac{3}{2}u_1 \sqrt{\nu}+  \sqrt{2\log \frac{1}{\delta}}\right)+
\sum_{g=1}^Q A_{g,\Pi}(\delta)\sqrt{\frac{V^\ast_{g,\Pi}}{n}}+ \frac{h_1+h_2}{\sqrt{n}}+\frac{C_0}{\rho_n}
\end{aligned}
 \end{equation*}
}
provided we collect at least $n\geq N(\delta,h_1,h_2)$ samples.

\end{proof}

\subsection{Proof of Theorem \ref{thm:multi:optimalgap:general}}\label{app:proof:thm2}

Below we state a complete version of the result given in Theorem~\ref{thm:multi:optimalgap:general}.
\begin{theorem*}[Empirical optimality gap under multi-cutoff regression discontinuity designs]
Under Assumptions~\ref{ass:general:lipdiff},~\ref{ass:overlap:general},~\ref{ass:gen:bo},~\ref{ass:rate:general}~and~\ref{ass:bdrate},
 let $\Pi$ be a prespecified measurable policy class satisfying $\tilde\pi\in\Pi\subseteq\Pi_{\mathrm S}$, $\nu<\infty$, and $\overline\nu<\infty$, and define $\hat\pi$ as in Equation~\eqref{equ:multi:estor}. Let $\pi^*$ be the optimal policy in $\Pi$. If the true difference function $d(w,x,g,g^\prime)$ belongs to $\mathcal{M}_d$, then, for any $\delta,h_1,h_2>0$, there exist $C_0>0$, $0<N(\delta,h_1,h_2)<\infty$, and a universal constant $u_1$ such that, for all $n\geq N(\delta,h_1,h_2)$, with probability at least $1-4\delta$,
\begin{equation*}
    \begin{aligned}
V(\pi^\ast)-V(\hat{\pi}) \leq &\ 
\frac{4\Lambda}{\sqrt n}\left(1+\frac{3}{2}u_1\sqrt\nu+\sqrt{2\log\frac{1}{\delta}}\right)
+\sum_{g=1}^Q A_{g,\Pi}(\delta)\sqrt{\frac{V^\ast_{g,\Pi}}{n}}+\frac{h_1+h_2}{\sqrt n}\\
& + \frac{2}{n}\sum_{i=1}^{n} \max _{w \in \{0,1\},g,g^\prime\in\mathcal{Q}}
              \lambda_{wgg^\prime} |X_i-\tilde c_{{g^\ast}}|+\frac{C_0}{\rho_n},
    \end{aligned}
\end{equation*}
 where ${g^\ast}=w(g\vee g^\prime)+(1-w)(g\wedge g^\prime)$ denotes the nearest cutoff index used to extrapolate $d(w,x,g,g^\prime)$.

\end{theorem*}

\begin{proof}
By Definitions~\ref{def:J}~and~\ref{def:K}, the regret of $\hat\pi$ relative to $\pi^*$ is 
  \begin{equation*}
        \begin{aligned}
R\left(\hat{\pi}, \pi^{*}\right) &=V\left(\pi^{*}\right)-V(\hat{\pi}) \\
&=V\left(\pi^{*}\right)-\hat{V}\left(\pi^{*}\right)+\hat{V}\left(\pi^{*}\right)-\hat{V}(\hat{\pi})+\hat{V}(\hat{\pi})-V(\hat{\pi}) \\
& =  J(\pi^*)-\hat{J}(\pi^*)+ K(\pi^*)-\hat{K}(\pi^*) +\hat{V}({\pi^*})-\hat{V}(\hat{\pi})+\hat{J}(\hat{\pi})-J(\hat{\pi})+\hat{K}(\hat{\pi})-K(\hat{\pi})\\
& \leq \underbrace{2 \sup _{\pi \in \Pi}\left|\hat{J}(\pi)-{J}(\pi)\right|}_{\cir{1}}+
\underbrace{\hat{V}({\pi^*})-\hat{V}(\hat{\pi})}_{\cir{2}}+ \underbrace{K(\pi^*)-\hat{K}(\pi^*)+\hat{K}(\hat{\pi})-K(\hat{\pi})}_{\cir{3}}
\end{aligned}
  \end{equation*}

We can similarly bound the terms $\cir{1}$ and $\cir{3}$ following the argument in  the proof of Theorem~\ref{thm:multi:safe:general}. We now turn to  term $\cir{2}$. Note that we have the following decomposition:
\begin{align}
 & \hat{V}\left(\pi^*\right)-\hat{V}(\hat{\pi})\nonumber\\
 =&\min_{d \in {\mathcal{M}_d}} \hat{V}(\pi^\ast,d)-\hat{V}(\hat{\pi})\nonumber\\
 &- \inf_{d \in {\mathcal{M}_d}} \frac{1}{n}\sum_{i=1}^{n} 
 \left[\sum_{g=1}^{Q} \mathds{1}\left(G_i=g\right)\cdot \pi^\ast(X_i,g)\left(1-\tilde\pi(X_i,g)\right) \sum_{g^\prime=1}^{g-1} \mathds{1}\left(X_i\in[\tilde c_{g^\prime}, \tilde c_{g^\prime+1}]\right)  d(1,X_i,g,g^{\prime}) \right.\nonumber\\
     &\qquad + \left. \sum_{g=1}^{Q} \mathds{1}\left(G_i=g\right)\cdot \left(1-\pi^\ast(X_i,g)\right)\tilde\pi(X_i,g) \sum_{g^\prime=g+1}^{Q} \mathds{1}\left(X_i\in[\tilde c_{g^\prime-1}, \tilde c_{g^\prime}]\right)   d(0,X_i,g,g^{\prime})\right]\nonumber\\
 &+\frac{1}{n}\sum_{i=1}^{n} 
 \left[\sum_{g=1}^{Q} \mathds{1}\left(G_i=g\right)\cdot \pi^\ast(X_i,g)\left(1-\tilde\pi(X_i,g)\right) \sum_{g^\prime=1}^{g-1} \mathds{1}\left(X_i\in[\tilde c_{g^\prime}, \tilde c_{g^\prime+1}]\right)  d(1,X_i,g,g^{\prime}) \right.\nonumber\\
     &\qquad + \left. \sum_{g=1}^{Q} \mathds{1}\left(G_i=g\right)\cdot \left(1-\pi^\ast(X_i,g)\right)\tilde\pi(X_i,g) \sum_{g^\prime=g+1}^{Q} \mathds{1}\left(X_i\in[\tilde c_{g^\prime-1}, \tilde c_{g^\prime}]\right)   d(0,X_i,g,g^{\prime})\right],\label{equ:proof:Vdiff}
\end{align}
where we obtain the RHS quantity by replacing the true model $d$ in $ \hat{V}\left(\pi^*\right)$ with the worst case value in the partially identified model class $ {\mathcal{M}_d}$. To deal with the first term in Equation~\eqref{equ:proof:Vdiff}, we show that 
\begin{align*}
    & \min_{d \in {\mathcal{M}_d}} \hat{V}(\pi^\ast,d)-\hat{V}(\hat{\pi})\\
    = & \min_{d \in {\mathcal{M}_d}} \hat{V}(\pi^\ast,d)
    -\min_{d \in {\widehat{\mathcal{M}}_d}} \hat{V}(\pi^\ast,d)
    + \min_{d \in {\widehat{\mathcal{M}}_d}} \hat{V}(\pi^\ast,d)
    - \min_{d \in {\widehat{\mathcal{M}}_d}} \hat{V}(\hat{\pi},d)
    + \min_{d \in {\widehat{\mathcal{M}}_d}} \hat{V}(\hat{\pi},d)
    -\hat{V}(\hat{\pi}) \\
    \leq &  
 \   \left|\min_{d \in {\mathcal{M}_d}}\hat{V}({\pi}^\ast,d)-\min_{d \in \widehat{\mathcal{M}}_{d}} \hat{V}({\pi}^\ast,d)\right|+ 0+ \left|\min_{d \in \widehat{\mathcal{M}}_{d}} \hat{V}(\hat{\pi},d)-\min_{d \in {\mathcal{M}_d}}\hat{V}(\hat{\pi},d)\right|\\
=  &
 \   \left|\min_{d \in {\mathcal{M}_d}}\hat{J}({\pi}^\ast,d)-\min_{d \in \widehat{\mathcal{M}}_{d}} \hat{J}({\pi}^\ast,d)\right|+ \left|\min_{d \in \widehat{\mathcal{M}}_{d}} \hat{J}(\hat{\pi},d)-\min_{d \in {\mathcal{M}_d}}\hat{J}(\hat{\pi},d)\right|\\
    \leq &  O_p(\rho_n^{-1}).
\end{align*}
where the second inequality is because  $\hat{\pi}$ maximizes $\min_{d \in \widehat{\mathcal{M}}_{d}}  \hat{V}(\pi,d)$ and $\hat{V}(\hat{\pi})\geq\min_{d \in {\mathcal{M}_d}}\hat{V}(\hat{\pi},d)$, and the last equality is due to Lemma~\ref{lemma:multi:noise}. Combining the above, we thus have, 
\begin{align*}
& \hat{V}\left(\pi^*\right)-\hat{V}(\hat{\pi}) \\
\leq & \sup_{d \in {\mathcal{M}_d}}\frac{1}{n}\sum_{i=1}^{n} 
 \left[\sum_{g=1}^{Q} \mathds{1}\left(G_i=g\right)\cdot \pi^\ast(X_i,g)\left(1-\tilde\pi(X_i,g)\right) \sum_{g^\prime=1}^{g-1} \mathds{1}\left(X_i\in[\tilde c_{g^\prime}, \tilde c_{g^\prime+1}]\right)  d(1,X_i,g,g^{\prime}) \right.\\
&\qquad \qquad + \left. \sum_{g=1}^{Q} \mathds{1}\left(G_i=g\right)\cdot \left(1-\pi^\ast(X_i,g)\right)\tilde\pi(X_i,g) \sum_{g^\prime=g+1}^{Q} \mathds{1}\left(X_i\in[\tilde c_{g^\prime-1}, \tilde c_{g^\prime}]\right)   d(0,X_i,g,g^{\prime})\right]\\
&- \inf_{d \in {\mathcal{M}_d}} \frac{1}{n}\sum_{i=1}^{n} 
 \left[\sum_{g=1}^{Q} \mathds{1}\left(G_i=g\right)\cdot \pi^\ast(X_i,g)\left(1-\tilde\pi(X_i,g)\right) \sum_{g^\prime=1}^{g-1} \mathds{1}\left(X_i\in[\tilde c_{g^\prime}, \tilde c_{g^\prime+1}]\right)  d(1,X_i,g,g^{\prime}) \right.\\
&\qquad \qquad + \left. \sum_{g=1}^{Q} \mathds{1}\left(G_i=g\right)\cdot \left(1-\pi^\ast(X_i,g)\right)\tilde\pi(X_i,g) \sum_{g^\prime=g+1}^{Q} \mathds{1}\left(X_i\in[\tilde c_{g^\prime-1}, \tilde c_{g^\prime}]\right)   d(0,X_i,g,g^{\prime})\right]\\
& +O_p(\rho_n^{-1})\\
\leq & \ \frac{1}{n}\sum_{i=1}^{n} \max _{w \in \{0,1\},g,g^\prime\in\mathcal{Q}}\left\{ \sup_{d \in {\mathcal{M}_d}}d(w,X_i,g,g^{\prime})- \inf_{d \in {\mathcal{M}_d}}d(w,X_i,g,g^{\prime})\right\}+O_p(\rho_n^{-1})\\
= & \ \frac{1}{n}\sum_{i=1}^{n} \max _{w \in \{0,1\},g,g^\prime\in\mathcal{Q}}\left\{ B_{ u}(w,X_i,g,g^\prime)- B_{\ell}(w,X_i,g,g^\prime)\right\}+O_p(\rho_n^{-1}),
\end{align*}
where the point-wise upper and lower bounds, $ B_{ u}(w,X_i,g,g^\prime)$ and $B_{\ell}(w,X_i,g,g^\prime)$, are defined in Equation~\eqref{equ:genmulti:bound}. Plugging in their explicit forms, we get the bound for $\hat{V}\left(\pi^*\right)-\hat{V}(\hat{\pi})$.

Finally, by applying the union bound, for any $\delta,h_1,h_2>0$, there exists $C_0>0$, $0<N(\delta,h_1,h_2)<\infty$ and a universal constant $u_1$ such that, with probability at least $1-4\delta$,
{
     \begin{align*}
R(\hat{\pi}, {\pi}^\ast) \leq &
\frac{4\Lambda}{\sqrt{n}}\left(1+ \frac{3}{2}u_1 \sqrt{\nu}+  \sqrt{2\log \frac{1}{\delta}}\right)+
\sum_{g=1}^Q A_{g,\Pi}(\delta)\sqrt{\frac{V^\ast_{g,\Pi}}{n}}+ \frac{h_1+h_2}{\sqrt{n}}\\
& + \frac{2}{n}\sum_{i=1}^{n} \max _{w \in \{0,1\},g,g^\prime\in\mathcal{Q}}
              \lambda_{wgg^\prime} |X_i-\tilde c_{{g^\ast}}|+\frac{C_0}{\rho_n}
\end{align*}
}
provided we collect at least $n\geq N(\delta,h_1,h_2)$ samples, where ${g^\ast}=w(g\vee g^\prime)+(1-w)(g\wedge g^\prime)$ denotes the nearest cutoff point extrapolated to $X_i$.

\end{proof}

\section{Additional empirical results}\label{append:sec:realplots}

\begin{table}[t]
    \centering
     \caption{Applicant-level summary for 23 departments in Colombia in 2010.}

    \begin{tabular}{lcccc}
       \hline
Department & Sample size & Baseline cutoff & RD Estimate & s.e. \\
\hline \\
Magdalena & 247 & -828 & 0.52 & 0.40 \\
La Guajira & 238 & -824 & 0.19 & 0.25 \\
Bolivar & 1,052 & -786 & 0.11 & 0.21 \\
Caqueta & 280 & -779 & 0.61 & 0.35 \\
Cauca & 416 & -774 & -0.02 & 0.86 \\
Cordoba & 741 & -764 & -0.15 & 0.17 \\
Cesar & 238 & -758 & 0.52 & 0.37 \\
Sucre & 671 & -755 & 0.17 & 0.12 \\
Atlantico & 652 & -754 & $0.67\textcolor{blue}{^\star}$ & 0.13 \\
Arauca & 250 & -753 & -0.56 & 1.02 \\
Valle Del Cauca & 700 & -732 & 0.16 & 0.35 \\
Antioquia & 562 & -729 & $0.70\textcolor{blue}{^\star}$ & 0.23 \\
Norte De Santander & 272 & -723 & -0.11 & 0.42 \\
Putumayo & 283 & -719 & -0.38 & 0.82 \\
Tolima & 442 & -716 & 0.17 & 0.31 \\
Huila & 544 & -695 & -0.01 & 0.23 \\
Narino & 698 & -678 & 0.20 & 0.18 \\
Cundinamarca & 661 & -676 & 0.17 & 0.68 \\
Risaralda & 306 & -672 & 0.41 & 0.35 \\
Quindio & 244 & -660 & -0.05 & 0.11 \\
Santander & 617 & -632 & $0.66\textcolor{blue}{^\star}$ & 0.27 \\
Boyaca & 504 & -618 & 0.28 & 0.33 \\
Distrito Capital & 1,017 & -559 & 0.38 & 0.19 \\
\hline \\
    \end{tabular}
     \footnotesize{\\$\textcolor{blue}{^\star}$ indicates a significant RD treatment effect at the 5\% level}

    \label{tab:emp:summary}
\end{table}

Figure~\ref{fig:realapp:paircalibration} illustrates the boundary calibration for one representative group pair. The cross-group difference is identified and estimated in the common-treatment region, and its estimated boundary level anchors the Lipschitz envelope over the extrapolation interval. Multiplying the estimated baseline constant by $M$ expands or contracts this envelope. As a computational diagnostic, all 506 pair-state boundary-level and derivative fits returned finite estimates; the smallest relevant-pair sample over a common-treatment region was 52, and the minimum fitted group probability was 0.010. These finite-sample diagnostics do not by themselves verify the population support condition.

\begin{figure}[t]
    \centering
    \includegraphics[width=0.95\linewidth]{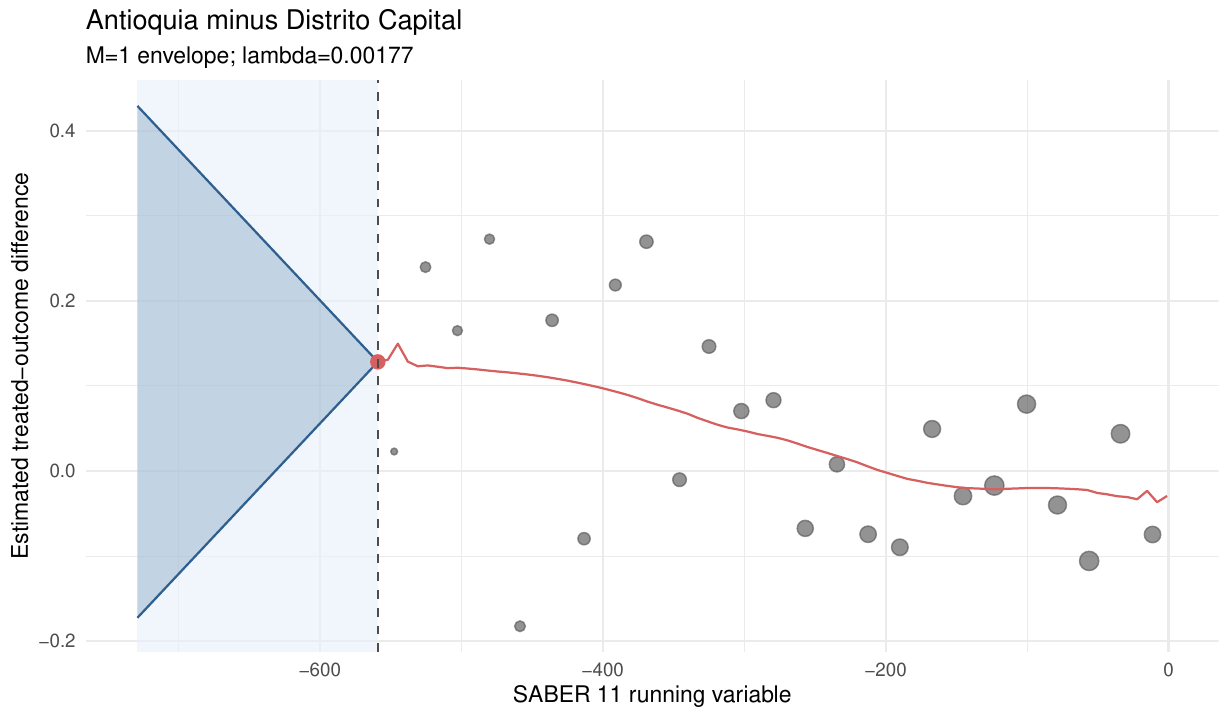}
    \caption{Illustration of the boundary calibration and $M=1$ extrapolation envelope for the treated-outcome difference between Antioquia and Distrito Capital. Points are binned averages of the cross-fitted pseudo-outcome, and the red line is a local-polynomial fit in the common-treatment region.}
    \label{fig:realapp:paircalibration}
\end{figure}

For the capacity-neutral analysis, let $\widehat L_g(c)$ denote the estimated lower-bound improvement from assigning cutoff $c$ to department $g$, relative to its status-quo cutoff, and let $N_g(c)$ be the number of applicants in department $g$ who would be eligible. On the finite candidate grid, we solve
\[
\max_{c_1,\ldots,c_Q}\sum_{g=1}^Q \widehat L_g(c_g)
\quad\text{subject to}\quad
\sum_{g=1}^Q N_g(c_g)\leq \sum_{g=1}^Q N_g(\tilde c_g).
\]
Because each department contributes one cutoff choice and one applicant count, this is a finite multiple-choice knapsack problem. We solve it exactly by dynamic programming and include the status quo as a feasible option for every department.

For the order-preserving analysis, departments are indexed by increasing status-quo cutoff and we impose $c_1\leq\cdots\leq c_Q$. Because the objective remains additive, an exact dynamic program records, for each department and candidate cutoff, the best value among feasible cutoffs for the preceding department. This restriction preserves the baseline ranking while still allowing every department cutoff to change. At $M=1$, the estimated lower-bound enrollment value is 0.8001 and 9,300 applicants are eligible.

Because ACCES used a single national cutoff through 2008, we also consider a common-cutoff benchmark over the currently supported candidate grid. The best common cutoff is $-779$, with an estimated lower-bound value of 0.7905 and 9,607 eligible applicants. Because the department-specific status quo is not contained in this restricted class, the benchmark does not inherit the safety comparison to the status quo. The result shows that requiring one common cutoff substantially lowers the best estimated worst-case value in this application.

\subsection{Alternative depth-2 policy-tree class}\label{append:sec:tree}

To illustrate the application of our criterion to another finite-complexity policy
class, let $r(x,g)=x-\widetilde c_g$ denote an applicant's score relative to the
status-quo cutoff. We consider binary decision trees that use the two features
$(r(x,g),\widetilde c_g)$. At each internal node, the tree partitions applicants according
to either their position relative to their own department's cutoff, $r(x,g)$, or the
level of that cutoff, $\widetilde c_g$. Splits on the first variable distinguish applicants
within departments, whereas splits on the second group departments with different cutoff
levels. A tree of depth at most two has
no more than two splits along any root-to-leaf path and no more than four terminal leaves,
each of which assigns treatment or control. If $\ell_T(r,\widetilde c_g)$ is the terminal leaf reached under tree $T$
and $a_\ell\in\{0,1\}$ is the action assigned to leaf $\ell$, the induced policy is
\[
\pi_T(x,g)=
\begin{cases}
a_{\ell_T(r(x,g),\widetilde c_g)},
& \widetilde c_1\leq x<\widetilde c_Q,\\
\widetilde\pi(x,g), & \text{otherwise}.
\end{cases}
\]
The policy class used in this supplementary analysis is
\[
\Pi_{\mathrm{tree}}^{(2)}
=\{\widetilde\pi\}\cup
\left\{\pi_T:T\text{ is an axis-aligned binary tree of depth at most two}\right\}.
\]
Thus, the class contains the status quo, and every candidate tree changes assignments
only within the supported score region. The tree class trades unrestricted
department-specific cutoff heterogeneity for a small number of data-selected partitions.
Because its depth and splitting variables are fixed, this class has finite VC dimension
and satisfies the policy-class complexity condition used in Theorems~\ref{thm:multi:safe:general}
and~\ref{thm:multi:optimalgap:general}.
In the empirical search, we additionally require each terminal leaf to contain at least
100 applicants, approximately 5\% of the 1,958 applicants in the supported region. We
maximize the same cross-fitted selected-reference lower-value criterion used in the main
analysis, with $M=1$.

The selected tree divides departments into two groups and applies a common cutoff
adjustment within each group. Within the supported score region, it is
\[
\widehat\pi_{\mathrm{tree}}(x,g)=
\begin{cases}
\mathds{1}\{x-\widetilde c_g>-25.5\},
& \widetilde c_g\leq -742.5,\\
\mathds{1}\{x-\widetilde c_g>-53.5\},
& \widetilde c_g>-742.5,
\end{cases}
\]
and it equals $\widetilde\pi(x,g)$ outside this region. Although the class permits more
general piecewise-constant rules, the selected policy is therefore monotone in score and
uses two shared cutoff adjustments. Its estimated lower-bound value is 0.7984, compared
with 0.7944 under the status quo, and it makes 9,389 applicants eligible rather than
8,971. The department-specific threshold policy used in the main analysis retains the
higher estimated lower-bound value of 0.8014. The estimated lower-bound value is similar
when the minimum leaf size is varied from 50 to 500 applicants.

\subsection{Fixed-policy all-reference benchmark}

To assess the empirical effect of selecting a single reference group without solving a second policy-learning problem, we hold the unrestricted $M=1$ policy fixed. For each applicant whose treatment assignment changes under this policy, we estimate the counterfactual conditional mean using every admissible reference group on the same cross-fitting folds. We form each reference-specific lower endpoint using the corresponding boundary difference and Lipschitz estimate and intersect it with the known outcome range $[0,1]$. The selected-reference plug-in evaluation uses the endpoint from the prespecified reference rule, whereas the all-reference evaluation takes the pointwise maximum across references. We then aggregate these endpoints over the fixed policy.

The estimated lower value under our main cross-fitted doubly robust procedure is 0.8014. Under the matched plug-in construction, the selected-reference lower envelope gives 0.8008 and the all-reference lower envelope gives 0.8050, a difference of 0.0042, or about 0.42 percentage points. Using the same plug-in construction in both rows isolates the effect of adding the other admissible references. This comparison is descriptive: it holds the learned policy fixed, does not reoptimize the policy under the all-reference criterion, and does not provide inference for the pointwise maximum.

\begin{table}[H]
\centering
\caption{Fixed-policy comparison of selected- and all-reference evaluations}
\begin{tabular}{lcc}
\hline
Evaluation & Lower value & Gain over status quo \\
\hline
Selected reference (main doubly robust) & 0.8014 & 0.0070 \\
Selected reference (plug-in) & 0.8008 & 0.0063 \\
All references (plug-in) & 0.8050 & 0.0106 \\
\hline
\end{tabular}
\label{tab:realapp:allreference}
\end{table}

\end{document}